%% file: QAT.tex
\documentclass[acmsmall,screen,nonacm]{acmart}

\input{includes}

\begin{document}

\title[The Time--Space Complexity of Checking Multiple Assertions in Quantum Programs]{The Time--Space Complexity of Checking\\Multiple Assertions in Quantum Programs}

\author{Shengyuan Yang}
\orcid{0009-0006-3298-7999}
\author{Charles Yuan}
\orcid{0000-0002-4918-4467}
\affiliation{%
  \institution{University of Wisconsin--Madison}
  \country{USA}
}
\email{syang686@wisc.edu, charlesyuan@cs.wisc.edu}

\renewcommand{\shortauthors}{Yang and Yuan.}

\settopmatter{}
\authorsaddresses{}

\begin{abstract}
Runtime assertions are a promising mechanism for testing and debugging quantum programs.
But unlike the classical world, checking a quantum program that contains multiple assertions often requires using additional space or running the program additional times.
For example, on current quantum hardware where mid-circuit measurement is restricted or costly, an assertion's pass/fail outcome cannot be revealed immediately.
Instead, it is routed into an ancilla qubit during execution and read out by a terminal measurement.
For a program with $n$ assertions, a naive strategy uses $n$ ancillas to learn all $n$ outcomes, while an alternative uses one ancilla but repeats program execution over $n$ rounds, checking one assertion per round. Both satisfy $S \cdot T = O(n)$, where $S$ is the number of ancillas and $T$ the number of executions: a fundamental time--space trade-off.

Can one do asymptotically better? We reveal that the answer depends sharply on the information to be learned. Reporting the outcomes of all assertions requires linear complexity, but two partial-information tasks of detecting whether any assertion fails, and of identifying the first failing assertion, require only logarithmic complexity --- an asymptotic improvement. Moreover, the checking strategies for these tasks can trade time for space in useful ways.
In this work, we formalize the complexity of checking multiple assertions in a quantum program.
Using this definition, we establish its landscape of asymptotic lower bounds and constructive upper bounds.
We confirm via a case study on Grover's algorithm that the resource costs of constructed strategies match theoretical predictions, illustrating the practical design space for quantum programmers.
\end{abstract}

\maketitle

\input{main/intro}
\input{main/prelim}
\input{main/formal}

\input{main/shared}
\input{main/exist}
\input{main/first}

\input{main/list}
\input{main/landscape}

\input{main/rela}
\input{main/imply}

\bibliographystyle{ACM-Reference-Format}
\bibliography{ref/ref}

\appendix
\input{appendix/applying}
\input{appendix/stability}
\input{appendix/exist-lb}
\input{appendix/alternative}
\input{appendix/first-lb}

\end{document}

%% file: includes.tex
\usepackage{threeparttable}
\usepackage{quantikz}
\usetikzlibrary{fit,backgrounds}
\usepackage{amsmath}
\usepackage{xspace}
\usepackage{relsize,tipx}
\usepackage{array,xspace,listings}
\usepackage{longtable}
\usepackage{enumitem}
\usepackage{tabularx}
\usepackage{subcaption}
\usepackage{graphicx}
\usepackage{color}
\usepackage{tikz}
\usepackage{colortbl}
\usepackage{multirow,makecell}
\usepackage{dsfont}
\usepackage{proof}
\usepackage{mathpartir}
\usepackage{wrapfig}
\usepackage{setspace}
\usepackage{stmaryrd}
\usepackage[linesnumbered,ruled,vlined]{algorithm2e}
\usepackage{booktabs,longtable}
\usepackage[dvipsnames,table]{xcolor}
\usepackage{placeins}
\usepackage{calc}
\usepackage{mdframed}
\usepackage{needspace}
\usepackage{cleveref}
\usepackage{pifont}

\DeclareMathAlphabet{\pazocal}{OMS}{zplm}{m}{n}

\newcommand{\ketbra}[2]{|{#1}\rangle\langle{#2}|}
\newcommand{\braketU}[3]{\langle{#1}|{#2}|{#3}\rangle}

\newcommand{\strategyhead}[2]{%
  \shortstack[c]{\strut\textsf{#1}\\[-1ex]\strut {\footnotesize Strategy~\ref{#2}}}%
}

\definecolor{mygreen}{rgb}{0.913,0.9535,0.8925}

\newtheorem{theorem}{Theorem}[section]

\AddToHook{env/assumption/begin}{\crefalias{theorem}{assumption}}

\makeatletter
\theoremstyle{acmdefinition}

\undef\definition
\undef\enddefinition
\newtheorem{definition}[theorem]{Definition}

\undef\strategy
\undef\endstrategy
\newtheorem{strategy}[theorem]{Strategy}

\theoremstyle{acmplain}
\makeatother

\makeatletter
\def\@acmplainindent{0pt}
\def\@proofindent{}
\makeatother

%% file: main/intro.tex
\setlength{\textfloatsep}{10pt}
\setlength{\intextsep}{10pt}

\section{Introduction}
\label{sec:intro}
Quantum computation offers the promise of asymptotic advantages for problems such as factoring, search, and simulation by manipulating \emph{qubits} -- or quantum bits -- whose state is a \emph{superposition} of zero and one.
Recent hardware advances~\cite{google2025quantum,bluvstein2023} have made increasingly realistic the prospect of running complex quantum programs needed by practical applications.
A barrier to the development of complex quantum programs, however, is that they are difficult to test and
debug: inspecting an intermediate state of a quantum computation requires performing a \emph{measurement}, a physical operation that in general can perturb the program's state and thereby corrupt its output.

To address this challenge, researchers have proposed \emph{quantum runtime assertions}, schemes to instrument the execution of a quantum program using measurements but minimize any effect on the output.
A growing body of work has studied the formal expressiveness of assertion
predicates~\cite{statisticalAssertions2019,ApproximateAssertion2021,ProjectionBased2020},
circuit constructions that check individual predicates~\cite{RuntimeAssertionCircuts2020,ApproximateAssertion2021,ProjectionBased2020}, tailoring of predicates to program structure~\cite{moveAssertionsAround}, and optimization of checks on hardware~\cite{assertionSlicing}.
But an essential question remains less explored: \emph{how, and how efficiently, can we check multiple assertions in a program?}

\paragraph{Multiple Assertion Checking}
To check a classical program with many assertions, one evaluates each on the instantaneous program state, reads out the pass/fail outcome immediately at negligible cost, and aborts upon failure.
But in a quantum program, reading out an assertion outcome requires measurement. Performing it immediately would mean measuring \emph{mid-circuit} before the program terminates --- incurring significant latency, noise, and reset time in hardware~\cite{mid-circuitLatency2021,mid-circuitNoise2021,mid-circuitNoise2025,reset2025}.

Prior work makes disparate assumptions on this issue.
\textsf{Proq}~\cite{ProjectionBased2020} assumes mid-circuit measurement is broadly available, and uses projective measurements at runtime to abort at the earliest failure.
Other proposals for quantum assertions~\cite{statisticalAssertions2019,RuntimeAssertionCircuts2020,ApproximateAssertion2021,assertionSlicing} target a \emph{terminal-measurement} model that performs measurement only at the end of execution.
This more restrictive model, commonly known as the static quantum circuit model~\cite{ibmStaticAndDynamic}, corresponds to a broad range of current devices where mid-circuit measurement is costly, error-prone, and only sometimes available~\cite{IBM-DynamicCircuits2025,ibmClassicalFeedforward2025}.

\paragraph{Running Example.}
\begin{wrapfigure}[15]{r}{0.3\textwidth}
\vspace*{-1ex}
\includegraphics[width=0.3\textwidth]{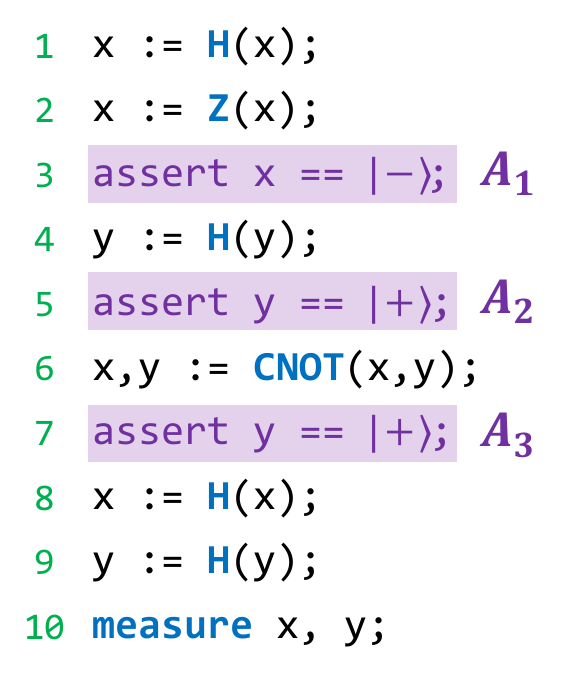}
\vspace*{-5.5ex}
\caption{A quantum program with three assertions. Both qubits \texttt{x} and \texttt{y} are initialized to zero.}
\label{fig:motiv-example}
\end{wrapfigure}

To illustrate checking of multiple assertions in this terminal-measurement model, we present the program in Fig.~\ref{fig:motiv-example}, which performs a series of \emph{quantum logic gates} on two qubits \verb|x| and \verb|y| and then measures both to produce its output.
Assuming the program is bug-free, the final output is $\texttt{x}=1$ and $\texttt{y}=0$.
To aid debugging, the program uses runtime assertions $A_1$, $A_2$, and $A_3$ to check that \texttt{x} and \texttt{y} are in prescribed states at specific points.

One can check these assertions even without the ability to perform mid-circuit measurement.
Instead, one inserts a \emph{checker circuit} that computes an assertion outcome into an \emph{ancilla} (helper) qubit, flipping it to 1 on failure and leaving it as 0 otherwise.
For example, the checker for $A_1$ can be implemented as two Hadamard gates and one CNOT~\cite{RuntimeAssertionCircuts2020,ApproximateAssertion2021}.
The ancilla qubit persists until the end of the program, where it is measured -- together with the output qubits \texttt{x} and \texttt{y} -- to reveal the corresponding assertion outcome.

Even for this small program, multiple checking strategies are possible.
We could (a) use three ancillas, one per assertion, and gather all three outcomes in one program execution; or (b) use one ancilla and run the program three times, with only one assertion enabled per execution.
These two strategies suggest a necessary trade-off between time and space.
For a program with $n$ assertions, these \emph{single-round/$n$-ancillas} and \emph{single-ancilla/$n$-rounds} strategies both satisfy $T \,{\cdot}\, S = n$, where $T$ is the number of program executions ending in measurement and $S$ the number of ancilla qubits used per execution.
This trade-off raises a natural question: \emph{can we do asymptotically better?}

\paragraph{{Assertion-Checking Tasks}}
A key insight is that the answer depends on the desired \emph{task}:
\begin{itemize}[topsep=2pt,itemsep=1pt]
\item \textsc{ListAll}: output the list of all failing assertions, if any exist.
\item \textsc{ExistFail}: output whether there exists a failing assertion.
\item \textsc{FirstFail}: output the index of the earliest failing assertion, if one exists.
\end{itemize}

Prior work implicitly targets \textsc{ListAll}, using either $n$ rounds~\cite{assertionSlicing,statisticalAssertions2019} or $n$ ancillas~\cite{RuntimeAssertionCircuts2020,ApproximateAssertion2021}, the two extremes in our running example.
Even prior work~\cite{ProjectionBased2020} that assumes the availability of mid-circuit measurement uses $n$ measurements in the general case, which exhibits the linear scaling.

\paragraph{Key Results}
We present two surprising results.
First, \textsc{ExistFail} and \textsc{FirstFail} are asymptotically cheaper than \textsc{ListAll}, even if we can measure only at the end.
We present strategies with logarithmic complexity solving these two tasks for arbitrary programs, whereas \textsc{ListAll} requires linear complexity.
\textsc{ExistFail} and \textsc{FirstFail} need not separately preserve every assertion outcome, and indeed, \textsc{FirstFail} coincides with the earliest failure revealed by mid-circuit measurement.
But remarkably, this information can be efficiently recovered by terminal measurement alone.

Second, strategies can tunably trade time for space.
For \textsc{ExistFail} and \textsc{FirstFail}, we prove upper and lower bounds of $S = \Theta(\log (1 + n / T))$ in the \emph{disjoint multi-round} setting where each program execution checks a disjoint set of assertions by terminal measurement.
The disjoint setting subsumes the $T=1,S=n$ and $S=1,T=n$ extremes above and enables interpolation between them.
Interestingly, if we further permit assertions to overlap across executions (the \emph{general multi-round} setting), \textsc{ExistFail}'s cost drops further, whereas \textsc{FirstFail} remains comparatively costly.

Table~\ref{tab:complexity} summarizes the complexity landscape we derive across the three tasks in the terminal-measurement model.
Our bounds are two-sided: in nearly all regimes, we provide constructive and matching upper and worst-case lower bounds that give asymptotically tight characterizations.

\paragraph{Classical Comparison}
Our results for \textsc{ExistFail} and \textsc{FirstFail} reveal where classical intuition fails to transfer to the quantum world:
one requires more cost than expected, the other less.

The former is \textsc{ExistFail}: a classical computer solves it in $O(1)$ space, whereas we prove that the quantum setting requires logarithmic space. The reason is \emph{reversibility}, a fundamental requirement of unitary quantum computation implying that one cannot simply collapse all possible patterns of assertion outcomes into a constant-size flag. Our upper bound uses a reversible counter of failures, while the matching lower bound introduces proof techniques that hinge on reversibility.

The latter is \textsc{FirstFail}.
Classically, we could use a single $O(\log n)$-size register to hold the index of the first failing assertion, updated only when no failure has yet been seen.
This conditional update performs a many-to-one mapping, which is not reversible, and a naive attempt to make this update reversible by writing down additional information blows up the space to $O(n)$.
By contrast, we give a novel construction that solves \textsc{FirstFail} in logarithmic space with no such blowup.

\begin{table}[t]
\centering
\caption{Summary of time--space complexity results. Each row depicts one assertion-checking task, and each column depicts how assertions may be distributed across program executions: the first permits only one execution, the second permits multiple executions but checks each assertion in at most one of them, and the third permits any assertion to be checked on any execution. Each cell gives upper and lower bounds for the number of ancilla qubits $S$ in terms of the numbers of assertions $n$ and program executions $T$.\vspace{-1.5ex}}
\small
\renewcommand{\arraystretch}{0}
\setlength{\tabcolsep}{4pt}
\begin{tabular}{
    >{\centering\arraybackslash}m{1.6cm}
    >{\centering\arraybackslash}m{2.8cm}
    >{\centering\arraybackslash}m{3.8cm}
    >{\centering\arraybackslash}m{4.0cm}
}
\specialrule{0.75pt}{0pt}{6pt}
{Task} & {Single-Round ($T = 1$)} & {Disjoint Multi-Round ($T \ge 1$)} & {General Multi-Round ($T \ge 1$)} \\
\specialrule{0.5pt}{4pt}{4pt}
\makecell{ \specialrule{0pt}{1pt}{1pt} \textsc{ExistFail}}
& \makecell{$S = \Theta(\log n)$\\ {\footnotesize Theorem~\ref{thm:single-round-exist-tightness}}}
& \makecell{$S = \Theta(\log (1 + \frac{n}{T}))$\\ {\footnotesize Theorem~\ref{thm:multi-round-disjoint-exist-complexity}}}
& \makecell{$S = \Theta\big(\frac{\log n}{T}\big) \ \ \text{\footnotesize $\forall T \,{\in}\, O\big(\frac{\log n}{\log \log n}\big)$}$\\ {\footnotesize Theorem~\ref{thm:multi-round-general-exist-loglog-tight}}} \\
\addlinespace[3pt]
\makecell{ \specialrule{0pt}{1pt}{1pt} \textsc{FirstFail}}
& \makecell{$S = \Theta(\log n)$\\ {\footnotesize Theorem~\ref{thm:single-round-first-tightness}}}
& \makecell{$S = \Theta(\log (1 + \frac{n}{T}))$\\ {\footnotesize Theorem~\ref{thm:multi-round-disjoint-first-complexity}}}
& \makecell{$S = \Theta(\log(1 + \frac{n}{T}))$\\ {\footnotesize Theorem~\ref{thm:multi-round-general-first-complexity}}} \\
\addlinespace[4pt]
\makecell{ \specialrule{0pt}{1.5pt}{1.5pt} \textsc{ListAll} }
& \makecell{$S = \Theta(n)$\\ {\footnotesize Corollary~\ref{cor:single-round-list-complexity}}}
& \makecell{$S = \Theta(\frac{n}{T})$\\ {\footnotesize Theorem~\ref{thm:multi-round-list-complexity}}}
& \makecell{$S = \Theta(\frac{n}{T})$\\ {\footnotesize Theorem~\ref{thm:multi-round-list-complexity}}} \\
\addlinespace[3pt]
\specialrule{0.75pt}{0pt}{0pt}
\end{tabular}\label{tab:complexity}
\end{table}

\paragraph{Cost Context}
\vspace{-0.5ex}
In this work, we study the cost of coordinating multiple assertions in an arbitrary program --- a cost orthogonal to checking individual assertions~\citep{statisticalAssertions2019,RuntimeAssertionCircuts2020,ApproximateAssertion2021,ProjectionBased2020} or specializing assertions to program structure~\citep{moveAssertionsAround}.
In the near term, our results enable testing of quantum programs on hardware for which mid-circuit measurement is costly or restricted~\cite{IBM-DynamicCircuits2025,ibmClassicalFeedforward2025}, and where compilers actively seek to minimize it~\cite{ReducingMidCircuitMeasurement2025,ReducingMidCircuitMeasurements2024}.
As mid-circuit measurement becomes more mature, the number of measurements becomes a primary resource to account for alongside the number of ancillas.

Our analysis is future-proofed for this progress.
As we prove, every single-round lower bound we derive for the number of ancillas in the terminal-measurement setting also lower-bounds the number of measurements plus the number of ancilla qubits in the mid-circuit measurement setting.
Moreover, each of our disjoint multi-round strategies is also a constructive strategy to trade ancillas for measurements in quantum assertion schemes~\cite{ProjectionBased2020} that use mid-circuit measurement.

\paragraph{Contributions}
\vspace{-0.5ex}
To summarize, in this work, we present the following contributions:%
{
\clubpenalty=100
\widowpenalty=100
\begin{itemize}[leftmargin=12pt, topsep=2pt, itemsep=2pt]
\item In Sec.~\ref{sec:problem-setting}, we introduce a general framework to analyze the time--space complexity of checking multiple assertions in a quantum program --- a cost that is orthogonal to the checking circuits for individual assertion predicates identified by prior work. We formalize how \emph{strategies} use ancillas and executions ending in measurement to coordinate assertion checking for different \emph{tasks}.

\item In Sec.~\ref{sec:shared}, we introduce a proof technique to tackle complexity lower bounds in the analysis of quantum programs.
The technique establishes lower bounds for predicate-checking tasks over arbitrary programs by relating them to the difficulty of distinguishing different patterns among fixed-length bit strings, which we can then analyze with automata-theoretic tools.

\item In Secs.~\ref{sec:exist}--\ref{sec:listAll}, we establish the complexity results in Table~\ref{tab:complexity} by adapting the technique of Sec.~\ref{sec:shared} to each assertion-checking task.
Our constructive strategies make no assumptions on the structure of programs and assertions, and they realize a tunable time--space trade-off for each task.

\item
In Sec.~\ref{sec:landscape}, we present a case study of assertion checking for Grover's quantum search algorithm, illustrating how the different checking strategies can be instantiated on a concrete program.
We provide a Qiskit implementation that reproduces this study in simulation, confirming that the concrete costs of strategies are consistent with our asymptotic predictions.
\end{itemize}}

Our work extends to the quantum setting a tradition of research that characterizes the complexity of foundational program analysis tasks~\cite{complexityOfPredictingAtomicityViolations,testingMessagePassingComplexity,optimalConsistencyChecking}, here shaped by constraints inherent to quantum computation: reversibility and the destructiveness of measurement.
Moving forward, as quantum applications demand programs that are increasingly complex yet robust, this work expands the programmer's arsenal of tools to test and debug programs, offers rigorous complexity guarantees, and reveals fundamental differences between the quantum and classical worlds.

%% file: main/prelim.tex
\section{Background on Quantum Computation}
\label{sec:background}
In this section, we provide the key concepts in quantum computation that are relevant to this work. For a comprehensive reference, please refer to the textbook of \citet{QCQI}.

\paragraph{Qubits and Quantum States.}
\vspace{-0.5ex}
The state of a single \emph{qubit} is a unit vector in a two-dimensional complex Hilbert space $\mathcal{H}_2$, conventionally written as $\lambda_0 \ket{0} + \lambda_1\ket{1}$ where $\ket{0}$ and $\ket{1}$ are the \emph{computational basis states} and $\lambda_0,\lambda_1 \in \mathbb{C}$ are \emph{amplitudes} satisfying $|\lambda_0|^2 + |\lambda_1|^2 = 1$. The basis states $\ket{0}$ and $\ket{1}$ correspond to classical bit values. When both amplitudes are nonzero, the qubit is in a \emph{superposition}; common examples of superpositions are {$\ket{+} = \tfrac{1}{\sqrt{2}}(\ket{0}+\ket{1})$} and {$\ket{-} = \tfrac{1}{\sqrt{2}}(\ket{0}-\ket{1})$}.

\vspace{-0.25ex}
More generally, an $n$-qubit system has a state space given by the tensor product $\mathcal{H}_{2^n} \triangleq \bigotimes_{i=1}^{n} \mathcal{H}_2$. Its computational basis is $\{\ket{x} | \, x \,{\in}\, \{0,1\}^n\}$, so any $n$-qubit \emph{pure state} can be expressed as a superposition of $n$-bit strings, i.e., $\ket{\psi}=\sum_{x\in \{0,1\}^n}\lambda_x\ket{x}$ with $\sum_x |\lambda_x|^2=1$. For example, \smash{$\tfrac{1}{\sqrt{2}}(\ket{00}+\ket{11})$} is a two-qubit state in superposition of the basis states $\ket{00}$ and $\ket{11}$. As is customary in quantum computation, we use $\ket{xy}$, $\ket{x}\ket{y}$, and $\ket{x,y}$ interchangeably to denote the tensor product $\ket{x}\otimes\ket{y}$.

\paragraph{Quantum Gates and Unitary Operators.}
\vspace{-0.5ex}
A \emph{quantum logic gate} manipulates the bit strings and their amplitudes within a quantum state without collapsing the state from superposition.
The semantics of a quantum gate is a unitary operator $U$, which is linear and norm-preserving. In particular, for any state \smash{$\ket{\psi} \,{=}\, \sum_{x\in \{0,1\}^n}\lambda_x\ket{x}$}, applying $U$ yields the new state \smash{$U\ket{\psi} \,{=}\, \sum_{x\in \{0,1\}^n} \lambda_x ~U\ket{x}$} without collapsing it from superposition.
Notably, any unitary operator is invertible, and $U^{-1}=U^{\dagger}$, with $U^{\dagger}$ its conjugate transpose. Hence, any quantum gate must be \emph{reversible}.
Examples include:
\begin{itemize}[topsep=2pt]
\item $X$ -- the bit-flip (NOT) gate, which maps $\ket{0}$ to $\ket{1}$ and maps $\ket{1}$ to $\ket{0}$.
\item $Z$ -- the phase-flip gate, which maps $\ket{1}$ to $-\ket{1}$ and leaves $\ket{0}$ unchanged.
\item $H$ -- the Hadamard gate, which maps $\ket{0}$ to $\ket{+}$ and maps $\ket{1}$ to $\ket{-}$.
\end{itemize}

\vspace{-0.25ex}
Quantum gates can be \emph{controlled} by other qubits to form larger unitaries. For example, the two-qubit CNOT gate, with the first qubit as control and the second as target, acts as $\ket{0, x} \mapsto \ket{0, x}$ and $\ket{1, x} \mapsto \ket{1, \text{NOT}~x}$ for $x \,{\in}\, \{0,1\}$; equivalently, it can be written as $\ket{y, x} \mapsto \ket{y, x\oplus y}$. More generally, multi-controlled-$X$ gates apply $X$ conditioned on multiple control qubits being $\ket{1}$.

\paragraph{Measurement.}
\vspace{-0.5ex}
A unitary operator evolves a quantum state while preserving its superposition. By contrast, \emph{measurement} extracts classical information from a quantum state and collapses the superposition. For example, measuring a single-qubit state $\lambda_0 \ket{0} + \lambda_1\ket{1}$ in the computational basis yields the classical outcome $0$ with probability $|\lambda_0|^2$ and $1$ with probability $|\lambda_1|^2$. Unless otherwise specified, all measurements in this paper are performed in the computational basis.

\paragraph{Projector.}
\vspace{-0.5ex}
{
\clubpenalty=0
\widowpenalty=0
Let $X$ be a closed subspace of $\mathcal{H}$. The projector onto $X$, denoted $P_X$, is a linear operator on $\mathcal{H}$ satisfying $P_X^2 \,{=}\, P_X$ and $P_X^{\dagger} \,{=}\, P_X$, with image $\mathrm{Im}(P_X)=X$. For any $|\psi\rangle\in\mathcal{H}$, we can decompose $|\psi\rangle=|\psi_X\rangle+|\psi_{\perp}\rangle$ where $|\psi_X\rangle\in X$ and $|\psi_{\perp}\rangle\in X^{\perp}$. Applying the projector thus yields $P_X|\psi\rangle=|\psi_X\rangle$; in particular, $|\psi\rangle$ lies in $X$ if and only if $P_X|\psi\rangle=|\psi\rangle$.
}

Projectors can be written in terms of outer products.
We write $\langle\psi|$ for the conjugate transpose of $|\psi\rangle$, so that $\langle\psi|\phi\rangle$ denotes the inner product of two states, and the \emph{outer product} $|\psi\rangle\langle\phi|$ denotes the linear operator that maps any state $|\xi\rangle$ to $\langle\phi|\xi\rangle\,|\psi\rangle$.
For a unit vector $|\psi\rangle$, the outer product $|\psi\rangle\langle\psi|$ is exactly the projector onto the one-dimensional subspace spanned by $|\psi\rangle$; more generally, the projector onto a subspace with orthonormal basis $\{|\psi_j\rangle\}_j$ is $\sum_j |\psi_j\rangle\langle\psi_j|$.

%% file: main/formal.tex
\section{Complexity Framework for Quantum Assertion Checking}
\label{sec:problem-setting}

In this section, we formalize a framework to analyze the time--space complexity of checking multiple runtime assertions in a quantum program.
First, we define (Sec.~\ref{sec:object}) quantum programs that use multiple assertions.
We then define tasks (Sec.~\ref{sec:tasks}) and strategies (Sec.~\ref{sec:strategy}) for assertion checking in the terminal-measurement model, together with their cost metrics.
We relate (Sec.~\ref{sec:midcircuit-bridge}) this model to mid-circuit measurements.
Finally, we lift (Sec.~\ref{sec:prob-tasks}) the definitions to the probabilistic setting.

\input{main/formal-object}
\input{main/formal-task}

\input{main/formal-strategy}
\input{main/formal-mid}
\input{main/formal-prob}

%% file: main/formal-object.tex
\subsection{Quantum Programs with Multiple Assertions}
\label{sec:object}

In general, a quantum program that contains multiple assertions is a sequence of unitary operations, interleaved with $n$ assertions.
Each assertion tests whether specific program qubits are in a prescribed state, and the program ends with a terminal measurement of the output state.

\begin{definition}[Program with Assertions]\label{def:prog}
We define a program with assertions as a tuple
\[
\pazocal{Q} \ = \ \big(\textsf{prog};\,|\varphi_0\rangle;\,U_0,\, A_1,\, U_1,\, A_2,\, \dots,\, A_n,\, U_n;\ \pazocal{M}\big),\quad \text{where}
\]
\begin{itemize}[topsep=1pt,itemsep=2pt,leftmargin=20pt]

\item[$-$] $\textsf{prog}$ is a multi-qubit quantum register that takes on the initial state $|\varphi_0\rangle$.

\item[$-$] Each of $U_0, U_1, \ldots, U_n$ is an arbitrary quantum logic gate on the qubits of \textsf{prog}; each may itself be a sequence of logic gates, since any sequence of unitary operations composes into a single unitary operation, and may also be empty (i.e., when two assertions are adjacent).

\item[$-$] $\pazocal{M}$ is a terminal measurement of (a subset of) \textsf{prog}, which produces the program output.

\item[$-$] Each $A_i$ is an assertion $\textsf{assert}(\bar q_i; P_i)$, whose predicate is given in projection-based form~\cite{ProjectionBased2020} by an ordered list $\bar q_i$ of checked program qubits and a projector $P_i$: the predicate is satisfied when the joint state $|\phi\rangle$ of the qubits in $\bar q_i$ lies entirely in $\mathrm{Im}(P_i)$, i.e., when $P_i|\phi\rangle \,{=}\, |\phi\rangle$.
For instance, $A_1$ and $A_2$ in Fig.~\ref{fig:motiv-example} can be expressed as $\textsf{assert}(x; |{-}\rangle\langle{-}|)$ and $\textsf{assert}(y; |{+}\rangle\langle{+}|)$.
This form uniformly expresses the predicates of prior assertion schemes~\cite{statisticalAssertions2019,RuntimeAssertionCircuts2020,ProjectionBased2020,ApproximateAssertion2021}.
The semantics of checking $A_i$ is captured by a corresponding \emph{checker unitary} $C_i$, which we formalize shortly.
\end{itemize}

\vspace{0.25ex}

We call the \emph{bare execution} of $\pazocal{Q}$ the execution obtained by removing all assertions. The bare execution applies the unitaries  $U_\iota$ in order and then performs $\pazocal{M}$.
We write $|\varphi_i\rangle \triangleq U_{i-1}\cdots U_0|\varphi_0\rangle$ for the state that is to be checked by the $i$-th assertion $A_i$, i.e., the bare-execution state after $U_{i-1}$.
\end{definition}

This definition reflects a computational model in which the program evolves unitarily and all measurements occur at the end, i.e., a \emph{static} circuit, in the terminology of \citet{ibmStaticAndDynamic}, as opposed to a dynamic one with mid-circuit measurement. How the program is written -- in any particular syntax, or given in any gate set -- is immaterial, provided it can be unrolled into a finite unitary sequence ending with measurement, interleaved with the assertions to be checked.

\paragraph{Checker Unitary.}
To capture the semantics of $\textsf{assert}(\bar q_i; P_i)$, we introduce the abstraction of a \emph{checker unitary} $C_i$.
Conceptually, $C_i$ records the runtime outcome of $A_i$ into an ancilla, flipping it exactly when the checked state lies in the failing subspace. Formally, acting on $\mathcal{H}_{\bar q_i} \otimes \mathcal{H}_2$, where $\mathcal{H}_{\bar q_i}$ is the joint state space of the checked qubits $\bar q_i$ and $\mathcal{H}_2$ that of the ancilla,
\begin{equation}
\label{eq:checker}
C_i ~\triangleq~ P_i \otimes I + (I - P_i) \otimes X,
\vspace{-0.25ex}
\end{equation}
such that for any state\footnote{For simplicity, as a matter of notation throughout the paper, an operator defined on a subset of qubits is implicitly extended by the identity when applied to a larger register. For example, applied to the full program register, $P_i$ acts as $P_i \otimes I_{\scriptscriptstyle \textsf{prog} \setminus \bar q_i}$. In particular, by linearity, the action~\eqref{eq:checker-action} holds verbatim when $\ket{\phi}$ is replaced by any state of the full program register, including states entangled across \smash{$\bar q_i$} and the remaining qubits.} $\ket{\phi}$ on the qubits $\bar q_i$ and any ancilla basis state $\ket{c}$,
\begin{equation}
\label{eq:checker-action}
C_i\big(\ket{\phi}_{\bar q_i} \otimes \ket{c}\big) ~ = ~\underbrace{\smash{P_i\ket{\phi}_{\bar q_i} \otimes \ket{c}}}_{\text{passing branch}} ~ + ~ \underbrace{\smash{(I - P_i)\ket{\phi}_{\bar q_i} \otimes \ket{c \oplus 1}}}_{\text{failing branch}}.
\end{equation}

This checker unitary thus routes the outcome into the ancilla \emph{coherently}, without measurement: the $\ket{c}$ branch carries the component of the state $\ket{\phi}$ that passes the assertion, and the $\ket{c \oplus 1}$ branch carries the component that fails.
A key property shared by non-destructive assertion schemes~\cite{ProjectionBased2020,RuntimeAssertionCircuts2020,ApproximateAssertion2021}, and captured by our checker unitary abstraction, is that checking an assertion leaves the program state unchanged whenever the assertion predicate is satisfied. If $\ket{\phi} {\in}\, \textrm{Im}(P_i)$, or in other words $P_i\ket{\phi} {=}\, \ket{\phi}$, then the failing branch vanishes and the passing branch leaves $\ket{\phi}_{\bar q_i} \otimes \ket{c}$.

Prior work offers constructions to realize $C_i$ as a quantum circuit for many classes of predicates $P_i$.
In App.~\ref{app:A}, we show that existing circuit constructions for quantum assertion predicates~\cite{ProjectionBased2020,RuntimeAssertionCircuts2020,ApproximateAssertion2021} are either already of the form in~\eqref{eq:checker}, or can be brought to it with minor modification.

By taking the abstraction of a checker unitary as primitive, we decouple the process and cost of realizing each individual assertion, which is the focus of prior research, from the \emph{strategy} that orthogonally governs how checkers are deployed: how many ancillas are used, how the $C_i$ are distributed and coordinated across rounds, and how to compress or aggregate recorded outcomes before the terminal readout.
All strategies and bounds derived in the rest of the paper hold regardless of the choices of assertion predicates $P_i$, checker circuits $C_i$, and program instructions $U_i$.

%% file: main/formal-task.tex
\subsection{Assertion Checking Tasks}
\label{sec:tasks}
Using the above, we formalize three fundamental assertion checking tasks that gather different levels of information over assertion failure patterns, i.e., length-$n$ bit vectors recording the definite pass/fail outcome of each assertion.
In Sec.~\ref{sec:prob-tasks}, we lift these definitions to the probabilistic case.

\begin{definition}[Deterministic Failure Pattern]
\label{def:classical}
A program with assertions $\pazocal{Q}$ has a \emph{deterministic failure pattern} $F = (F_1, \ldots, F_n) \in \{0,1\}^n$ if, for every $i \in [n]$ where $[n]$ abbreviates $\{1,\ldots,n\}$, the bare execution state $\ket{\varphi_i}$ (Def.~\ref{def:prog}) checked by the $i$-th program assertion lies entirely in the passing subspace
$\mathrm{Im}(P_i)$ when $F_i = 0$, or in its orthogonal complement
$\mathrm{Im}(P_i)^{\perp}$ when $F_i = 1$.
\end{definition}

\vspace{-0.25ex}
Given a program-with-assertions having deterministic failure pattern $F$, the checking tasks are:
\begin{itemize}
\item \textsc{ListAll}: output the set $\{\,i\in[n]\mid F_i=1\,\}$, or equivalently, the full failure pattern $F$.
\item \textsc{ExistFail}: decide whether $\exists i\in[n]$ such that $F_i=1$.
\item \textsc{FirstFail}: output $\min\{\,i\in[n]\mid F_i=1\,\}$, or $\bot$ if no such index exists.
\end{itemize}

We then say that a strategy \emph{solves} one of the tasks above if it answers correctly, with certainty, given any program with assertions that has a deterministic failure pattern.

%% file: main/formal-strategy.tex
\subsection{Assertion Checking Strategies}
\label{sec:strategy}

Next, we define the concept of an \emph{assertion-checking strategy}, which consists of one or more \emph{rounds}, each one being an instrumented execution of the program under study. In each round, the strategy may allocate ancillas, insert checker unitaries to route assertion outcomes to ancillas, and apply ancilla-processing logic to aggregate information.
Finally, it applies a decoder to the terminal measurement outcomes on the ancillas collected across the rounds, yielding an answer to the target task (Sec.~\ref{sec:tasks}).
This conceptualization of strategies is intended to capture the debugging workflows available to a programmer, and to apply to any instance of a program with assertions.

\begin{definition}[Instrumentation]\label{def:instr}
An \emph{instrumentation} $\pazocal{I}$ is a program transformation that specifies:
\begin{enumerate}[topsep=1pt,itemsep=1pt,leftmargin=20pt]
\item[(1)] the size of an ancilla register $\textsf{anc}$, denoted as $m$, with $\textsf{anc}$ initialized as $|0^m\rangle$;
\item[(2)] a set $\mathcal{E}\subseteq [n]$ of enabled assertion indices; assertions outside $\mathcal{E}$ will be skipped;
\item[(3)] for each assertion $i \,{\in}\, \mathcal{E}$, an \emph{assertion-handling block} $B_i = G_\ell \cdots G_1$ (for some $\ell \ge 0$); that is, a sequence of unitary operations in which each $G_j$ either (i) applies the checker unitary $C_i$ of $A_i$ (cf.\ \eqref{eq:checker}) with a designated ancilla qubit $a \,{\in}\, \textsf{anc}$, denoted $C_{i\to a}$ for short, or (ii) applies an arbitrary unitary $V$ on $\textsf{anc}$ alone to coordinate or process the recorded assertion outcomes.
\end{enumerate}
The transformed program $\pazocal{I}(\pazocal{Q})$ is obtained from $\pazocal{Q}$ by replacing, within its interleaved sequence ($U_0, A_1, U_1, \ldots, A_n, U_n$) each assertion $A_i$ with $B_i$ if $i \,{\in}\, \mathcal{E}$ and with the identity if $i \,{\notin}\, \mathcal{E}$.
$\pazocal{I}(\pazocal{Q})$ is an executable sequence of unitary operations, and references the input program $\pazocal{Q}$ only through the count and positions of its assertions.
\end{definition}

\begin{definition}[Single-Round Strategy]\label{def:stra}
\vspace{-0.25ex}
A \emph{single-round strategy} consists of one instrumentation $\pazocal{I}$ with a classical decoder \textsf{Dec}. On any input program-with-assertions $\pazocal{Q}$, the strategy executes $\pazocal{I}(\pazocal{Q})$ from the initial state $|\varphi_0\rangle_{\textsf{prog}}\otimes |0^{m}\rangle_{\textsf{anc}}$, and measures (a subset of) the ancilla qubits at termination. \textsf{Dec} then maps the measurement outcome, a bit string $R$, to an answer for the target task.
\end{definition}

\paragraph{Checking Contract}
The definitions above imply usual semantic separations between programming and testing.
First, a strategy is \emph{uniform}: it is fixed before the program under test is provided, referencing it only through the assertion count and positions, and the decoder reads the ancilla outcomes alone. A strategy does not inspect the program text or the assertion predicates, nor exploit the program output, whose correlation with the assertion failures is specific to each program. A procedure that answers by inspecting the source would not be performing runtime checking.

Second, a strategy \emph{observes without rewriting}: information flows from \textsf{prog} to the ancillas only through the checker unitaries of the declared assertions $A_i$, and all further processing acts on the state of \textsf{anc} alone.
This separation is required to ensure that the checking of assertions cannot cheat by moving quantum information out of the program or corrupt the program state.
In the following lemma, we prove this property holds for any instrumentation under our definition.

\begin{lemma}[Program-State Invariance]
\label{lem:invariance}
For any instrumentation $\pazocal{I}$ and any program with assertions $\pazocal{Q}$ that has a deterministic failure pattern, the execution of $\pazocal{I}(\pazocal{Q})$ evolves the program register in the same way as the bare execution of $\pazocal{Q}$ and leaves it unentangled from the ancillas.
\end{lemma}

\begin{proof}[Proof]
\vspace{-0.25ex}
Every step of $\pazocal{I}(\pazocal{Q})$ preserves the product form of the initial state. Program segments $U$ act on $\textsf{prog}$ alone. Within each block $B_i$, each operation either applies an ancilla-processing unitary $V$ on $\textsf{anc}$ alone, or is a checker unitary call $C_{i\to a}$ that acts on a state lying entirely in one subspace of its predicate, so one branch of~\eqref{eq:checker-action} vanishes and the call reduces to an ancilla-only operation, leaving the program register intact. Induction over all steps completes the argument.
\end{proof}

\vspace{-0.25ex}
Strategies can vary in their number of rounds:

\vspace{-0.25ex}
\begin{definition}[Multi-Round Strategy]\label{def:multi-stra}
A \emph{multi-round strategy} consists of a sequence of instrumentations $\pazocal{I}_1, \ldots, \pazocal{I}_T$ with a classical decoder \textsf{Dec}. 
On any input $\pazocal{Q}$, it executes $\pazocal{I}_t(\pazocal{Q})$ independently for $T$ rounds,\footnote{By two rounds, we refer to two different instrumented executions of the program that enable different assertions or perform different assertion processing logic, rather than two repetitions of the same execution solely for confidence amplification of probabilistic outcomes. If desired, any fixed single-round strategy can still be repeated for sampling-based amplification.} re-initializing \textsf{prog} and \textsf{anc} and measuring (a subset of) the ancilla qubits at termination in each round; \textsf{Dec} then maps the outcomes $(R_1, \ldots, R_T)$ to an answer for the target task.
\end{definition}

We further classify multi-round strategies using the additional criterion of \emph{disjointness}.
Let $\mathcal{E}_t$ be the enabled assertion indices used by $\pazocal{I}_t$.
The strategy is \emph{disjoint} if $\mathcal{E}_t \cap \mathcal{E}_{t'}=\emptyset$ for all $t\neq t'$.
We introduce the criterion of disjointness to separate the two regimes: permitting multiple rounds to re-check the same assertion under different ancilla-processing logic effectively provides multiple encodings of the same failure pattern, which can yield better time--space complexity for some tasks.

\paragraph{Cost Metrics}
In this work, we study the complexity of assertion checking along two dimensions:
\renewcommand{\labelitemi}{${-}$}
\begin{itemize}[itemsep=1pt, topsep=1pt, leftmargin=20pt, labelsep=3pt]
    \item \emph{Time (rounds)}: $T(\pazocal{S}) \triangleq \text{the number of rounds, i.e., the length of its instrumentation sequence}$.
    \item \emph{Space (ancillas)}: $S(\pazocal{S}) \triangleq \max_t \{\text{size of the ancilla register } \textsf{anc} \text{ in } \pazocal{I}_t \}$.
\end{itemize}

\vspace{0.25ex}
For all proposed explicit strategies, we further report their operational costs as follows:
\begin{itemize}[itemsep=1pt, topsep=1pt, leftmargin=20pt, labelsep=3pt]
    \item \emph{Measurements}: $M(\pazocal{S}) \triangleq$ the total number of single-qubit measurements performed on \textsf{anc} at the end of each round, summed over all $\pazocal{I}_t$.
    \item \emph{Checker Calls}: $C(\pazocal{S}) \triangleq$ the number of checker unitaries $C_i$ applied per assertion $A_i$.
    \item \emph{Additional Gates}: $G(\pazocal{S}) \triangleq$ the total number of additional gates used over all $\pazocal{I}_t$, excluding the original program $\pazocal{Q}$ itself and the internals of the checker unitaries $C_i$.
\end{itemize}

The additional gate cost $G$ is introduced by the assertion-handling blocks $B_i$. For sake of concreteness, we count $G$ with respect to the elementary gate set $\{{X}, \textrm{C}^k\textrm{NOT}\mid k\ge 1\}$.

\begin{strategy}[single-round/$n$-ancillas]
\label{str:one-round-n-anc}
Our first strategy from the introduction (Sec~\ref{sec:intro}) uses one instrumentation $\pazocal{I}_1$ with $\mathsf{anc}$ register size $n$ and $\mathcal{E}_1\,{=}\,[n]$. The assertion-handling block for the $i$-th assertion is simply its checker unitary, writing the outcome into the $i$-th qubit of $\mathsf{anc}$. The terminal measurement of $\mathsf{anc}$ reveals all assertion outcomes. Thus $T\,{=}\,1$, $S\,{=}\,n$, $M\,{=}\,n$, $C\,{=}\,1$, and $G\,{=}\,0$.
\end{strategy}

\begin{strategy}[single-ancilla/$n$-rounds]
\label{str:n-round-one-anc}
The other strategy uses instrumentations $\pazocal{I}_1,\ldots,\pazocal{I}_n$, where $\pazocal{I}_k$ sets $\mathcal{E}_k\,{=}\,\{k\}$ and allocates one ancilla. In the $k$-th round, the assertion-handling block for the $k$-th assertion calls its checker unitary, writing the outcome into the sole ancilla, and the terminal measurement reveals that assertion's outcome. Thus $T\,{=}\,n$, $S\,{=}\,1$, $M\,{=}\,n$, $C\,{=}\,1$, and $G\,{=}\,0$.
\end{strategy}

%% file: main/formal-mid.tex
\subsection{Relationship to Mid-Circuit Measurement Model}
\label{sec:midcircuit-bridge}

Though our definition of assertion checking strategies is given in a setting permitting only terminal measurement, it is directly linked to the setting where mid-circuit measurement is available.

\begin{definition}
\label{def:mid-str}
A \emph{mid-circuit strategy} is a single-round strategy (Def.~\ref{def:stra}) whose assertion-handling blocks may additionally perform \emph{measure-and-reset} steps: a designated ancilla qubit is measured in the computational basis, the outcome is recorded, and the same qubit is reset to $\ket{0}$ for reuse.\footnote{Coupling re-initialization to measurement reflects hardware practice, where a reset is realized by at least a measurement.}
The decoder reads all recorded mid-circuit outcomes together with the terminal readout, so the cost ${M}$ counts every single-qubit measurement, whether mid-circuit or terminal.
\end{definition}

In this setting, intermediate readout allows assertions to be checked sequentially as the program executes forward.
Thus, a single forward execution with immediate readout, rather than repeated rounds, is the natural unit of analysis (i.e., $T$ is fixed to 1).
The number of logical measurements $M$ becomes a relevant architectural cost as a first-class parameter alongside the ancilla count $S$.

The key connection between the two settings is that lower and upper bounds on the complexity of the terminal-measurement model transfer to the mid-circuit measurement model:

\begin{theorem}
\label{thm:mc-deferral}
Every mid-circuit strategy $\pazocal{S}_{\mathrm{mid}}$ (Def.~\ref{def:mid-str}) using ${S}$ ancillas and ${M}$ measurements that solves an assertion checking task induces a single-round terminal-measurement strategy for the same task using at most ${S} + {M}$ ancillas. Consequently, any single-round space lower bound for all terminal-measurement strategies is also a lower bound on ${S} + {M}$ for all mid-circuit strategies.
\end{theorem}

\begin{proof}
Given $\pazocal{S}_{\mathrm{mid}}$, a single-round terminal-measurement strategy $\pazocal{S}_{\mathrm{end}}$ replaces each measure-and-reset of a qubit $a$ with a $\mathrm{SWAP}$ onto a fresh $\ket{0}$ ancilla.
Since nothing later acts on this ancilla, its measurement may be safely deferred to the end of the circuit~\cite{QCQI} with the outcome distribution unchanged. There are at most $M$ new ancillas, meaning that $\pazocal{S}_{\mathrm{end}}$ uses at most $S + M$ ancillas.
\end{proof}

We next relate the upper bounds, subject to a mild structural condition that holds for all strategies proposed in this paper: the sets of enabled assertions $\mathcal{E}_1, \dots, \mathcal{E}_T$ are each contiguous in indices.

\begin{theorem}
\label{thm:mc-recast}
Every contiguous disjoint $T$-round terminal-measurement strategy $\pazocal{S}_{\mathrm{end}}$ that uses $S$ ancillas induces a mid-circuit strategy for the same task using $S$ ancillas and ${M} \le T \cdot S$ measurements.
\end{theorem}

\begin{proof}[Proof]
\vspace{-0.5ex}
Assuming contiguity as above, we may re-index the rounds so that the indices across $\mathcal{E}_1, \dots, \mathcal{E}_T$ of $\pazocal{S}_{\mathrm{end}}$ appear in increasing order. We then construct a mid-circuit strategy $\pazocal{S}_{\mathrm{mid}}$ on an $S$-qubit ancilla register as follows. Each assertion $A_i$ retains the handling block $B_i$ it has in the round of $\pazocal{S}_{\mathrm{end}}$ that enables $A_i$. The single forward execution of $\pazocal{S}_{\mathrm{mid}}$ is the concatenation of $T$ program segments, with the $t$-th checking the assertions in $\mathcal{E}_t$. At each segment boundary, the whole ancilla register is measured and reset. Because checking within a segment never disturbs the program register by Lem.~\ref{lem:invariance}, the ancilla state at the end of segment $t$ equals the final ancilla state of round $t$ in $\pazocal{S}_{\mathrm{end}}$. At most $S$ measurements occur per segment, hence $M \le T \cdot S$.
\end{proof}

\vspace{-0.5ex}
One implication of the above theorems is that assertion schemes built on mid-circuit projective measurement~\cite{ProjectionBased2020} can reduce their measurement count, a potential benefit on hardware where mid-circuit measurement is available but still costly.
Rather than performing $\Theta(n)$ measurements to check $n$ assertions one by one, the strategies summarized in Table~\ref{tab:complexity} and formalized in the following sections provide the ability to aggregate outcomes within segments, measuring only at segment boundaries.
For \textsc{ExistFail} and \textsc{FirstFail}, the total complexity drops to ${S} + {M} = \Theta(\log n)$.

%% file: main/formal-prob.tex
\subsection{Semantic Stability of the Checking Tasks}
\label{sec:prob-tasks}

In this section, we extend the three checking tasks to the general probabilistic setting.
We show that in this setting, \textsc{ExistFail} and \textsc{FirstFail} have well-defined answers independently of the strategy used to check them --- they are \emph{semantically stable}.
By contrast, \textsc{ListAll} is not, a distinction with consequences for how a programmer may distribute the checking of assertions across rounds.

\paragraph{Probabilistic Case.}
In general, an assertion may fail only with some probability, i.e., the checked program state lies neither entirely inside nor entirely outside the asserted subspace $\mathrm{Im}(P_i)$.
The failure probability of $A_i$ is by definition
$p_i \triangleq \Vert(I-P_i)|\varphi_i\rangle \Vert^2 = \langle \varphi_i |(I-P_i)|\varphi_i\rangle$.
We then impose a standard gap promise~\cite{quantumPropertyTesting, quantumStateCertification}: passing an assertion means meeting its predicate exactly, and a bug manifests as a detectable violation. Formally, for some fixed $\eta \,{>}\, 0$, the promise assumes that $p_i \,{\in}\, \{0\}\cup[\eta,1]$ for all $i$. The probabilistic versions of the tasks are then:
\begin{itemize}
\item $\textsc{ListAll}_\eta(\pazocal{Q}) \triangleq \{\,i\in[n]\mid p_i  \ge \eta\,\}$;
\item $\textsc{ExistFail}_\eta(\pazocal{Q}) \triangleq$ decide whether $\exists i\in[n]$ such that $p_i  \ge \eta$;
\item $\textsc{FirstFail}_\eta(\pazocal{Q}) \triangleq \min\{\,i\in[n]\mid p_i \ge \eta\,\}$ if such $i$ exists, and $\bot$ otherwise.
\end{itemize}

For brevity, our main bounds in the rest of this paper (see Table~\ref{tab:complexity}) are stated for the deterministic case. The lower bounds carry over to the probabilistic case, as deterministic failure patterns are exactly the special case $p_i \,{\in}\, \{0,1\}$. On the other side, a chosen checking strategy can be repeated to amplify confidence: the strategy is executed repeatedly and its outcomes are aggregated classically, leaving the per-round space $S$ unchanged while increasing the number of samples.

\paragraph{Semantic Stability}
In the deterministic setting, Lem.~\ref{lem:invariance} shows that checking leaves the program register unchanged.
In the probabilistic setting, the same does not hold, and the destructiveness of quantum measurement means that checking one assertion can change the state subsequently seen by later ones. Thus, perhaps unintuitively, the marginal failure probability of a later assertion may depend on whether earlier assertions are even checked.
Such inter-assertion disturbance is a known phenomenon~\cite{assertionSlicing,sequentialProjectiveM}. The structural conditions under which this disturbance exists, in terms of the assertion predicates and the program, are an interesting question in its own right.

In App.~\ref{app:robust-proof} (Thm.~\ref{thm:instable}), we prove that \textsc{ListAll}$_\eta$ is the very task whose answer can be affected by this instability. Conceptually, the act of checking can itself change what there is to report, and thus the recovered failure set may depend on the particular strategy used to check it.
This gives a programmer targeting \textsc{ListAll}$_\eta$ a clear decision rule.
When the assertions are known not to disturb one another, several may be checked per round.
Absent such knowledge, the reference \textsc{ListAll}$_\eta$ answer is guaranteed by running one instrumented execution per assertion, with only that assertion enabled via its checker unitary; this is exactly the single-ancilla/$n$-rounds strategy. Redistributing assertions can produce an incorrect answer, as witnessed by Example~\ref{exp:instability} in App.~\ref{app:robust-proof}.

By contrast, in Thm.~\ref{thm:stable} we prove that \textsc{ExistFail}$_\eta$ and \textsc{FirstFail}$_\eta$ admit strategy-independent task semantics --- a programmer may freely distribute assertions across rounds.
Conceptually, any disturbance to later assertions can arise only after some earlier assertion has failure probability at least $\eta$. Once this happens, the output of \textsc{ExistFail}$_\eta$ is already true, and the output of \textsc{FirstFail}$_\eta$ is already determined by this prefix.
Together with the instability of \textsc{ListAll}$_{\eta}$, our results give a full formal understanding of semantic stability for all three quantum assertion checking tasks.

%% file: main/shared.tex
\section{Shared Proof Ingredients: The Lower-Bound Transfer}
\label{sec:shared}
In this section, we formalize a proof technique that underpins the lower-bound arguments in the following sections.
Specifically, we show that every assertion-checking strategy corresponds to a mathematical object, which we call a \emph{finite-dimensional unitary transition system}, that solves a problem of distinguishing patterns among bit-strings.
A lower bound on the size of this transition system directly transfers to a lower bound on the space complexity of the checking strategy.

\begin{definition}[Finite-Dimensional Unitary Transition System]
\label{def:unitary-transition-system}
A one-round finite-dimensional unitary transition system for input bit-strings $x$ of length $n$ is a 4-tuple
{\setlength{\abovedisplayskip}{3pt}
\setlength{\belowdisplayskip}{2pt}
\setlength{\abovedisplayshortskip}{0pt}
\setlength{\belowdisplayshortskip}{0pt}
\begin{align*}
\pazocal{U}_n =
(\pazocal{H}, ~\{\pazocal{T}^{\scriptscriptstyle (i)}_0\}_{i=1}^n, ~\{\pazocal{T}^{\scriptscriptstyle (i)}_1\}_{i=1}^n, ~\ket{\psi_0}),
\end{align*}}

\noindent
where $\pazocal{H}$ is a finite-dimensional Hilbert space, each $\pazocal{T}^{\scriptscriptstyle (i)}_b{:}\, \pazocal{H}\,{\to}\,\pazocal{H}$ is a unitary for $i\,{\in}\,[n]$ and $b\,{\in}\,\{0,1\}$, and $\ket{\psi_0}\,{\in}\,\pazocal{H}$ is the initial state.
Given an input string $x=b_1\cdots b_n\in\{0,1\}^n$, the final state is
{
\setlength{\abovedisplayskip}{3pt}
\setlength{\belowdisplayskip}{2pt}
\setlength{\abovedisplayshortskip}{0pt}
\setlength{\belowdisplayshortskip}{0pt}
\begin{align*}
|\psi(x)\rangle
~\triangleq~
(\pazocal{T}^{\scriptscriptstyle (n)}_{b_n}\cdots \pazocal{T}^{\scriptscriptstyle (2)}_{b_2}\pazocal{T}^{\scriptscriptstyle (1)}_{b_1})\ket{\psi_0}.
\end{align*}}

A \emph{$T$-round unitary transition system} is a tuple $(\pazocal{U}^{1}_n, \ldots, \pazocal{U}^{T}_n)$ of unitary transition systems that share the same input length $n$. Given input $x$, it produces the transcript $(r_1, \ldots, r_T)$ by measuring each final state $|\psi^{t}(x)\rangle$ of $\pazocal{U}^{t}_n$ independently in the computational basis.
\end{definition}

\begin{definition}[Bit-String Pattern Distinguishing Tasks]
\label{def:distinguish}
Given an input string $x = b_1 \cdots b_n \in \{0,1\}^n$, we define three tasks $\textsc{List}$, $\textsc{Exist}$, and $\textsc{First}$ that are identical to the three tasks in Sec.~\ref{sec:tasks} respectively with the input string $x$ replacing the failure pattern $F$.
A $T$-round unitary transition system \emph{solves} each task if there exists a classical decoder $\textsf{Dec}$ such that, on every input $x \,{\in}\, \{0,1\}^n$, applying $\textsf{Dec}$ to the transcript $(r_1, \ldots, r_T)$ outputs the correct answer to the task with certainty.
\end{definition}

We relate lower bounds for assertion checking strategies to those for transition systems via two steps. First, we show that an instrumentation fixes a family of ancilla-only unitaries (Lem.~\ref{lem:ancilla-evolution}). Second, these unitaries assemble into a unitary transition system that solves the
corresponding distinguishing task, with the failure pattern playing the role of the input string (Lem.~\ref{lem:reduction}).

\begin{lemma}[Induced Ancilla Evolution]
\label{lem:ancilla-evolution}
Let $\pazocal{I}$ be an instrumentation with an ancilla register of size $m$.
For each $i \,{\in}\, [n]$, there exists a pair of ancilla-only unitaries $W^{i}_0, W^{i}_1$, determined by $\pazocal{I}$ alone, such that for every $\pazocal{Q}$ with deterministic failure pattern $F$, the state of\,\ $\pazocal{I}(\pazocal{Q})$ immediately after the $i$-th assertion position is equal to $|\varphi_i\rangle_{\textsf{prog}} \otimes (W^{i}_{F_i} \cdots W^{1}_{F_1})|0^{m}\rangle$.
\end{lemma}

\vspace{-2ex}
\begin{proof}
If $i \,{\notin}\, \mathcal{E}$, set $W^{i}_0 = W^{i}_1 \triangleq I$. If $i \,{\in}\, \mathcal{E}$, write the assertion-handling block as $B_i = G_\ell \cdots G_1$ (Def.~\ref{def:instr}) and set \smash{$W^{i}_b \triangleq g_\ell \cdots g_1$} for $b \,{\in}\, \{0,1\}$, where \smash{$g_j \,{\triangleq}\, V$} if $G_j$ applies an ancilla-processing unitary $V$, and if \smash{$G_j = C_{i \to a}$}, then \smash{$g_j \,{\triangleq}\, I$} in $W^{i}_0$ and \smash{$g_j \,{\triangleq}\, X_a$} in $W^{i}_1$, with $X_a$ the $X$ gate on the designated ancilla $a$. Each pair is fixed by $\pazocal{I}$ alone; the failure pattern enters only by selecting which of the two describes the block's run-time action.

By Lem.~\ref{lem:invariance}, the joint state of $\pazocal{I}(\pazocal{Q})$ remains a product whose program factor follows the bare-execution trajectory, so it suffices to show that the $i$-th assertion position applies $W^{i}_{F_i}$ to the ancilla factor. This is immediate for $i \,{\notin}\, \mathcal{E}$.
For $i \,{\in}\, \mathcal{E}$, every checker call of the block meets the program factor $\ket{\varphi_i}$, which the deterministic failure pattern places entirely in one subspace of~\eqref{eq:checker-action}.
Only one branch survives: on $F_i \,{=}\, 0$ the passing branch, fixing the ancilla, and on $F_i \,{=}\, 1$ the failing branch, flipping it.
Extended by linearity over the components of the designated ancilla, the call thus acts as $I_{\textsf{prog}} \otimes (X_a)^{F_i}$, while ancilla-processing gates act as $I_{\textsf{prog}} \otimes V$ by definition.
Composing the block's operations yields $I_{\textsf{prog}} \otimes W^{i}_{F_i}$, and chaining over $i$ from the initial $|0^{m}\rangle$ gives the claim.
\end{proof}

\vspace{-0.5ex}
\begin{lemma}
\label{lem:reduction}
Let $\pazocal{S}$ be a $T$-round strategy using $S_t$ ancillas in round $t$. If $\pazocal{S}$ solves \emph{\textsc{ListAll}}, \emph{\textsc{ExistFail}}, or \emph{\textsc{FirstFail}}, then there exists a $T$-round unitary transition system whose round-$t$ dimension is $2^{S_t}$ that solves the corresponding bit-string pattern distinguishing task on every input $x \in \{0,1\}^n$.
\end{lemma}

\begin{proof}
\vspace{-0.5ex}
For each round $t$, let $W^{t,i}_0, W^{t,i}_1$ ($i \in [n]$) denote the ancilla-only unitaries induced by $\pazocal{I}_t$ (Lem.~\ref{lem:ancilla-evolution}), and construct the $T$-round unitary transition system $(\pazocal{U}^1_n, \ldots, \pazocal{U}^T_n)$ with
{\setlength{\abovedisplayskip}{3pt}
\setlength{\belowdisplayskip}{2pt}
\setlength{\abovedisplayshortskip}{0pt}
\setlength{\belowdisplayshortskip}{0pt}
\begin{align*}
\pazocal{U}^t_n \triangleq \big(\pazocal{H}_{\textsf{anc}_t},\, \{W^{t,i}_0\}_{i=1}^n,\, \{W^{t,i}_1\}_{i=1}^n,\, |0^{S_t}\rangle\big).
\end{align*}
}
Each pair $W^{t,i}_0 $ and $ W^{t,i}_1$ is fixed by $\pazocal{I}_t$ alone, meaning that the constructed system is independent of any program instance, and has per-round dimension $\dim \pazocal{H}_{\textsf{anc}_t} = 2^{S_t}$. On input $x = b_1 \cdots b_n \in \{0,1\}^n$, the final state of round $t$ is $|\psi^t(x)\rangle = (W^{t,n}_{b_n} \cdots W^{t,1}_{b_1})|0^{S_t}\rangle$.

For any $x \,{\in}\, \{0,1\}^n$, let $\pazocal{Q}$ be any program-with-assertions with deterministic failure pattern $F = x$.\footnote{
As a technical note, the existence of such $\pazocal{Q}$ for every $x \,{\in}\, \{0,1\}^n$ is guaranteed by a trivial instance: take a single fresh program qubit with no further updates, and let the $i$-th assertion check whether the qubit is in $\ket{0}$ if $x_i \,{=}\, 0$, or in $\ket{1}$ if $x_i \,{=}\, 1$.
} By Lem.~\ref{lem:ancilla-evolution} at position $n$, each round $t$ of $\pazocal{S}$ on $\pazocal{Q}$ terminates, after the final segment $U_n$, in the product state $U_n|\varphi_n\rangle_{\textsf{prog}} \otimes (W^{t,n}_{F_n} \cdots W^{t,1}_{F_1})|0^{S_t}\rangle$, whose ancilla part is exactly the final state \smash{$|\psi^t(x)\rangle$} of \smash{$\pazocal{U}^t_n$} on input $x$. 
Because this state is a product across $\textsf{prog}$ and $\textsf{anc}_t$, and the $T$ rounds are executed independently with fresh initializations, the readout $(R_1, \ldots, R_T)$ of $\pazocal{S}$ on $\pazocal{Q}$ is distributed exactly as the transcript $(r_1, \ldots, r_T)$ of the transition system on $x$, read only at the positions $\pazocal{S}$ measures in each round.
Finally, equip the system with a decoder that reads each $r_t$ at those positions and invokes \textsf{Dec}. 
With the input bit-string playing the role of the failure pattern, the correct answers of the assertion checking and bit-string distinguishing tasks coincide on $\pazocal{Q}$ and $x$. 
Since $\pazocal{S}$ solves the task on $\pazocal{Q}$, the constructed decoder answers correctly on $x$.
As $x$ ranges over $\{0,1\}^n$, the transition system solves the corresponding distinguishing task.
\end{proof}

\vspace{-0.25ex}
As a consequence of Lem.~\ref{lem:reduction}, any lower bound on the dimension of transition systems solving a distinguishing task is a lower bound on the ancilla budget of strategies solving the corresponding checking task. If there existed any strategy that beats the ancilla bound, it would yield a transition system that beats the dimension bound and solves the same distinguishing task.

\vspace{0.5ex}
Finally, we state one more correctness lemma that will be repeatedly invoked in the following sections. Whenever two inputs to a transition system have different correct answers, some round of the transition system must tell them apart at the level of measurement outcomes.

\begin{lemma}[Answer Separation]\label{lem:answer-separation}
Let $(\pazocal{U}^1_n,\dots,\pazocal{U}^T_n)$ be a $T$-round unitary transition system that solves one of the distinguishing tasks of Def.~\ref{def:distinguish} with decoder $\textsf{Dec}$. Then for any two inputs $u,v \,{\in}\, \{0,1\}^n$ whose correct answers differ, there exists a round $t \,{\in}\, [T]$ in which the final states $|\psi^t(u)\rangle$ and $|\psi^t(v)\rangle$ are supported on disjoint sets of computational basis states. In particular, $\langle \psi^t(u) \,|\, \psi^t(v)\rangle = 0$.
\end{lemma}

\begin{proof}
\vspace{-0.5ex}
Assume toward contradiction that in every round $t$, some computational basis state has nonzero amplitude in both $|{\psi^t(u)}\rangle$ and $|{\psi^t(v)}\rangle$.
Select one such outcome $r_t$ per round. Since the rounds are measured independently, the transcript $(r_1, \dots, r_T)$ occurs with nonzero probability on both inputs, so the deterministic decoder returns the same answer on $u$ and $v$, contradicting the fact that their correct answers differ.
\end{proof}

%% file: main/exist.tex
\section{Complexity of the \textsc{ExistFail} Task}
\label{sec:exist}

This section presents our complexity results for \textsc{ExistFail} in three settings: single-round (Sec.~\ref{sec:single-round-exist}), multi-round with disjoint assertions (Sec.~\ref{sec:disjoint-multi-round-exist}), and general multi-round (Sec.~\ref{sec:general-multi-round-exist}).

\input{main/exist-single}
\input{main/exist-disjoint-multi}

\input{main/exist-general-multi}

%% file: main/exist-single.tex
\subsection{Single-Round Space Complexity of \textsc{ExistFail}}
\label{sec:single-round-exist}

We first establish the lower bound $S \,{\ge}\, \lceil \log(n\,{+}\,1)\rceil$ in Thm.~\ref{thm:single-round-exist-lower-bound}, and then present a strategy using $S \,{\le}\, \lceil \log(n\,{+}\,1)\rceil\,{+}\,1$ ancillas to witness the upper bound in Thm.~\ref{thm:single-round-exist-upper-bound}. As a consequence, the single-round space complexity of \textsc{ExistFail} is $S = \Theta(\log n)$, shown in Thm.~\ref{thm:single-round-exist-tightness}.

\begin{theorem}
\label{thm:single-round-exist-lower-bound}
    For any single-round strategy (as defined in Def.~\ref{def:stra}) that solves \emph{\textsc{ExistFail}}, it holds that $S \ge \lceil \log(n\,{+}\,1)\rceil$, with $n$ the number of assertions.
\end{theorem}

\begin{proof}
\vspace{-0.5ex}
Such a strategy solves \textsc{ExistFail} over all $2^n$ deterministic failure patterns. By Lemma~\ref{lem:reduction}, it thus yields a one-round unitary transition system\,\,{\small \smash{$\pazocal{U}_n = (\pazocal{H}, \{\pazocal{T}^{\scriptscriptstyle (i)}_0\}_{i=1}^n, \{\pazocal{T}^{\scriptscriptstyle (i)}_1\}_{i=1}^n, |0^S\rangle)$}} with $\dim(\pazocal{H}) = 2^S$ that solves the corresponding task \textsc{Exist} (Def.~\ref{def:distinguish}).

The correct answers to \textsc{Exist} on $0^n$ and on any $x \neq 0^n$ differ, so Lemma~\ref{lem:answer-separation} (with $T=1$) gives
$\langle \psi(0^n)\,|\,\psi(x)\rangle = 0$ for every $x \neq 0^n$.
By the following lemma (Lem.~\ref{thm:unitary-dim-lb}), any such transition system satisfies $\dim(\pazocal{H}) \ge n + 1$. It follows that $2^S \ge n+1$, and $S \ge \lceil \log(n+1)\rceil$ as required.
\end{proof}

\begin{lemma}
\label{thm:unitary-dim-lb}
Let\,\,{\small $\pazocal{U}_n=(\pazocal{H},\{\pazocal{T}^{\scriptscriptstyle (i)}_0\}_{i=1}^n,\{\pazocal{T}^{\scriptscriptstyle (i)}_1\}_{i=1}^n,|\psi_0\rangle)$} be a one-round finite-dimensional unitary transition system as in Def.~\ref{def:unitary-transition-system}.
Assuming that $\forall x\neq 0^n$, { $\langle \psi(0^n) \,|\, \psi(x)\rangle$} $ = 0$, we have $\dim(\pazocal{H})\ge n+1$.
\end{lemma}

\begin{proof}[Proof Sketch]
\vspace{-0.5ex}
We sketch the key idea; the full proof is in Appendix~\ref{app:unitary-dim-lb-proof}.
For each $k \,{\in}\, \{0,\ldots,n\}$, let $|s_k\rangle$ be the intermediate state after reading $k$ zeros along the all-zero string transition, and denote the unitary segment induced by reading $u\in\{0,1\}^{\ell}$ from step $i+1$ to $i+\ell$ as \smash{$\pazocal{T}^{[i]}_{u}$}. The key observation is a pumping-style orthogonality property:
{\setlength{\abovedisplayskip}{2pt}
\setlength{\belowdisplayskip}{3pt}
\setlength{\abovedisplayshortskip}{0pt}
\setlength{\belowdisplayshortskip}{0pt}
\begin{align*}
u\neq 0^\ell
\quad \text{implies} \quad
\langle s_{i+\ell}\mid \pazocal{T}^{[i]}_{u}\mid s_i\rangle = 0.
\end{align*}
}

\noindent
Otherwise, the input $0^i u\,0^{n-i-\ell}$ would produce a final state not orthogonal with \smash{$|\psi(0^n)\rangle$},
contradicting the assumption that \smash{$|\psi(0^n)\rangle$} is orthogonal to \smash{$|\psi(x)\rangle$} for every $x\neq 0^n$.

\smallskip
This property yields a counting argument. Consider the $n\,{+}\,1$ strings $x_k \,{\triangleq}\, 0^{n-k}1^k$, where $0 \,{\le}\, k \,{\le}\, n$, and let $|t_k\rangle\triangleq |\psi(x_k)\rangle$ be their final states.
We claim they are pairwise orthogonal.
First, $|t_0\rangle=|\psi(0^n)\rangle$ is orthogonal to every $|t_k\rangle$ with $k \,{\ge}\, 1$ by assumption.
For $1 \,{\le}\, i \,{<}\, j \,{\le}\, n$, factoring out of $\langle t_i \,|\, t_j\rangle$ the common suffix evolution of the last $i$ ones reduces it to \smash{$\langle s_{n-i}\mid \pazocal{T}^{[n-j]}_{1^{j-i}}\mid s_{n-j}\rangle$}, which is zero by the pumping property, since $1^{j-i} \,{\neq}\, 0^{j-i}$.
Hence, the $n+1$ states $|t_0\rangle,|t_1\rangle,\ldots,|t_n\rangle$ are pairwise orthogonal and $\pazocal{H}$ contains at least $n+1$ linearly independent vectors, so $\dim(\pazocal{H})\ge n+1$.
\end{proof}

\begin{theorem}
\label{thm:single-round-exist-upper-bound}
There exists a single-round strategy solving \emph{\textsc{ExistFail}} with $S \le \lceil \log(n \,{+}\, 1)\rceil \,{+}\, 1$.
\end{theorem}

\begin{proof}
\vspace{-0.25ex}
We propose a single-round \textsc{ExistFail} strategy using
$S=\lceil \log(n \,{+}\, 1)\rceil \,{+}\, 1$ ancilla qubits:

\begin{strategy}[\textsc{\emph{Single-Round Modulo Increment}}]
\label{str:single-round-modulo-increment}
\vspace{-0.5ex}
The instrumentation $\pazocal{I}$ partitions $\textsf{anc}$ into a \emph{counter register} $\textsf{ctr}$ of size $\ell\triangleq \lceil \log(n \,{+}\, 1)\rceil$, intended to encode a value in $\{0,1,\ldots,n\}$, along with a \emph{single-qubit flag} $\textsf{fail}$, both initialized to zero. For each assertion position $i$, $\pazocal{I}$ inserts these steps:
\begin{enumerate}[topsep=1pt,itemsep=1pt,leftmargin=2em]
\item[(1)] apply the checker unitary $C_{i}$ with $\textsf{fail}$ as the designated target, denoted by $C_{i \to \textsf{fail}}$;
\item[(2)] apply ancilla-only unitary $V_{\mathrm{CINC}}$ that, conditioned on $\textsf{fail}=\ket{1}$, increments $\textsf{ctr}$ by $1$ modulo $2^\ell$;
\item[(3)] apply $C_{i}$ with $\textsf{fail}$ again to uncompute $\textsf{fail}$ back to $\ket{0}$. 
Formally, the assertion-handling block
\end{enumerate}
{\setlength{\abovedisplayskip}{3pt}
\setlength{\belowdisplayskip}{3pt}
\setlength{\abovedisplayshortskip}{0pt}
\setlength{\belowdisplayshortskip}{0pt}
\begin{align*}
B_i  ~=~
C_{i \to\textsf{fail}};
\ V_\mathrm{CINC};
\ C_{i \to\textsf{fail}}.
\end{align*}
}

\noindent
In addition, the handling block $B_n$ of the final assertion appends one more ancilla-only unitary $V_{\neq 0}$, which routes into \textsf{fail} whether $\textsf{ctr}$ is nonzero.\footnote{Formally, $V_{\neq 0} \triangleq |0^\ell\rangle\langle 0^\ell| \otimes I + (I - |0^\ell\rangle\langle 0^\ell|) \otimes X$, acting on $(\textsf{ctr}, \textsf{fail})$, with \textsf{ctr} as control and \textsf{fail} as target.}
At the end of executing the transformed program, \textsf{ctr} equals the number of failing assertions, so $\textsf{fail} = \ket{1}$ iff some assertion fails. Thus, measuring the single qubit \textsf{fail} at termination and applying $\textsf{Dec}$ that outputs its value solves \textsc{ExistFail}.
\end{strategy}

\vspace{-0.5ex}
Strategy~\ref{str:single-round-modulo-increment} uses $S=\lceil \log(n \,{+}\, 1)\rceil \,{+}\, 1$ ancilla qubits, which proves the theorem.
\end{proof}

\begin{example}
To visualize how Strategy~\ref{str:single-round-modulo-increment} operates, let the program contain $n \,{=}\, 4$ assertions and suppose the failure pattern is $F \,{=}\, (0,1,0,1)$.
Then $\ell \,{=}\, \lceil \log 5\rceil \,{=}\, 3$, so the strategy uses a $3$-qubit counter $\textsf{ctr}$ and a one-qubit flag $\textsf{fail}$.
All passing assertions leave $\textsf{ctr}$ unchanged, while failing assertions increment it by one.
Therefore, the evolution over ancillas is

\vspace{-2ex}
{\small
\setlength{\fboxsep}{0pt}
\setlength{\fboxrule}{0pt}
\[
\begin{aligned}
\ket{000}_{\textsf{ctr}}\ket{0}_{\textsf{fail}}\,
&\xrightarrow{\cdots} \,
\fcolorbox{gray!100}{gray!7.5}{$ \xrightarrow[\text{unchanged}]{\text{A}_1~\text{pass}} \ket{000}\ket{0} $\,} \,\xrightarrow{\cdots} \, \fcolorbox{gray!100}{gray!7.5}{$ \xrightarrow[\text{write \textsf{fail}}]{C_{2 \to \textsf{fail}}} \ket{000}\ket{1} \xrightarrow[\text{increment}]{V_{\mathrm{CINC}}} \ket{001}\ket{1} \xrightarrow[\text{uncompute \textsf{fail}}]{C_{2 \to \textsf{fail}}} \ket{001}\ket{0}$} \,\,\,A_2 \text{ fail}
\\[0ex]
&\xrightarrow{\cdots} \,
\fcolorbox{gray!100}{gray!7.5}{$\displaystyle \xrightarrow[\text{unchanged}]{\text{A}_3~\text{pass}} \ket{001}\ket{0} $\,} \, \xrightarrow{\cdots} \,\fcolorbox{gray!100}{gray!7.5}{$\displaystyle \xrightarrow[\text{write \textsf{fail}}]{C_{4 \to \textsf{fail}}} \ket{001}\ket{1} \xrightarrow[\text{increment}]{V_{\mathrm{CINC}}} \ket{010}\ket{1} \xrightarrow[\text{uncompute \textsf{fail}}] {C_{4 \to \textsf{fail}}} \ket{010}\ket{0}$} \,\,\,A_4 \text{ fail}
\\[0ex]
&\xrightarrow{\cdots} \,
|010\rangle\,|0\rangle \xrightarrow[\text{route } \textsf{ctr} \neq 0]{V_{\neq 0}} |010\rangle\,|1\rangle \xrightarrow{\text{measure \textsf{fail}}} 1 \xrightarrow{\textsf{Dec}} \text{true (exist)}
\end{aligned}
\]}

\noindent
where the final \textsf{ctr} value is 2, so $V_{\neq 0}$ sets $\textsf{fail} = \ket{1}$ and measuring it outputs true for \textsc{ExistFail}.
\end{example}

\paragraph{Cost of Strategy~\ref{str:single-round-modulo-increment}}
The strategy has $T \,{=}\, 1$ and $S \,{=}\, \lceil \log(n \,{+}\, 1)\rceil \,{+}\, 1$.
The checker unitary is invoked twice per assertion, so $C\,{=}\,2$.
Only \textsf{fail} is measured at the end, so $M = 1$.
For gate cost, each assertion applies one controlled increment on an $\ell$-qubit counter, where $\ell \,{=}\, \lceil \log(n+1)\rceil$.
A staircase of $\ell$ multi-controlled gates implements this operation with no additional workspace. 
The single unitary $V_{\neq 0}$ adds $O(\ell)$ gates. 
Thus, the total non-checker gate cost is \smash{$G = n\cdot {O}(\ell) + {O}(\ell) ={O}(n\log n)$}.

\begin{theorem}
\label{thm:single-round-exist-tightness}
The single-round space complexity of \emph{\textsc{ExistFail}} satisfies $S = \Theta(\log n)$.
\end{theorem}

\begin{proof}
Immediate by Theorem~\ref{thm:single-round-exist-lower-bound} and Theorem~\ref{thm:single-round-exist-upper-bound}.
\end{proof}

%% file: main/exist-disjoint-multi.tex
\subsection{Disjoint Multi-Round Time-Space Trade-Off for \textsc{ExistFail}}
\label{sec:disjoint-multi-round-exist}

We first prove a $(2^S-1)\cdot T = \Omega(n)$ lower bound in Thm.~\ref{thm:multi-round-disjoint-exist-lower-bound}, then give a matching upper bound in Thm.~\ref{thm:multi-round-disjoint-exist-upper-bound}. Consequently, the complexity is a trade-off of form $S = \Theta(\log (1 + \frac{n}{T}))$ in Thm~\ref{thm:multi-round-disjoint-exist-complexity}.

\begin{theorem}
\label{thm:multi-round-disjoint-exist-lower-bound}
For any disjoint multi-round strategy (as defined in Def.~\ref{def:multi-stra}) that solves \emph{\textsc{ExistFail}}, it holds that
$(2^S-1)\cdot T \ge n$, with $n$ the number of assertions. 
\end{theorem}
\begin{proof}
\vspace{-0.5ex}
Let $\mathcal{E}_t$ be the enabled assertion indices of the instrumentation $\pazocal{I}_t$, and let $L_t \triangleq |\mathcal{E}_t|$.
Such a strategy in particular solves \textsc{ExistFail} over all $2^n$ deterministic failure patterns, so by Lemma~\ref{lem:reduction} it yields a $T$-round unitary transition system solving \textsc{Exist}, whose round-$t$ final state on input $x$ is $|\psi^{t}(x)\rangle = \big(W^{t,n}_{x_n} \cdots W^{t,1}_{x_1}\big)|0^{S_t}\rangle$, where, by their construction in Lemma~\ref{lem:ancilla-evolution}, $W^{t,i}_0 = W^{t,i}_1 = I$ for every $i \,{\notin}\, \mathcal{E}_t$.
Hence $|\psi^{t}(x)\rangle$ depends only on the bits of $x$ within $\mathcal{E}_t$.

\vspace{0.25ex}
First, every index must belong to some $\mathcal{E}_t$: if $i \,{\notin}\, \mathcal{E}_t$ for all $t$, then the all-zero pattern and the pattern in which only the $i$-th assertion fails induce identical final states in every round, hence identically distributed transcripts, and no decoder can answer \textsc{Exist} correctly on both.
Because the strategy is disjoint, all $\mathcal{E}_t$ therefore form a partition of $[n]$, and \smash{$\textstyle\sum_{t=1}^T L_t = n$}.

\vspace{0.25ex}
Next, fix a round $t$, and let $x$ be any nonzero pattern with $x_i = 0$ for all $i \notin \mathcal{E}_t$; we claim that $\langle \psi^t(0^n)\,|\,\psi^t(x)\rangle = 0$. The correct answers to
\textsc{Exist} on $0^n$ and $x$ differ, so by Lemma~\ref{lem:answer-separation} some round's final states on the two inputs have disjoint basis supports. For every $t' \neq t$, disjointness of the strategy gives $x_i = 0$ for all $i \in \mathcal{E}_{t'}$; since the round-$t'$
final state depends only on the bits within $\mathcal{E}_{t'}$, we get
$|\psi^{t'}(x)\rangle \,{=}\, |\psi^{t'}(0^n)\rangle$ --- identical states, whose supports coincide and are in particular not disjoint. The separating round must therefore be $t$, and the claimed orthogonality follows.

\vspace{0.25ex}
Now reading only the positions in $\mathcal{E}_t$ in increasing order, round $t$ is a one-round unitary transition system for inputs of length $L_t$ (positions outside $\mathcal{E}_t$ contribute the identity), and the claim above is exactly the orthogonality hypothesis of Lemma~\ref{thm:unitary-dim-lb} with $n$ replaced by $L_t$.
Hence $2^{S_t} = \dim(\pazocal{H}_{\textsf{anc}_t}) \ge L_t + 1$, i.e., $2^{S_t} - 1 \ge L_t$.
Summing over all rounds gives $\textstyle\sum_{t=1}^T (2^{S_t}-1)\ge \textstyle\sum_{t=1}^T L_t = n$.
Finally, we have $2^{S_t} {-}\, 1 \,{\le}\, 2^S {-}\, 1$ for every $t$, and thus $T \cdot (2^S {-}\, 1) \,{\ge}\, \textstyle\sum_{t=1}^T (2^{S_t} {-}\, 1) \,{\ge}\, n$.
\end{proof}

\begin{theorem}
\vspace{-0.5ex}
\label{thm:multi-round-disjoint-exist-upper-bound}
For all $T\le n$, there exists a disjoint multi-round strategy solving \emph{\textsc{ExistFail}} with $(2^S \,{-}\, 1) \,{\cdot}\, T \,{<} 11n$.
\end{theorem}

\begin{proof}
\vspace{-1ex}
We construct the following strategy:

\vspace{-0.5ex}
\begin{strategy}[\textsc{\emph{Partitioned Modulo Increment}}]
\label{str:partitioned-modulo-increment}
Partition the $n$ assertions into $T$ disjoint contiguous index sets $\mathcal{E}_1, \ldots, \mathcal{E}_T$, each of size at most $\lceil n/T\rceil$.
In round $t$, use the instrumentation from the \emph{single-round modulo-increment strategy} (Strategy~\ref{str:single-round-modulo-increment}) on $\mathcal{E}_t$ only, using a counter register of size \smash{$\left\lceil \log(\lceil n/T\rceil \,{+}\, 1)\right\rceil$} and one flag qubit $\textsf{fail}$, with $V_{\neq 0}$ appended to the block of the last assertion in $\mathcal{E}_t$. Each round measures $\textsf{fail}$, and after all $T$ rounds, outputs true iff some round measures $1$.
\end{strategy}

\vspace{-1ex}
We now analyze the trade-off.
Let $L \,{\triangleq}\, \lceil n/T\rceil$.
Each round uses at most $S \,{=}\, \lceil \log(L \,{+}\, 1)\rceil \,{+}\, 1$ ancillas.
$2^S = 2^{\lceil \log(L \,{+}\, 1)\rceil \,{+}\, 1} < 4 (L+1)$.
Hence, $2^S {-}\, 1 \,
{<} \, 4L\,{+}\,3$. Multiply by $T$ to get $(2^S{-}\,1)\cdot T \,{<}\, (4L \,{+}\, 3)T$. Using $L \le n/T \,{+}\, 1$, we obtain $(2^S{-}\,1)\cdot T \,{<}\, 4n\,{+}\,7T$.
Because every round contains at least one assertion, we have $T\le n$, and therefore $(2^S-1)\cdot T < 11n$.
\end{proof}

\paragraph{Cost of Strategy~\ref{str:partitioned-modulo-increment}}
\vspace{-0.5ex}
For any $T\le n$, the strategy uses
\smash{$S \,{=}\, \left\lceil \log(\lceil n/T\rceil \,{+}\, 1)\right\rceil \,{+}\, 1$}
ancillas per round.
Each assertion is checked in exactly one round, and its checker unitary is invoked twice, so $C \,{=}\, 2$.
$M = T$, since only \textsf{fail} is measured.
Each enabled assertion contributes a controlled increment on a counter of bit width \smash{$\ell \,{=}\, \left\lceil \log(\lceil n/T\rceil \,{+}\, 1)\right\rceil$}.
This operation costs ${O}(\ell)$ elementary gates and each round adds $O(\ell)$ for $V_{\neq 0}$. Summing over all rounds yields $G={O}(n\ell)={O}\left(n\log(\lceil n/T\rceil \,{+}\, 1)\right)$.

\begin{theorem}
\label{thm:multi-round-disjoint-exist-complexity}
For all $T \,{\le}\, n$, the disjoint multi-round complexity of \emph{\textsc{ExistFail}} is $S = \Theta\left(\log (1 + \frac{n}{T})\right)$.
\end{theorem}

\begin{proof}
\vspace{-1ex}
For the lower bound side, Thm.~\ref{thm:multi-round-disjoint-exist-lower-bound} shows $(2^S-1)T \ge n$, hence
$\textstyle S \ge \log(1+\frac{n}{T})$.

For the upper bound side, Thm.~\ref{thm:multi-round-disjoint-exist-upper-bound} gives a strategy satisfying $(2^S-1)T < 11n$, which implies
{\setlength{\abovedisplayskip}{3pt}
\setlength{\belowdisplayskip}{2pt}
\setlength{\abovedisplayshortskip}{0pt}
\setlength{\belowdisplayshortskip}{0pt}
\[
    \textstyle S < \log  (1+\frac{11n}{T}) = O  \bigl(\log  (1+\frac{n}{T})\bigr). \text{ Therefore } S = \Theta  \left(\log  (1+\frac{n}{T})\right). \qedhere
\]}%
\end{proof}

%% file: main/exist-general-multi.tex
\subsection{General Multi-Round Time-Space Trade-Off for \textsc{ExistFail}}
\label{sec:general-multi-round-exist}
In the most general case, the enabled assertion indices may overlap arbitrarily across rounds, so that a failure pattern can be encoded differently in different rounds; the natural complexity measure is then the product $S \,{\cdot}\, T$.
We prove an $S \,{\cdot}\, T \,{\ge}\, \log(n\,{+}\,1)$ lower bound (Thm.~\ref{thm:multi-round-general-exist-lower-bound}) and construct a family of parameterized upper bound strategies (Thm.~\ref{thm:multi-round-general-exist-upper-bound-lcm}).
Thm.~\ref{thm:multi-round-general-exist-loglog-tight} shows this construction is asymptotically tight for $T \,{=}\, O\left({\log n}/{\log\log n}\right)$, so the general multi-round complexity of \textsc{ExistFail} satisfies \smash{$S = \Theta\big(\frac{1}{T}\log n\big)$} throughout this regime.

\begin{theorem}
\label{thm:multi-round-general-exist-lower-bound}
For any multi-round strategy (as defined in Def.~\ref{def:multi-stra}) that solves \emph{\textsc{ExistFail}}, it holds that $S\cdot T \ge \log(n+1)$, with $n$ the number of assertions.
\end{theorem}

\begin{proof}
\vspace{-0.5ex}
Such a strategy solves \textsc{ExistFail} over all $2^n$ deterministic failure patterns. Therefore, by Lemma~\ref{lem:reduction} it yields a $T$-round unitary transition system $(\pazocal{U}^{1}_n, \ldots, \pazocal{U}^{T}_n)$ solving \textsc{Exist}, where $\pazocal{U}^{t}_n = (\pazocal{H}_{\textsf{anc}_t},\{W^{t,i}_0\}_{i=1}^n,\{W^{t,i}_1\}_{i=1}^n,|0^{S_t}\rangle)$ is built from the ancilla-only unitaries of Lemma~\ref{lem:ancilla-evolution}, and the transcript $(r_1, \ldots, r_T)$ is decoded by \textsf{Dec}.
Now form a one-round unitary transition system
{\setlength{\abovedisplayskip}{3pt}
\setlength{\belowdisplayskip}{2pt}
\setlength{\abovedisplayshortskip}{0pt}
\setlength{\belowdisplayshortskip}{0pt}
\begin{align*}
\pazocal{U}^{\Pi}_n=
(\pazocal{H}^{\Pi},\{\Pi^{\scriptscriptstyle (i)}_0\}_{i=1}^n,\{\Pi^{\scriptscriptstyle (i)}_1\}_{i=1}^n,|0^{\Pi}\rangle),
\end{align*}
}
\noindent
where $\pazocal{H}^{\Pi} \,{\triangleq}\, {\text{\small $\textstyle \bigotimes_{t=1}^{T}$}} \pazocal{H}_{\textsf{anc}_t}$, $|0^{\Pi}\rangle \,{\triangleq}\, {\text{\small $\textstyle \bigotimes_{t=1}^{T}$}} |0^{S_t}\rangle$, and $\Pi^{\scriptscriptstyle (i)}_b \,{\triangleq}\, {\text{\small $\textstyle \bigotimes_{t=1}^{T}$}} W^{t,i}_b$.
Since the position-$i$ update acts round-wise, the final state of $\pazocal{U}^{\Pi}_n$ on input $x \,{\in}\, \{0,1\}^n$ factorizes as
{\setlength{\abovedisplayskip}{3pt}
\setlength{\belowdisplayskip}{2pt}
\setlength{\abovedisplayshortskip}{0pt}
\setlength{\belowdisplayshortskip}{0pt}
\begin{align*}
|\psi^{\Pi}(x)\rangle= \textstyle \bigotimes_{t=1}^T |\psi^{t}(x)\rangle,
\end{align*}
}

\noindent
namely, the tensor product of the final states produced by each round on input $x$.
Measuring $|\psi^{\Pi}(x)\rangle$ in the computational basis and grouping the outcome bits into $T$ blocks is equivalent to measuring each round's final state independently; thus applying \textsf{Dec} to the grouped outcome of $\pazocal{U}^{\Pi}_n$ reproduces the transcript decoding, and $\pazocal{U}^{\Pi}_n$ solves \textsc{Exist} as a one-round system.
By Lemma~\ref{lem:answer-separation}, $\langle \psi^\Pi(0^n)\,|\,\psi^\Pi(x)\rangle = 0$ for every $x \neq 0^n$.
Applying Lemma~\ref{thm:unitary-dim-lb} to $\pazocal{U}^{\Pi}_n$ gives \smash{$\dim(\pazocal{H}^{\Pi}) \,{\ge}\, n \,{+}\, 1$}.
On the other hand, \smash{$\dim(\pazocal{H}^{\Pi}) = \prod_{t=1}^T 2^{S_t} = 2^{\sum_{t=1}^T S_t} \le 2^{ST}$}. Hence, $2^{ST}\ge n+1$, i.e., $S\cdot T \ge \log(n+1)$.
\end{proof}

\vspace{-1ex}
To complement the lower bound, we next give a construction based on \emph{multi-modulus fingerprinting}: each round enables all assertions, but records the total number of failures modulo a different integer.
The decoder checks whether all recorded residues are zero.
This technique works whenever the least common multiple (lcm) of the chosen moduli exceeds the maximum number of failures.

\begin{theorem}
\label{thm:multi-round-general-exist-upper-bound-lcm}
Let $\mu_1,\ldots,\mu_T$ be integers with $\mu_t\ge 2$ for all $t$, and let $\Lambda\triangleq \mathrm{lcm}(\mu_1,\ldots,\mu_T)$.
If $\Lambda>n$, then there exists a multi-round strategy solving \emph{\textsc{ExistFail}} on $n$ assertions with
$S = 1 \,{+}\, \max_{t} \lceil \log \mu_t\rceil$.
\end{theorem}

\begin{proof}
\vspace{-0.5ex}
We construct the following strategy:

\begin{strategy}[\textsc{\emph{Multi-Round lcm Fingerprinting}}]
\label{str:lcm-fingerprint}
\vspace{-0.5ex}
For each round $t\in[T]$, $\pazocal{I}_t$ enables all $n$ assertions, i.e., $\mathcal{E}_t=[n]$.
Let $\ell_t\triangleq \lceil \log \mu_t\rceil$.
Round $t$ uses an $\ell_t$-qubit counter register $\textsf{ctr}_t$ together with a one-qubit flag $\textsf{fail}$, all initialized to $\ket{0}$.
Each assertion-handling block for $A_i$ in round $t$ is
{\setlength{\abovedisplayskip}{1pt}
\setlength{\belowdisplayskip}{2pt}
\setlength{\abovedisplayshortskip}{0pt}
\setlength{\belowdisplayshortskip}{0pt}
\begin{align*}
B_i^{\scriptscriptstyle (t)}
=
C_{i \to \textsf{fail}};
\ V^{(t)}_{\mathrm{CINC}};
\ C_{i \to \textsf{fail}},
\end{align*}
}

\noindent
where the first checker-unitary call routes the assertion outcome into \textsf{fail} and the second call uncomputes \textsf{fail}. 
Here, \smash{$V^{\scriptscriptstyle (t)}_{\scriptscriptstyle \mathrm{CINC}}$}, conditioned on $\textsf{fail}=\ket{1}$, increments $\textsf{ctr}_t$ by $1$ modulo $\mu_t$ on the basis states $\ket{0}, \ldots, |\mu_t \,{-}\, 1\rangle$, with an arbitrary unitary extension to the remaining ones.
The block of the last assertion appends the routing unitary $V_{\neq 0}$ of Strategy~\ref{str:single-round-modulo-increment} (with $|0^{\ell_t}\rangle$ in place of $|0^{\ell}\rangle$), so that \textsf{fail} records whether $\textsf{ctr}_t$ is nonzero. Let $r_t \in \{0,1\}$ be the measurement outcome of \textsf{fail} at the end of round $t$. After all $T$ rounds, the decoder outputs 0 iff $r_t = 0$ for every $t$, and outputs 1 otherwise.
\end{strategy}

\paragraph{Correctness.}
\vspace{-0.5ex}
Let $k \,{\in}\, \{0,\ldots,n\}$ be the number of failing assertions.
In round $t$, every failure contributes one increment modulo $\mu_t$, while every pass contributes none, so the final counter value is exactly $k \bmod \mu_t$ and $r_t = 0$ iff $k \equiv 0 \pmod{\mu_t}$.
Hence the decoder outputs $0$ iff $k\equiv 0 \pmod{\mu_t}$ for all $t$,
equivalently iff $k \,{\equiv}\, 0 \pmod{\Lambda}$.
Because $0 \,{\le}\, k \,{\le}\, n \,{<}\, \Lambda$, this condition holds iff $k \,{=}\, 0$.
Thus, the strategy solves \textsc{ExistFail}.
Its per-round space cost is $\ell_t \,{+}\, 1 \,{=}\, \lceil \log\mu_t\rceil \,{+}\, 1$, so $S \,{=}\, 1+\max_{t}\lceil \log\mu_t\rceil$.
\end{proof}

\begin{example}
\vspace{-0.5ex}
Let $n=5$, and suppose the failure pattern is $F=(1,1,0,1,1)$, so there are $k=4$ failures.
Choose $T=2$ with $(\mu_1,\mu_2)=(2,3)$.
Then $\Lambda=\mathrm{lcm}(2,3)=6>5$.
Hence, the two rounds use counter sizes $\ell_1=\lceil \log 2\rceil =1$ and $\ell_2=\lceil \log 3\rceil =2$. In both rounds, all five assertions are enabled. For the present failure pattern, the counter registers evolve as:

{\small
\setlength{\arraycolsep}{1pt}
\renewcommand{\arraystretch}{2}
$\begin{array}{lllllll}
\hspace{30pt} \text{Round }1~(\mu_1=2): & \hspace{5pt}\ket{0}_{\textsf{ctr}_1} & \xrightarrow{A_1~\text{fail}} \ket{1} & \xrightarrow{A_2~\text{fail}} \ket{0} & \xrightarrow{A_3~\text{pass}} \ket{0} & \xrightarrow{A_4~\text{fail}} \ket{1} & \xrightarrow{A_5~\text{fail}} \ket{0},
\\
\hspace{30pt} \text{Round }2~(\mu_2=3): & \hspace{5pt} \ket{00}_{\textsf{ctr}_2} & \xrightarrow{A_1~\text{fail}} \ket{01} & \xrightarrow{A_2~\text{fail}} \ket{10} & \xrightarrow{A_3~\text{pass}} \ket{10} & \xrightarrow{A_4~\text{fail}} \ket{00} & \xrightarrow{A_5~\text{fail}} \ket{01}.
\end{array}$
}

\vspace{0.5ex}
Therefore, $V_{\neq 0}$ leaves $\textsf{fail} = \ket{0}$ in round 1 (as $4 \bmod 2 = 0$) and sets $\textsf{fail} = \ket{1}$ in round 2 (as $4 \bmod 3 = 1$). So round 2 records a $1$, and the decoder outputs true for \textsc{ExistFail}.
\end{example}

Theorem~\ref{thm:multi-round-general-exist-upper-bound-lcm} gives a family of upper bounds parameterized by the chosen moduli.
We further show that this construction remains tight throughout the regime $T = O( \log n / \log\log n)$.
In particular, throughout this regime, one can increase the number of rounds $T$ and correspondingly reduce the per-round space to \smash{$S = O\big(\frac{\log n}{T}\big)$}, while remaining optimal in the rounds--ancillas product by matching the universal lower bound $S \,{\cdot}\, T \,{=}\, \Omega(\log n)$.
At the upper end of this tight range, when $T = \Theta(\log n / \log\log n)$, this approach yields a strategy with $S = \Theta(\log\log n)$.

\begin{theorem}
\label{thm:multi-round-general-exist-loglog-tight}
For any parameter regime where
$T = {O}(\log n / \log\log n)$,
the general multi-round complexity of \emph{\textsc{ExistFail}} satisfies
$S = \Theta\big(\frac{\log n}{T}\big)$.
\end{theorem}

\begin{proof}
The lower bound follows from Thm.~\ref{thm:multi-round-general-exist-lower-bound}.
For the upper bound, fix a constant $c>0$ such that {$T \le c { \frac{\log n}{\log\log n}}$} for all sufficiently large $n$, let $\alpha \,{\triangleq}\, c+1$, and define
\smash{$S' \,{\triangleq}\, \lceil \alpha \frac{\log(n+1)}{T} \rceil \,{+}\, 2$}.
It suffices to exhibit $T$ distinct primes in $[2^{S'-2},2^{S'-1}]$, from which we choose $\mu_1,\ldots,\mu_T$ and apply Thm.~\ref{thm:multi-round-general-exist-upper-bound-lcm}.
\smallskip
Let $x \triangleq 2^{S'{-}\,2}$. Since $S'\,{-}\,2 \ge \alpha\frac{\log(n+1)}{T}$ and $T \le c\frac{\log n}{\log\log n}$,
{\setlength{\abovedisplayskip}{2pt}
\setlength{\belowdisplayskip}{1pt}
\setlength{\abovedisplayshortskip}{0pt}
\setlength{\belowdisplayshortskip}{0pt}
\begin{align*}
x \ge 2^{\alpha \cdot \log(n+1)/T} \ge 2^{\alpha \cdot \log n/T} \ge 2^{(\alpha/c) \cdot\log\log n} = (\log n)^{\alpha/c}.
\end{align*}}

By the Prime Number Theorem the number of primes in $[x,2x]=[2^{S'-2},2^{S'-1}]$ is asymptotically $\frac{x}{\ln x}$; since $x \ge \log n$ and $z/\ln z$ is increasing for large $z$, this count is, for all sufficiently large $n$,
{\small
\setlength{\abovedisplayskip}{3pt}
\setlength{\belowdisplayskip}{4pt}
\setlength{\abovedisplayshortskip}{0pt}
\setlength{\belowdisplayshortskip}{0pt}
\begin{align*}
\Omega \left(\frac{x}{\ln x}\right)
\ge
\Omega \left( \frac{(\log n)^{\alpha/c}}{(\alpha/c) \,{\cdot}\, \ln\log n}\right)
= \Omega \left( (\log n)^{1/c} \cdot \frac{\log n}{\log \log n} \right)
=
\omega \left(\frac{\log n}{\log\log n}\right).
\end{align*}}

Since this strictly dominates {$\frac{\log n}{\log\log n}$}, it is at least $T$ for large $n$ by assumption, so $T$ distinct primes $\mu_1,\ldots,\mu_T$ exist in the interval. These primes are pairwise coprime, hence
{
\setlength{\abovedisplayskip}{3pt}
\setlength{\belowdisplayskip}{4pt}
\setlength{\abovedisplayshortskip}{0pt}
\setlength{\belowdisplayshortskip}{0pt}
\begin{align*}
\Lambda
=
\mathrm{lcm}(\mu_1,\ldots,\mu_T)
=
\mu_1\mu_2\cdots\mu_T
\ge
(2^{S'-2})^T
=
2^{T(S'-2)}.
\end{align*}}

By the definition of $S'$, $T(S'-2) \ge \log(n+1)$, hence $\Lambda  \,{\ge}\, 2^{T(S'-2)}   \,{\ge}\, n\,{+}\,1 \,{>}\, n$ and the moduli satisfy the premise of Thm.~\ref{thm:multi-round-general-exist-upper-bound-lcm}. Then, Thm.~\ref{thm:multi-round-general-exist-upper-bound-lcm} yields a strategy with $S \le 1 + \max_t \lceil \log \mu_t\rceil$. Since each $\mu_t \le 2^{S'-1}$, we have $\lceil \log \mu_t\rceil \le S'\,{-}\,1$, so $S \le S'$ and thus $S \,{\cdot}\, T \le S' \,{\cdot}\, T = O(\log n)$. With the matching lower bound from Thm.~\ref{thm:multi-round-general-exist-lower-bound}, this completes the proof.
\end{proof}

For any $T$ beyond this regime, the disjoint multi-round strategy of Strategy~\ref{str:partitioned-modulo-increment}, being in particular a general multi-round strategy, still provides the upper bound $S = O(\log(1 + \frac{n}{T}))$, while Thm.~\ref{thm:multi-round-general-exist-lower-bound} continues to give $S \cdot T = \Omega(\log n)$. We leave the tightness for regime beyond this to future work.

\paragraph{Cost of Strategy~\ref{str:lcm-fingerprint}}
\vspace{-0.5ex}
The strategy uses $T$ rounds and \smash{$S=1+\max_t \ell_t$}, where $\ell_t=\lceil \log \mu_t\rceil$.
Every round enables all $n$ assertions and invokes the checker unitary twice per assertion, so $C \,{=}\, 2T$.
Only \textsf{fail} is measured per round, hence $M = T$.
In round $t$, each assertion contributes one controlled increment modulo $\mu_t$, realized without workspace as a staircase increment modulo $2^{\ell_t}$ followed, when $\mu_t$ is not a power of two, by the controlled transposition $|0^{\ell_t}\rangle \,{\leftrightarrow}\, |\mathrm{bin}(\mu_t)\rangle$ (cf.\ Strategy~\ref{str:single-round-index-transposition}) fixing the wrap-around --- $\Theta(\ell_t)$ elementary gates in total; $V_{\neq 0}$ adds $O(\ell_t)$ per round.
Hence \smash{$G=O\bigl(n\textstyle\sum_{t} \ell_t\bigr)$}; for the instantiation of Thm.~\ref{thm:multi-round-general-exist-loglog-tight}, $\sum_{t} \ell_t \le T(S-1) = O(\log n)$, so $G = O(n \log n)$.

%% file: main/first.tex
\section{Complexity of the \textsc{FirstFail} Task}
\label{sec:first}

This section presents our complexity results for \textsc{FirstFail} under single-round (Section~\ref{sec:single-round-first}), disjoint multi-round (Section~\ref{sec:disjoint-multi-round-first}), and general multi-round (Section~\ref{sec:general-multi-round-first}) settings.

\input{main/first-single}
\input{main/first-disjoint-multi}

\input{main/first-general-multi}

%% file: main/first-single.tex
\subsection{Single-Round Space Complexity of \textsc{FirstFail}}
\label{sec:single-round-first}

We establish the lower bound $S \,{\ge}\, \lceil \log(n \,{+}\, 1)\rceil$ in Thm.~\ref{thm:single-round-first-lower-bound} and propose a matching strategy using $S \,{\le}\, \lceil \log(n \,{+}\, 1)\rceil \,{+}\, 1$ in Thm.~\ref{thm:single-round-first-upper-bound}, so the single-round complexity of \textsc{FirstFail} is $S \,{=}\, \Theta(\log n)$ (Thm.~\ref{thm:single-round-first-tightness}).
Additionally, Thm.~\ref{thm:firstfail-no-trivial-update} shows that although \textsc{FirstFail} asymptotically shares its single-round space complexity with \textsc{ExistFail}, valid strategies unavoidably require more finesse.

\begin{theorem}
\label{thm:single-round-first-lower-bound}
    For any single-round strategy (as defined in Def.~\ref{def:stra}) that solves \emph{\textsc{FirstFail}}, it holds that $S \ge \lceil \log(n\,{+}\,1)\rceil$, with $n$ the number of assertions.
\end{theorem}
\begin{proof}
\vspace{-0.5ex}
\textsc{FirstFail} has $n\,{+}\,1$ possible outputs: either $\bot$ (no failure) or an index in $[n]$.
A single-round strategy measures $S$ ancilla qubits at termination, obtaining a bit string $R \,{\in}\, \{0,1\}^{S}$, and then deterministically maps $R$ to an output.
For correctness, distinct outputs must arise from distinct measurement outcomes.
Hence $2^S \,{\ge}\, n \,{+}\, 1$, which implies $S \,{\ge}\, \lceil \log(n \,{+}\, 1)\rceil$.
\end{proof}

\vspace{-0.5ex}
The above lower bound follows from an information-theoretic counting argument.
It also follows immediately from Thm.~\ref{thm:single-round-exist-lower-bound}, because \textsc{ExistFail} is no harder than \textsc{FirstFail}: any single-round strategy for \textsc{FirstFail} yields one for \textsc{ExistFail} by outputting true iff the result is not $\bot$.

\begin{theorem}
\label{thm:single-round-first-upper-bound}
There exists a single-round strategy solving \emph{\textsc{FirstFail}} with $S \le \lceil \log(n{+}1)\rceil{+}1$.
\end{theorem}
\begin{proof}
\vspace{-0.5ex}
We give the following strategy:

\begin{strategy}[\textsc{\emph{Single-Round Index Transposition}}]
\label{str:single-round-index-transposition}
\vspace{-0.5ex}
Let $\ell \,{\triangleq}\, \lceil \log(n \,{+}\, 1)\rceil$.
The instrumentation $\pazocal{I}$ uses $\ell \,{+}\, 1$ ancillas, partitioned into an $\ell$-qubit index register $\textsf{idx}$ and a one-qubit flag $\textsf{fail}$, both initialized to $\ket{0}$.
For each assertion $A_i$, the assertion-handling block used by $\pazocal{I}$ is:
{\setlength{\abovedisplayskip}{3pt}
\setlength{\belowdisplayskip}{2pt}
\setlength{\abovedisplayshortskip}{0pt}
\setlength{\belowdisplayshortskip}{0pt}
\begin{align*}
 B_i
=
C_{i \to \textsf{fail}};
\ V^i_{\mathrm{Tran}};
\ C_{i \to \textsf{fail}},
\end{align*}}

\noindent
where $V^i_{\mathrm{Tran}}$, conditioned on $\textsf{fail}=\ket{1}$, swaps the two basis states $| 0^\ell \rangle$ and $| \mathrm{bin}(i) \rangle$ and fixes all other basis states of $\textsf{idx}$.
Let $i^*$ be the first failing index.
At $A_{i^*}$, we have $\textsf{fail} \,{=}\, \ket{1}$ and $\textsf{idx} \,{=}\, |0^\ell \rangle$, so $V^{i^*}_{\mathrm{Tran}}$ updates $\textsf{idx}$ to $|\mathrm{bin}(i^*)\rangle$.
For any later failing assertion $A_j$ with $j\,{>}\, i^*$, the controlled transposition between $|0^\ell\rangle$ and $|\mathrm{bin}(j)\rangle$ leaves $|\mathrm{bin}(i^*)\rangle$ unchanged.
Thus at termination, $\textsf{idx}=|\mathrm{bin}(i^*)\rangle$, or $|0^\ell \rangle$ if none fail.
Measuring $\textsf{idx}$ and decoding the result in decimal therefore solves \textsc{FirstFail}.
\end{strategy}

\vspace{-0.5ex}
Strategy~\ref{str:single-round-index-transposition} uses $S=\lceil \log(n+1)\rceil+1$ ancillas.
This proves the theorem.
\end{proof}

\begin{example}
\vspace{-0.5ex}
Let $n=4$ and suppose the failure pattern is $F=(0,1,0,1)$, so the first failure is at $2$.
Then $\ell=\lceil \log 5\rceil=3$, so the strategy uses a $3$-qubit index register $\textsf{idx}$ and a one-qubit flag $\textsf{fail}$.
Passing assertions leave $\textsf{idx}$ unchanged, while a failing assertion applies the unitary $V^i_{\mathrm{Tran}}$ (where each $V^i_{\mathrm{Tran}}$ depends on the current assertion index $i$) on $\textsf{idx}$. Therefore, the ancilla evolution is:

\vspace{-2ex}
{\small
\setlength{\fboxsep}{0pt}
\setlength{\fboxrule}{0pt}
\[
\begin{aligned}
|000\rangle_{\textsf{idx}}\ket{0}_{\textsf{fail}} \hspace{-2pt}
&\xrightarrow{\cdots}
\fcolorbox{gray!100}{gray!7.5}{$ \xrightarrow[\text{unchanged}]{\text{A}_1~\text{pass}} \ket{000}\ket{0} $\,}  \hspace{-0.75pt} \xrightarrow{\cdots} \hspace{-0.75pt}  \fcolorbox{gray!100}{gray!7.5}{$ \xrightarrow[\text{write \textsf{fail}}]{C_{2 \to \textsf{fail}}} \ket{000}\ket{1} \xrightarrow[|000\rangle \leftrightarrow |010\rangle]{V^2_{\mathrm{Tran}}} \ket{010}\ket{1} \xrightarrow[\text{uncompute \textsf{fail}}]{C_{2 \to \textsf{fail}}} \ket{010}\ket{0}$\,} \,\,A_2 \text{ fail}
\\[0ex]
&\xrightarrow{\cdots}
\fcolorbox{gray!100}{gray!7.5}{$\displaystyle \xrightarrow[\text{unchanged}]{\text{A}_3~\text{pass}} \ket{010}\ket{0} $\,} \hspace{-0.75pt} \xrightarrow{\cdots} \hspace{-0.75pt} \fcolorbox{gray!100}{gray!7.5}{$\displaystyle \xrightarrow[\text{write \textsf{fail}}]{C_{4\to \textsf{fail}}} \ket{010}\ket{1} \xrightarrow[|000\rangle \leftrightarrow |100\rangle]{V^4_{\mathrm{Tran}}} \ket{010}\ket{1} \xrightarrow[\text{uncompute \textsf{fail}}] {C_{4\to \textsf{fail}}} \ket{010}\ket{0}$\,} \,\,A_4 \text{ fail}
\\[0ex]
&\xrightarrow{\cdots}
\text{execution ends} \xrightarrow{\text{measure \textsf{idx}}} 010 \xrightarrow{\textsf{Dec}} 2
\end{aligned}
\]}

\vspace{-0.5ex}
Thus, the final index register is $\ket{\mathrm{bin}(2)}=\ket{010}$, and measuring $\textsf{idx}$ yields the answer $2$.
\end{example}

\paragraph{Cost of Strategy~\ref{str:single-round-index-transposition}}
The strategy has $T\,{=}\,1$ and $S \,{=}\, \lceil \log(n+1)\rceil+1$.
The checker unitary is invoked twice per assertion, so $C \,{=}\, 2$.
Only $\textsf{idx}$ is measured at the end, hence $M=\lceil \log(n+1)\rceil$.
The ancilla-only update $V^i_{\mathrm{Tran}}$ can be implemented with ${O}(\log n)$ multi-controlled NOT gates without additional workspace, yielding the total non-checker gate cost $G = O(n\log n)$.
Concretely, the circuit for $V^i_{\mathrm{Tran}}$ is constructed as: fix a \emph{Gray path} from {\small $0^\ell$} to $z \,{=} \, \mathrm{bin}(i)\,{\in}\,\{0,1\}^{\ell}$ of length $\gamma \,{=}\, \mathrm{wt}(z)$ (Hamming weight), flipping 0-bits to 1 from \emph{right to left}: $0^{\ell}=v_0^{\scriptscriptstyle (i)} \rightarrow        v_1^{\scriptscriptstyle (i)} \rightarrow \cdots \rightarrow v_\gamma^{\scriptscriptstyle (i)}=z$.
Then define the palindromic gate sequence:
{\setlength{\abovedisplayskip}{4pt}
\setlength{\belowdisplayskip}{4pt}
\setlength{\abovedisplayshortskip}{0pt}
\setlength{\belowdisplayshortskip}{0pt}
\begin{align*}
V^i_{\mathrm{Tran}} \,=\, G^{\scriptscriptstyle (i)}_1 G^{\scriptscriptstyle (i)}_2 \cdots G^{\scriptscriptstyle (i)}_{\gamma-1} G^{\scriptscriptstyle (i)}_\gamma G^{\scriptscriptstyle (i)}_{\gamma-1} \cdots G^{\scriptscriptstyle (i)}_2 G^{\scriptscriptstyle (i)}_1,
\end{align*}
where each {\smash{\small $G^{\scriptscriptstyle (i)}_k$}} is a multi-controlled NOT on the unique target bit that differs between {\smash{\small $v^{\scriptscriptstyle(i)}_{k-1}$}} and {\smash{\small $v^{\scriptscriptstyle (i)}_{k}$}}, with controls on (i) $\mathsf{fail}\,{=}\,1$ and (ii) all non-target bits of $\mathsf{idx}$ fixed to match {\smash{\small $v^{\scriptscriptstyle (i)}_{k-1}$}}\footnote{{\scriptsize $G^{\scriptscriptstyle (i)}_k$} implements the controlled transposition between the basis states {\scriptsize ${v^{\scriptscriptstyle (i)}_{k-1}}$} and {\scriptsize ${v^{\scriptscriptstyle (i)}_{k}}$}, and acts as the identity elsewhere.}.
For the above example, the circuits realizing $V^1_{\mathrm{Tran}}, V^2_{\mathrm{Tran}}, V^3_{\mathrm{Tran}}$ and $V^4_{\mathrm{Tran}}$ are shown below, from left to right:

\begin{figure}[htbp]
\vspace{-2ex}
\centering
{\footnotesize
\begin{subfigure}[t]{0.21\textwidth}
\centering
\begin{quantikz}[baseline=(current bounding box.center),column sep=0.35cm,row sep={0.6cm,between origins}]
    \lstick{\rule[-3.5ex]{0pt}{7ex}\textsf{fail~}}
    & \ctrl{3}   \gategroup[wires=4,steps=1,
                   style={draw,thin,densely dashed,rounded corners,inner xsep=1pt,inner ysep=2pt},
                   background]{$G_{\scriptscriptstyle 1}^{\scriptscriptstyle (1)}$}
    & \qw \\
    \lstick{$\textsf{idx}_1$}  & \octrl{2}  & \qw \\
    \lstick{$\textsf{idx}_2$}  & \octrl{1}  & \qw \\
    \lstick{\rule[-3.5ex]{0pt}{7ex}$\textsf{idx}_3$}  & \targ{}    & \qw
\end{quantikz}
\end{subfigure}
\hfill
\begin{subfigure}[t]{0.21\textwidth}
\centering
\begin{quantikz}[baseline=(current bounding box.center),column sep=0.35cm,row sep={0.6cm,between origins}]
    \lstick{\rule[-3.5ex]{0pt}{7ex}\textsf{fail~}}
    & \ctrl{2}   \gategroup[wires=4,steps=1,
                   style={draw,thin,densely dashed,rounded corners,inner xsep=1pt,inner ysep=2pt},
                   background]{$G_{\scriptscriptstyle 1}^{\scriptscriptstyle (2)}$}
    & \qw \\
    \lstick{$\textsf{idx}_1$}  & \octrl{1}  & \qw \\
    \lstick{$\textsf{idx}_2$}  & \targ{}    & \qw \\
    \lstick{\rule[-3.5ex]{0pt}{7ex}$\textsf{idx}_3$}  & \octrl{-1} & \qw
\end{quantikz}
\end{subfigure}
\hfill
\hspace*{-2ex}
\begin{subfigure}[t]{0.32\textwidth}
\centering
\begin{quantikz}[baseline=(current bounding box.center),column sep=0.35cm,row sep={0.6cm,between origins}]
    \lstick{\rule[-3.5ex]{0pt}{7ex}\textsf{fail~}}
    & \ctrl{3}   \gategroup[wires=4,steps=1,
                   style={draw,thin,densely dashed,rounded corners,inner xsep=1pt,inner ysep=2pt},
                   background]{$G_{\scriptscriptstyle 1}^{\scriptscriptstyle (3)}$}
    & \ctrl{2}  \gategroup[wires=4,steps=1,
                   style={draw,thin,densely dashed,rounded corners,inner xsep=1pt,inner ysep=2pt},
                   background]{$G_{\scriptscriptstyle 2}^{\scriptscriptstyle (3)}$}
    & \ctrl{3}   \gategroup[wires=4,steps=1,
                   style={draw,thin,densely dashed,rounded corners,inner xsep=1pt,inner ysep=2pt},
                   background]{$G_{\scriptscriptstyle 1}^{\scriptscriptstyle (3)}$}
    & \qw \\
    \lstick{$\textsf{idx}_1$}  & \octrl{2}  & \octrl{1} & \octrl{2}  & \qw \\
    \lstick{$\textsf{idx}_2$}  & \octrl{1}  & \targ{}   & \octrl{1}  & \qw \\
    \lstick{\rule[-3.5ex]{0pt}{7ex}$\textsf{idx}_3$}  & \targ{}    & \ctrl{-1} & \targ{}    & \qw
\end{quantikz}
\end{subfigure}
\hfill
\begin{subfigure}[t]{0.21\textwidth}
\centering
\begin{quantikz}[baseline=(current bounding box.center),column sep=0.35cm,row sep={0.6cm,between origins}]
    \lstick{\rule[-3.5ex]{0pt}{7ex}\textsf{fail~}}
    & \ctrl{1}   \gategroup[wires=4,steps=1,
                   style={draw,thin,densely dashed,rounded corners,inner xsep=1pt,inner ysep=2pt},
                   background]{$G_{\scriptscriptstyle 1}^{\scriptscriptstyle (4)}$}
    & \qw \\
    \lstick{$\textsf{idx}_1$}  & \targ{}    & \qw \\
    \lstick{$\textsf{idx}_2$}  & \octrl{-1} & \qw \\
    \lstick{\rule[-3.5ex]{0pt}{7ex}$\textsf{idx}_3$}  & \octrl{-2} & \qw
\end{quantikz}%
\end{subfigure}}%
\end{figure}
\vspace*{-2ex}%

The total gates used over all $V^i_{\mathrm{Tran}}$ is computed as $ {G}\ = \textstyle \sum_{i=1}^{n}\bigl(2\,\ell(i)-1\bigr) $ where $\ell(i)=\mathrm{wt}\,\big(\mathrm{bin}(i)\big)$. A convenient worst-case upper bound is (as if all nonzero $\ell$-bit strings occurred in $\{\mathrm{bin}(i) \, | \, i \,{\in}\, [n]\}$):
\[
{G}\ \le\ \textstyle \sum_{i=1}^{\ell}\binom{\ell}{i}\,(2i-1)
        \,=\, (\ell-1) \cdot \,2^{\ell} + 1
        \,=\, \mathrm{O}(\ell \,{\cdot}\, 2^{\ell})
        \,=\, \mathrm{O}(n\log n).
\]

\begin{theorem}
\label{thm:single-round-first-tightness}
The single-round space complexity of \emph{\textsc{FirstFail}} satisfies $S \,{=}\, \Theta(\log n)$.
\end{theorem}
\begin{proof}
\vspace{-0.5ex}
Immediate by Thm.~\ref{thm:single-round-first-lower-bound} and Thm.~\ref{thm:single-round-first-upper-bound}.
\end{proof}

\paragraph{Hardness of \textsc{FirstFail}.}
Classically, \textsc{FirstFail} admits a one-line sequential algorithm: keep a single $O(\log n)$-bit register, and write the current index into it whenever an assertion fails while the register is still empty.
Mathematically, this conditional write is a many-to-one mapping: a failure at position $i$ sends both an empty register and one already holding $i$ to the same state, which is precisely what reversibility (a quantum-specific constraint) forbids.
The generic remedy, one history bit per assertion, would restore reversibility at $O(n)$ ancillas, essentially falling back to the single-round/$n$-ancillas baseline strategy, losing the logarithmic saving.

Strategy~\ref{str:single-round-index-transposition} dissolves the dilemma with a twist: replace the destructive write by the transposition that \emph{exchanges} $|0^{\ell}\rangle$ and $|\mathrm{bin}(i)\rangle$. The two maps agree everywhere a forward execution can reach: the first failure writes its index, and every later failure at position $i$ meets a register holding some $|\mathrm{bin}(i^*)\rangle$ with $i^* \neq i$, which the transposition fixes. The un-write branch $|\mathrm{bin}(i)\rangle \mapsto |0^{\ell}\rangle$, the very branch that makes the map a bijection, is simply never exercised. No history is kept, and uncomputing the \textsf{fail} flag after each use keeps the persistent memory to the index register alone.

Although \textsc{FirstFail} thus matches the single-round space complexity of \textsc{ExistFail}, attaining the bound demands more intricate ancilla logic. Our matching constructions pay in one of two ways: the ancilla update \emph{varies across indices} (the transpositions above), or it acts nontrivially \emph{even when an assertion passes} (Strategy~\ref{str:single-round-modulo-decrement}, an alternative upper bound in App.~\ref{app:firstfail-alternative}). The theorem below shows this dichotomy is inherent to single-round \textsc{FirstFail} strategies.

\vspace{-0.5ex}
\begin{theorem}
\label{thm:firstfail-no-trivial-update}
Let $\pazocal{S}$ be any single-round strategy (Def.~\ref{def:stra}) that solves \emph{\textsc{FirstFail}} over all $2^n$ deterministic failure patterns, and let $W^{1,i}_0, W^{1,i}_1$ (for $i \,{\in}\, [n]$) be the ancilla-only unitaries induced by its instrumentation (Lem.~\ref{lem:ancilla-evolution}). If the ancilla-processing logic is index-independent, i.e., $W^{1,i}_0 = W_0$ and $W^{1,i}_1 = W_1$ for all $i \,{\in}\, [n]$, then neither $W_0$ nor $W_1$ is the identity.
\end{theorem}

\begin{proof}
\vspace{-0.5ex}
Assume toward contradiction that at least one of $W_0$ or $W_1$ equals the identity. For a deterministic failure pattern $x \,{=}\, b_1 \cdots b_n \,{\in}\, \{0,1\}^n$, the final ancilla state is $W_{b_n} \cdots W_{b_1}|0^{S}\rangle$, where $S$ is the ancilla size of $\pazocal{S}$ (Lem.~\ref{lem:ancilla-evolution}).
Since the strategy always answers correctly and the decoder is a function of the ancilla readout, any two patterns with different \textsc{FirstFail} outputs must lead to different final ancilla states.

\emph{Case 1: $W_0 = I$.} Here, the final ancilla state depends only on the number of $1$'s in $x$. In particular, $x' = 10^{n-1}$ and $x'' = 0^{n-1}1$ both contain exactly one $1$, so they induce the same final ancilla state, but $\textsc{FirstFail}(x') = 1$ and $\textsc{FirstFail}(x'') = n$, a contradiction.

\emph{Case 2: $W_1 = I$.} Here, the final ancilla state depends only on the number of $0$'s in $x$. Again, $x' = 10^{n-1}$ and $x'' = 0^{n-1}1$ contain the same number of $0$'s, hence induce the same final ancilla state, while their \textsc{FirstFail} outputs differ. This is also a contradiction, so $W_0 \neq I$ and $W_1 \neq I$.
\end{proof}

%% file: main/first-disjoint-multi.tex
\subsection{Disjoint Multi-Round Time-Space Trade-Off for \textsc{FirstFail}}
\label{sec:disjoint-multi-round-first}

We next characterize the complexity of disjoint multi-round strategies for \textsc{FirstFail}.
As in the \textsc{ExistFail} case, we prove the lower bound $(2^S \,{-}\, 1) \,{\cdot}\, T \,{\ge}\, n$ (Thm.~\ref{thm:multi-round-disjoint-first-lower-bound}), and a matching upper bound $(2^S \,{-}\, 1) \,{\cdot}\, T \,{<}\, 11n$ (Thm.~\ref{thm:multi-round-disjoint-first-upper-bound}), so the complexity is $(2^S \,{-}\, 1) \,{\cdot}\, T \,{=}\, \Theta(n)$ (Thm.~\ref{thm:multi-round-disjoint-first-complexity}).

\begin{theorem}
\label{thm:multi-round-disjoint-first-lower-bound}
\vspace{-0.5ex}
For any disjoint multi-round strategy (Def.~\ref{def:multi-stra}) that solves \emph{\textsc{FirstFail}}, it holds that $(2^S-1)\cdot T \ge n$, with $n$ the number of assertions.
\end{theorem}
\begin{proof}
\vspace{-0.5ex}
Immediate from Thm.~\ref{thm:multi-round-disjoint-exist-lower-bound}, because any disjoint multi-round strategy for \textsc{FirstFail} yields one for \textsc{ExistFail} by post-processing its final answer, outputting true iff the result is not $\bot$.
\end{proof}

\begin{theorem}
\label{thm:multi-round-disjoint-first-upper-bound}
\vspace{-0.5ex}
For every $T \le n$, there exists a disjoint multi-round strategy solving \emph{\textsc{FirstFail}} with $(2^S-1)\cdot T < 11n$.
\end{theorem}
\begin{proof}
\vspace{-0.5ex}
We construct the following strategy:

\begin{strategy}[\textsc{\emph{Partitioned Index Transposition}}]
\vspace{-0.5ex}
\label{str:partitioned-index-transposition}
Partition the $n$ assertions into $T$ disjoint contiguous index sets $\mathcal{E}_1,\ldots,\mathcal{E}_T$, in program order, each of size at most $\lceil n/T\rceil$.
In round $t$, run the single-round index-transposition strategy (Strategy~\ref{str:single-round-index-transposition}) on the assertions in $\mathcal{E}_t$ only, using an index register of size $\lceil \log(\lceil n/T\rceil \,{+}\, 1)\rceil$ and one flag qubit $\textsf{fail}$.
Let $r_t$ be the measurement outcome of round $t$.
After all $T$ rounds, scan the rounds in order: if every $r_t \,{=}\, 0$, output $\bot$; otherwise, let $t^*$ be the first round with $r_{t^*} \,{\neq}\, 0$, decode from $r_{t^*}$ the index inside $\mathcal{E}_{t^*}$, and add the offset \smash{$\sum_{j=1}^{t^*-1} |\mathcal{E}_j|$}.
\end{strategy}

\vspace{-1ex}
Because each round uses $S\,{=}\,\lceil \log(\lceil n/T\rceil \,{+}\, 1)\rceil \,{+}\, 1$ ancillas, applying the same calculation as in Thm.~\ref{thm:multi-round-disjoint-exist-upper-bound} proves the theorem.
\end{proof}

\paragraph{Cost of Strategy~\ref{str:partitioned-index-transposition}}
\vspace{-0.5ex}
For $T \le n$, the strategy uses $T$ rounds and $S\,{=}\,\lceil \log(\lceil n/T\rceil\,{+}\,1)\rceil\,{+}\,1$
ancillas per round. $C \,{=}\, 2$, and $M \,{=}\, T \,{\cdot}\, \lceil \log(\lceil n/T\rceil \,{+}\, 1)\rceil$.
The single-round cost analysis of Strategy~\ref{str:single-round-index-transposition} gives round-$t$ gate cost $G_t={O}(|\mathcal{E}_t|\ell)$ where $\ell$ is the size of the index register of this round.
Therefore, $G=\sum_{t=1}^T G_t={O}(n\ell)={O}\left(n\log(\lceil n/T\rceil+1)\right)$.

\begin{theorem}
\label{thm:multi-round-disjoint-first-complexity}
For all $T \le n$, the disjoint multi-round complexity of \emph{\textsc{FirstFail}} is $S = \Theta\left(\log (1 + \frac{n}{T})\right)$.
\end{theorem}
\begin{proof}
\vspace{-0.5ex}
Immediate from
Thm.~\ref{thm:multi-round-disjoint-first-lower-bound} and Thm.~\ref{thm:multi-round-disjoint-first-upper-bound}, by following the same proof as Thm.~\ref{thm:multi-round-disjoint-exist-complexity}.
\end{proof}

%% file: main/first-general-multi.tex
\subsection{General Multi-Round Time-Space Trade-Off for \textsc{FirstFail}}
\label{sec:general-multi-round-first}

For general multi-round strategies, we prove that \textsc{FirstFail} still admits $S \,{=}\, \Theta\big(\log (1\,{+}\,\frac{n}{T})\big)$, the same complexity as in the disjoint setting --- unlike \textsc{ExistFail}, the ability to encode a failure pattern in multiple rounds brings no improvement.
The key result is the lower bound in Thm.~\ref{thm:multi-round-general-first-lower-bound}; the disjoint upper bound strategy carries over unchanged and yields tightness (Thm.~\ref{thm:multi-round-general-first-complexity}).

\begin{theorem}
\label{thm:multi-round-general-first-lower-bound}
For any multi-round strategy $\pazocal{S}$ (Def.~\ref{def:multi-stra}) that solves \emph{\textsc{FirstFail}}, it holds that $T \cdot 4^{S} \ge n/\big(2\log(1+\sqrt{2})\big)$, with $n$ the number of assertions.
\end{theorem}

\begin{proof}
\vspace{-0.5ex}
Such a strategy solves \textsc{FirstFail} over all $2^n$ deterministic failure patterns, so by Lemma~\ref{lem:reduction} it yields a $T$-round unitary transition system solving the corresponding task \textsc{First} (Def.~\ref{def:distinguish}), whose round-$t$ dimension is $2^{S_t}$, with $S_t$ the ancilla size used in round $t$. Lemma~\ref{lem:first-packing} below, applied with $d_t = 2^{S_t}$, gives
\smash{$\textstyle\sum_{t=1}^{T} 4^{S_t} \ge n/\big(2\log(1+\sqrt{2})\big)$}. Since $S_t \le S$ for every $t$, the left side is at most $T \cdot 4^{S}$, and the claim follows.
\end{proof}

\begin{lemma}
\label{lem:first-packing}
Let $(\pazocal{U}^{1}_n, \ldots, \pazocal{U}^{T}_n)$ be a $T$-round unitary transition system (Def.~\ref{def:unitary-transition-system}) that solves \emph{\textsc{First}} (Def.~\ref{def:distinguish}), and let $d_t$ denote the dimension of the Hilbert space of $~\pazocal{U}^{t}_n$. Then
\smash{$\textstyle\sum_{t=1}^{T} d_t^{\,2} \ge n/\big(2\log(1+\sqrt{2})\big)$}.
\end{lemma}

\begin{proof}[Proof Sketch]
\vspace{-0.5ex}
The full proof is given in App.~\ref{app:first-packing-proof}. For round $t$, write \smash{$\pazocal{T}^{t}_{x}$} for the \emph{input operator} of $x$, i.e., the product of the round-$t$ transition unitaries selected by the bits of $x$, so that the round-$t$ final state is $|\psi^{t}(x)\rangle = \pazocal{T}^{t}_{x}|\psi^{t}_0\rangle$ with $|\psi^{t}_0\rangle$ the initial state.

The key step is a pairwise separation property: for \emph{every} pair of distinct inputs $x \,{\neq}\, y$, some round $t$ satisfies \smash{$\Vert\pazocal{T}^{t}_{x} - \pazocal{T}^{t}_{y}\Vert_{\mathrm{op}} \ge \sqrt{2}$}, including pairs with $\textsc{First}(x) = \textsc{First}(y)$, whose final states themselves need not be orthogonal. Let $i$ be the first position where $x$ and $y$ differ, say $x_i = 0$ and $y_i = 1$.
For each round $t$, define the unit witness vector
\smash{$|v_t\rangle \,{\triangleq}\, \bigl(\pazocal{T}^{t}_{w}\bigr)^\dagger\,
\pazocal{T}^{t}_{0^{i-1}} |\psi^t_0\rangle$}, where
$w \,{\triangleq}\, x_1\cdots x_{i-1} = y_1\cdots y_{i-1}$ is the common prefix; the
prefix segment maps $|v_t\rangle$ to the state reached by an all-zero prefix, so
the two inputs act on $|v_t\rangle$ as the modified inputs
$\tilde{x} \,{\triangleq}\, 0^i x_{i+1}\cdots x_n$ and
$\tilde{y} \,{\triangleq}\, 0^{i-1} 1\, y_{i+1}\cdots y_n$:
\[
\pazocal{T}^t_x |v_t\rangle = |\psi^t(\tilde{x})\rangle, \qquad
\pazocal{T}^t_y |v_t\rangle = |\psi^t(\tilde{y})\rangle.
\]
Now $\textsc{First}(\tilde{y}) = i$ while $\textsc{First}(\tilde{x}) \in \{\bot\} \cup \{i{+}1,\dots,n\}$: the transported answers always differ, so Lemma~\ref{lem:answer-separation} yields a round $t$ in which $|\psi^t(\tilde{x})\rangle$ and $|\psi^t(\tilde{y})\rangle$ are orthogonal. At this round, $\pazocal{T}^t_x|v_t\rangle$ and $\pazocal{T}^t_y|v_t\rangle$ are unit vectors, since the segments are unitary and $|v_t\rangle$ is a unit vector; hence
{\setlength{\abovedisplayskip}{1pt}
\setlength{\belowdisplayskip}{3pt}
\setlength{\abovedisplayshortskip}{0pt}
\setlength{\belowdisplayshortskip}{0pt}
\begin{align*}
\big\Vert\big(\pazocal{T}^{t}_{x} - \pazocal{T}^{t}_{y}\big)|v_t\rangle\big\Vert^2
~=~ 1 + 1 - 2\,\mathrm{Re}\,\langle \pazocal{T}^{t}_{x} v_t \,|\, \pazocal{T}^{t}_{y} v_t\rangle
~=~ 2,
\end{align*}
}

\noindent
and therefore $\Vert\pazocal{T}^{t}_{x} - \pazocal{T}^{t}_{y}\Vert_{\mathrm{op}} \ge \big\Vert\big(\pazocal{T}^{t}_{x} - \pazocal{T}^{t}_{y}\big)|v_t\rangle\big\Vert = \sqrt{2}$.
Consequently, the $2^n$ tuples \smash{$(\pazocal{T}^{1}_{x}, \ldots, \pazocal{T}^{T}_{x})$} are pairwise $\sqrt{2}$-separated under the norm $\max_{t}\Vert\cdot\Vert_{\mathrm{op}}$, while all of them lie in the unit ball of a real vector space of dimension \smash{$2\sum_{t} d_t^{\,2}$}. A volumetric packing argument yields \smash{$2^n \le (1+\sqrt{2})^{2\sum_t d_t^2}$}, and taking logarithms proves the lemma.
\end{proof}

\begin{theorem}
\label{thm:multi-round-general-first-complexity}
\vspace{-0.5ex}
For all $T \le n$, the general multi-round complexity of \emph{\textsc{FirstFail}} is $S = \Theta\big(\log(1 + \frac{n}{T})\big)$.
\end{theorem}

\begin{proof}
\vspace{-0.5ex}
For the lower bound, Thm.~\ref{thm:multi-round-general-first-lower-bound} gives $T \cdot 4^{S} \ge n/\big(2\log(1+\sqrt{2})\big)$, i.e., $S \ge \frac{1}{2}\log\frac{n}{T} - O(1)$, which is $\Omega\big(\log(1+\frac{n}{T})\big)$ once $\tfrac{n}{T}$ exceeds a suitable constant; when $\frac{n}{T} = O(1)$, it suffices that $S \ge 1$ (an $S = 0$ strategy has a constant transcript and cannot distinguish two patterns with different outputs). For the upper bound, Strategy~\ref{str:partitioned-index-transposition} is disjoint, hence in particular general multi-round; by Thm.~\ref{thm:multi-round-disjoint-first-upper-bound} it satisfies $(2^S - 1)\cdot T < 11n$, i.e., $S = O\big(\log(1 + \frac{n}{T})\big)$.
\end{proof}

\paragraph{Hardness of \textsc{FirstFail} under General Multi-Round Setting.}
\vspace{-0.5ex}
The result stands in sharp contrast to \textsc{ExistFail}, where encoding the same failure pattern in multiple rounds beats the disjoint trade-off (Sec.~\ref{sec:general-multi-round-exist}). The obstruction is the index-revealing answer of \textsc{FirstFail}: transporting two distinct inputs onto a common witness state makes their answers differ at the first disagreement position, so every pair of inputs must be separated by some round, leaving no room for the multi-encoding compression (e.g., counting failures modulo a chosen integer per round) that benefits \textsc{ExistFail}.

%% file: main/list.tex
\section{Complexity of the \textsc{ListAll} Task}
\label{sec:listAll}

This section studies the \textsc{ListAll} task over $n$ assertions, which performs full reporting of the $n$-bit failure pattern.
Unlike \textsc{ExistFail} and \textsc{FirstFail}, the time--space complexity of \textsc{ListAll} is governed purely by the total amount of information that must be extracted from the program.

\begin{theorem}
\label{thm:multi-round-list-lb}
For any $T$-round strategy that solves \emph{\textsc{ListAll}}, it holds that $S\cdot T \ge n$.
\end{theorem}

\begin{proof}\vspace{-0.5ex}
An output of \textsc{ListAll} is an arbitrary subset of $[n]$, hence there are $2^n$ possible outputs.
The strategy measures at most $S$ ancilla qubits per round, therefore obtains at most $2^{ST}$ possible outcome transcripts.
Because distinct outputs must correspond to disjoint transcripts, $2^{ST}\ge 2^n$.
Taking logarithms, $S\cdot T \ge n$ follows.
\end{proof}

\begin{theorem}
\label{thm:multi-round-list-ub}
For every $T \le n$, there exists a $T$-round strategy solving \emph{\textsc{ListAll}} with $S=\lceil n/T\rceil$. In particular, $S\cdot T \le n+T$.
\end{theorem}

\begin{proof}
\vspace{-0.5ex}
We construct the following strategy:

\begin{strategy}[\textsc{\emph{Partitioned Full Reporting}}]
\vspace{-0.5ex}
\label{str:multi-round-listAll}
Partition the $n$ assertions into $T$ disjoint contiguous index sets $\mathcal{E}_1, \ldots, \mathcal{E}_T$, each of size at most $\lceil n/T\rceil$.
In round $t$, enable only the assertions in $\mathcal{E}_t$ and allocate one ancilla per enabled assertion to record the full local failure pattern.
After all $T$ rounds, concatenate the $T$ recorded bit strings to obtain the full failure pattern on all $n$ assertions.
\end{strategy}

\vspace{-0.5ex}
Strategy~\ref{str:multi-round-listAll} uses one ancilla per enabled assertion in each round, so its per-round space cost is $S=\lceil n/T\rceil$.
Hence,
$S\cdot T = T\cdot \lceil n/T\rceil \le n+T$.
\end{proof}

\paragraph{Cost of Strategy~\ref{str:multi-round-listAll}.}
\vspace{-0.5ex}
The strategy uses $T$ rounds and $S\,{=}\,\lceil n/T\rceil$ per round.
Each assertion is enabled in one round, checked once, and contributes exactly one measurement.
Thus, $C\,{=}\, 1$ and $M\,{=}\,n$.
No additional ancilla processing is added, hence $G=0$. As two extremes, setting $T=1$ yields the \emph{single-round/n-ancillas} strategy, and setting $T=n$ yields the \emph{single-ancilla/n-rounds} strategy.

\begin{corollary}
\label{cor:single-round-list-complexity}
The single-round space complexity of \emph{\textsc{ListAll}} satisfies $S=\Theta(n)$.
\end{corollary}
\begin{proof}
\vspace{-0.5ex}
Immediate from Thm.~\ref{thm:multi-round-list-lb} and Thm.~\ref{thm:multi-round-list-ub} by setting $T=1$.
\end{proof}

\begin{theorem}
\label{thm:multi-round-list-complexity}
For all $T \le n$, the multi-round time--space complexity of \emph{\textsc{ListAll}} is $S = \Theta(\frac{n}{T})$.
\end{theorem}
\begin{proof}
\vspace{-0.5ex}
The lower bound is Thm.~\ref{thm:multi-round-list-lb}.
For the upper bound, Thm.~\ref{thm:multi-round-list-ub} gives $S\cdot T \le n+T \le 2n$ for every $T\le n$.
Therefore $S\cdot T = \Theta(n)$, and hence $S = \Theta(\frac{n}{T})$.
\end{proof}

%% file: main/landscape.tex
\section{Case Study and Landscape of Assertion Checking Strategies}
\label{sec:landscape}

In this section, we shift our analysis from asymptotic complexity to the concrete landscape of assertion checking strategies in practice.
Using a case study on assertion checking for Grover's quantum search algorithm, we summarize the families of checking strategies developed in this work, the trade-off profiles they realize, and the distinct design points they occupy.

\input{main/landscape-strategies}
\input{main/landscape-case}

%% file: main/landscape-strategies.tex
\subsection{Trade-Off Landscape of Proposed Strategies}
\label{sec:trade-off}
\definecolor{DarkGreen}{RGB}{0,102,34}%
Table~\ref{tab:trade-off} summarizes the strategy families developed in Sec.~\ref{sec:exist}--\ref{sec:listAll}.
The \textcolor{DarkGreen}{green} rows report our two primary complexity measures: the number of rounds $T$ and the per-round ancilla qubit cost~$S$.
The \textcolor{RoyalBlue}{blue} rows additionally report the operational costs over all rounds: the numbers of single-qubit measurements $M$, per-assertion checker calls $C$, and non-checker gate cost $G$.

\paragraph{Trade-off Profiles}
\vspace{-0.5ex}
The strategies realize each task's basic time--space trade-off under disjoint rounds: modulo increment and index transposition attain $(2^S - 1) \,{\cdot}\, T \,{=}\, \Theta(n)$ for \textsc{ExistFail} and \textsc{FirstFail}, and full reporting attains the optimal $S \,{\cdot}\, T \,{=}\, \Theta(n)$ for \textsc{ListAll}, subsuming the two baselines (Strategy~\ref{str:n-round-one-anc} and ~\ref{str:one-round-n-anc}) as its $T = 1$ and $T = n$ endpoints. Beyond disjoint rounds, the two partial-information tasks behave differently: overlapping rounds let LCM fingerprinting push \textsc{ExistFail} to the optimal product $S \,{\cdot}\, T \,{=}\, \Theta(\log n)$, whereas for \textsc{FirstFail} the possibility of overlap brings no asymptotic improvement.

\begin{table}[t]
\centering
\begin{threeparttable}
\caption{Summary of time--space trade-offs provided by the proposed assertion checking strategies.\vspace{-1.5ex}}
\label{tab:trade-off}
{\footnotesize
\renewcommand{\arraystretch}{1.1}
\setlength{\tabcolsep}{10pt}
\definecolor{mygreen}{rgb}{0.913,0.9535,0.8925}%
\begin{tabular}{lcccc}
\toprule
 & \multicolumn{2}{c}{\textsc{ExistFail}} & \textsc{FirstFail} & \textsc{ListAll} \\
 \cmidrule(lr){2-3} \cmidrule(lr){4-4} \cmidrule(lr){5-5}
 & \strategyhead{Modulo Increment}{str:partitioned-modulo-increment}
 & \strategyhead{LCM Fingerprint}{str:lcm-fingerprint}
 & \strategyhead{Index Transposition}{str:partitioned-index-transposition}
 & \strategyhead{Full Reporting}{str:multi-round-listAll} \\
\midrule
\\[-3.5ex]

\rowcolor{mygreen!50} \multicolumn{1}{l}{$\,\textcolor{DarkGreen}{T}$}
& any $T \le n$
& $O(\log n / \log \log n)$
& any $T \le n$
& any $T \le n$ \\

\rowcolor{mygreen!50} \multicolumn{1}{l}{$\,\textcolor{DarkGreen}{S}$}
& $\lceil \log (\lceil n/T \rceil \,{+}\, 1) \rceil \,{+}\, 1$
& $ \max_{t} \lceil \log \mu_t \rceil  \,{+}\, 1$
& $\lceil \log(\lceil n/T\rceil \,{+}\, 1) \rceil \,{+}\, 1$
& $\lceil n/T \rceil$ \\
\\[-2.5ex]

\rowcolor{CornflowerBlue!5} \multicolumn{1}{l}{$\textcolor{RoyalBlue}{M}$}
& $T$
& $T$
& $T \cdot \lceil \log(\lceil n/T\rceil \,{+}\, 1) \rceil$
& $n$ \\

\rowcolor{CornflowerBlue!5} \multicolumn{1}{l}{$\textcolor{RoyalBlue}{C}$}
& $2$
& $2T$
& $2$
& $1$ \\

\rowcolor{CornflowerBlue!5} \multicolumn{1}{l}{$\textcolor{RoyalBlue}{G}$}
& O$(n\log(\lceil n/T\rceil \,{+}\, 1))$
& O$(n \textstyle \sum_t \lceil \log \mu_t  \rceil )$
& O$(n\log(\lceil n/T\rceil \,{+}\, 1))$
& 0 \\
\\[-3ex]
\bottomrule
\end{tabular}}
\begin{tablenotes}[flushleft]
\footnotesize
\vspace{0.5ex}
\item In Strategy~\ref{str:lcm-fingerprint}, $\mu_1,\ldots,\mu_T$ are $T$ integers such that $\forall t ,\mu_t\ge 2$, and $\mathrm{lcm}(\mu_1,\ldots,\mu_T) > n$.
\end{tablenotes}
\end{threeparttable}
\vspace{-1ex}
\end{table}

\paragraph{Strategy vs. End-to-End Costs}
\vspace{-0.5ex}
Table~\ref{tab:trade-off} reports only the cost of the checking strategies, i.e., that of coordination alone.
$S$ counts only the ancillas used to route, aggregate, and read out assertion outcomes, and $G$ counts only non-checker gates.
We deliberately separate these costs from the bare program and the internal implementation of each checker unitary to emphasize their independence from the particular program, assertion predicates, and checker constructions.
By contrast, we next analyze the total end-to-end costs of instantiating these strategies on an explicit program.

%% file: main/landscape-case.tex
\subsection{Case Study: Grover's Algorithm}
\label{sec:case}

We now present a case study that instantiates our assertion-checking strategies on Grover's search algorithm~\citep{grover}, compiles them into executable circuits, and reports the resulting resource costs.

\paragraph{Implementation}
\vspace{-0.5ex}
We implemented a Qiskit library providing the four strategy families of Table~\ref{tab:trade-off} as uniform program transformations.
Given any program and checker circuits, our implementation produces a static per-round circuit, which uses fresh ancillas and terminal measurement only.
Our test suite validates all program examples, including those in Secs.~\ref{sec:exist}--\ref{sec:listAll}, in end-to-end simulation on smaller instances, and counts gates in the largest instances where simulation is intractable.

\paragraph{Program and Assertions}
\vspace{-0.5ex}
We analyze Grover's algorithm as expressed in pseudocode below, with the understanding that for analysis the loop will be statically unrolled as specified in Def.~\ref{def:prog}.

\vspace{0.5ex}
Let $q \triangleq (q_1,\ldots,q_m)$ be the $m$-qubit search register and let $\alpha$ be the one-qubit oracle ancilla.
Let $f:\{0,1\}^m \to \{0,1\}$ specify the marked items, and let $O_f$ be the standard bit-oracle $O_f\ket{x}\ket{y}=\ket{x}\ket{y\oplus f(x)}$.
We assume that there exists one marked string $w^* \,{\in}\, \{0,1\}^m$ such that $f(w^*)=1$.
Let $\ket{s}\triangleq \ket{+}^{\otimes m}$ be the uniform superposition, and let $U_{\text{diff}}\triangleq 2\ketbra{s}{s}-I$ be the Grover diffusion operator. The standard Grover iterate is one oracle call followed by one application of $U_{\text{diff}}$:

{
\vspace{-2ex}
\small
\setlength{\abovedisplayskip}{0pt}
\setlength{\belowdisplayskip}{3pt}
\setlength{\abovedisplayshortskip}{0pt}
\setlength{\belowdisplayshortskip}{4pt}
\begin{alignat*}{2}
\quad &\textcolor{violet}{\text{input } m = 12} && \text{\textsf{for} } i := 1 \text{ \textsf{to} } itr \text{ \textsf{do} } \\
&\smash{itr = \lfloor \pi/4 \cdot 2^{m/2} \rfloor} ; && \quad (q_1,\ldots,q_m,\alpha):= O_f(q_1,\ldots,q_m,\alpha); \\
&(q_1,\ldots, q_m) = H^{\otimes m}(q_1,\ldots,q_m); && \quad \textsf{assert}(\alpha;~\ketbra{-}{-}); \textcolor{violet}{\text{\small $\ // A_{2i+1}$}} \\
&\textsf{assert}(q_1,\ldots,q_m;~| s\rangle  \langle s|); \textcolor{violet}{\text{\small{$\ // A_1$}}} \quad && \quad (q_1,\ldots,q_m) := U_{\text{diff}}(q_1,\ldots,q_m); \\
&\alpha := X (\alpha); ~~ \alpha := H (\alpha); && \quad \textsf{assert}(q_1,\ldots,q_m;~P_{\text{sym}}); \textcolor{violet}{\text{\small $\ //A_{2i+2}$}} \\
&\textsf{assert}(\alpha;~\ketbra{-}{-}); \textcolor{violet}{\text{\small{$\ // A_2$}}} && \text{\textsf{measure}}(q_1,\ldots,q_m);
\end{alignat*}
}

We take $m = 12$, so the search space has size $2^m = 4096$, and use $itr = \lfloor \pi/4 \cdot 2^{m/2} \rfloor = 50$ iterations. After unrolling the loop, the program contains $2 + 2\,itr = 102$ assertions $A_1, \ldots, A_{102}$: the assertion after the $i$-th oracle call is $A_{2i+1}$, and the assertion after the following diffusion step is $A_{2i+2}$.

The assertions capture essential program invariants.
Assertion $A_1$ checks correct initialization of the search register, namely that after applying $H^{\otimes m}$ the state is $|s\rangle$. Assertion $A_2$ checks correct initialization of the oracle ancilla in $|-\rangle$, which is required for phase kickback, and each $A_{2i+1}$ checks that the ancilla remains in $|-\rangle$ after the oracle call: when initialized in $|-\rangle$, the bit-oracle satisfies $O_f(|x\rangle|-\rangle) = (-1)^{f(x)}|x\rangle|-\rangle$, so the ancilla is unchanged and only contributes phase.
Assertions $A_{2i+2}$ check a symmetry invariant: all unmarked basis states carry a common amplitude. The assertion spot-checks the $2^k$ unmarked strings sharing a fixed prefix $c \in \{0,1\}^{m-k}$ (any $c$ differing from the prefix of $w^*$): their amplitudes are equal when, conditioned on the first $m-k$ qubits being $|c\rangle$, the remaining $k$ qubits are in $|+\rangle{}^{\otimes k}$, i.e., the projector $P_{\mathrm{sym}} \triangleq (I - |c\rangle\langle c|) \otimes I
  + |c\rangle\langle c| \otimes (|{+}\rangle\langle{+}|){}^{\otimes k}$.  We set $k=4$.
This assertion passes on every state in the Grover subspace and fires on
symmetry-breaking faults, e.g., a mis-implemented diffusion operator.

All of the above assertions can be realized as checker unitaries by adapting the \emph{logical-OR circuit} of \citet{ApproximateAssertion2021} (App.~\ref{app:A:runtime}), taking $O(m)$ gates for $A_1$ and $A_{2i+2}$, and $O(1)$ gates for the others.

\paragraph{Test Instances}
\vspace{-0.5ex}
We instantiate the oracle $O_f$ with two problems of different gate costs: (a)
A \emph{3-SAT oracle} marks the unique satisfying assignment of a 12-variable 3-SAT formula; and (b) A \emph{preimage oracle} marks the unique preimage
\smash{$w^* = g^{-1}(t)$} of a public target string $t$ under a fixed shallow reversible permutation $g$.
By keeping the search size, the iteration count, and the assertions fixed, the checker-to-program ratio, i.e., the total gate count of the checker unitaries (with each assertion counted once) as a fraction of the bare execution, becomes a controlled variable of the study.
This ratio governs the price in gates that each strategy pays to save qubits in each program.

\paragraph{Compilation}
\vspace{-0.5ex}
We use Qiskit to compile each instrumented Grover's instance to the Clifford + $R_z(\pi/4)$ gate set, the standard universal gate set for fault-tolerant quantum computers. We report their counts separately and disable rewriting passes that would confound our analysis. One work qubit, shared across all multi-controlled gate decompositions, is included in the qubit counts.

We perform only one optimization to remove unnecessary overhead. Each round of a partitioned strategy ends right after the last assertion of its enabled block: the remaining program gates act on the program register alone and cannot change the outcome distribution of the measured ancillas.
We truncate them, reducing total program cost from $T$ full executions to roughly $(T{+}1)/2$.

\begin{table}[t] \caption{Resource costs of the instrumented Grover's algorithm ($m = 12$, $n = 102$) for both oracles. Unlike the strategy-level $S$ and $G$ in Table~\ref{tab:trade-off} that are independent of the program and checker circuits, the columns here report full end-to-end costs of the compiled circuits. The total measurements $M_{\mathrm{tot}} = M + T \cdot m$ where $M$ is from Table~\ref{tab:trade-off}.
All percentages are relative to single-round Full-Reporting, i.e., the $T=1,S=n$ \textsc{ListAll} baseline.}
\label{tab:grover}
\vspace*{-2ex}%
\centering
\footnotesize
\setlength{\tabcolsep}{2.75pt}
\begin{tabular}{ccc| rrr| rrr| r}
\toprule
& & & \multicolumn{3}{c|}{3-SAT oracle} & \multicolumn{3}{c|}{Preimage oracle} & \\ \cmidrule(lr){4-6} \cmidrule(lr){7-9} Task & Strategy & $T$ & Qubits & Clifford & $R_z(\pi/4)$ & Qubits & Clifford & $R_z(\pi/4)$ & $M_{\mathrm{tot}}\hspace{-0.5ex}$ \\ \midrule \multirow{2}{*}{\shortstack{\textsc{List}\\\textsc{All}}} & \multirow{2}{*}{\shortstack{\textsf{Full-}\\\textsf{Reporting}}} & 1 & 122 ($0\%$) & 360.1k ($0\%$) & 269.4k ($0\%$) & 116 ($0\%$) & 38.6k ($0\%$) & 23.9k ($0\%$) & 114 \\ & & 2 & 71 ($-42\%$) & 533.3k ($+48\%$) & 400.4k ($+49\%$) & 65 ($-44\%$) & 51.2k ($+33\%$) & 32.2k ($+35\%$) & 126 \\
\midrule
\multirow{3}{*}[-0.8ex]{\shortstack{\textsc{Exist}\\\textsc{Fail}}} & \multirow{2}{*}{\shortstack{\textsf{Modulo-}\\\textsf{Increment}}} & 1 & 28 ($-77\%$) & 396.5k ($+10\%$) & 293.1k ($+9\%$) & 22 ($-81\%$) & 74.8k ($+94\%$) & 47.6k ($+99\%$) & 13 \\ & & 2 & 27 ($-78\%$) & 563.0k ($+56\%$) & 419.4k ($+56\%$) & 21 ($-82\%$) & 80.7k ($+109\%$) & 51.2k ($+114\%$) & 26 \\ \cmidrule[0.25pt](lr){2-10} & \textsf{LCM}$(8,13)$ & 2 & 25 ($-80\%$) & 774.3k ($+115\%$) & 571.8k ($+112\%$) & 19 ($-84\%$) & 130.9k ($+239\%$) & 80.8k ($+238\%$) & 26 \\
\midrule
\multirow{2}{*}{\shortstack{\textsc{First}\\\textsc{Fail}}} & \multirow{2}{*}{\shortstack{\textsf{Index-}\\\textsf{Transposition}}} & 1 & 28 ($-77\%$) & 415.9k ($+15\%$) & 302.5k ($+12\%$) & 22 ($-81\%$) & 94.2k ($+144\%$) & 57.0k ($+139\%$) & 19 \\ & & 2 & 27 ($-78\%$) & 575.7k ($+60\%$) & 425.5k ($+58\%$) & 21 ($-82\%$) & 93.4k ($+142\%$) & 57.2k ($+139\%$) & 36 \\
\bottomrule
\end{tabular}
\end{table}

{
\clubpenalty=0
\widowpenalty=0
\paragraph{Results}
\vspace{-0.5ex}
Table~\ref{tab:grover} reports the total qubits and gates for the instrumented Grover's algorithm for each strategy and each oracle.
In addition to the strategy-level costs in Table~\ref{tab:trade-off}, these totals include the bare Grover circuit and the checker circuits and ancillas they use.
We emphasize that this table should be read not as implying any universally dominant strategy but as a concrete exploration of the design space for quantum program testing.
Indeed, we make the following observations:
\begin{itemize}[topsep=2pt,leftmargin=15pt,itemsep=2pt]
\item \emph{We witness the space savings.}
For both 3-SAT and preimage, \textsc{ExistFail} and \textsc{FirstFail} require $77\%$ to $84\%$ fewer qubits than \textsc{ListAll}, reinforcing the asymptotics of Table~\ref{tab:complexity}. At $n = 102$, every logarithmic-size register fits in at most $8$ qubits, compared to $102$ qubits for \textsc{ListAll}.

\item \emph{Their price in gates is relatively cheaper for more complex programs.}
Because the cost of assertion checking is amortized against the bare program, the relative gate overhead of assertion checking shrinks as the program grows.
For the simpler preimage oracle, assertion checking takes $48.6\%$ of the overall gate budget, whereas it takes only $3.4\%$ for the more complex 3-SAT oracle.
Correspondingly, the $R_z(\pi/4)$ gates overhead introduced by single-round \textsc{ExistFail} is $+99\%$ for the preimage oracle but only $+9\%$ for 3-SAT\@.

\item \emph{More rounds buy less space at task-dependent rates.}
Moving from $T = 1$ to $T = 2$ nearly halves qubits for \textsc{ListAll} ($116 \to 65$, $122 \to 71$), but saves only one qubit for \textsc{ExistFail} and \textsc{FirstFail} ($22 \to 21$, $28 \to 27$).
This result is consistent with asymptotics: $S = \Theta(n/T)$ shrinks multiplicatively with greater $T$, whereas $S = \Theta(\log(1 + n/T))$ shrinks by about one qubit per doubling of $T$.
One round is therefore already near optimal for partial-information tasks.
\end{itemize}
}

\paragraph{Summary}
\vspace{-1ex}
Our empirical results suggest that the optimal checking strategy depends both on the information that a developer desires to learn and on the hardware resource that is the bottleneck.
When qubits are scarce, the single-round partial-information strategies are a sweet spot whose gate premium shrinks as the program grows more complex.
Full reporting uses the fewest gates when qubits are plentiful, while also potentially benefiting from multi-round strategies.

%% file: main/rela.tex
\section{Related Work}

\paragraph{Quantum Runtime Assertions}
Prior research on runtime assertions for quantum programs has extensively studied the expressiveness, realization, and optimization of programs containing assertions.
At the formal level, statistical assertions~\cite{statisticalAssertions2019}, runtime assertion circuits~\cite{RuntimeAssertionCircuts2020}, approximate assertions~\cite{ApproximateAssertion2021}, and projection-based assertions~\cite{ProjectionBased2020} progressively enlarge the class of predicates that can be checked, with projection-based assertions giving the most general assertion predicate of subspace membership.
At the practical level, recent work also studies how to evaluate assertions on noisy devices, e.g., by slicing a program into smaller pieces~\cite{assertionSlicing}, or how to relocate and refine assertions to improve error localization~\cite{moveAssertionsAround}.
Prior work has thus thoroughly studied the cost of realizing individual assertions, whether via a circuit instantiating the checker unitary~\cite{RuntimeAssertionCircuts2020,ApproximateAssertion2021} or via a projective measurement~\cite{ProjectionBased2020} that can be transformed into a checker unitary (App.~\ref{app:A:Proq}).

We study a complementary source of cost: coordinating the checking of many assertions in a quantum program, under a physically motivated machine model where qubits and measurements are costly or limited.
Existing approaches either
(i) destructively measure program qubits and repeat program execution for each assertion~\cite{statisticalAssertions2019,assertionSlicing},
(ii) assume mid-circuit measurements~\cite{ProjectionBased2020}, or
(iii) effectively realize the single-round/$n$-ancillas baseline strategy by allocating one ancilla per assertion~\cite{RuntimeAssertionCircuts2020,ApproximateAssertion2021}.
Our results are orthogonal to the means of realizing any individual predicate, and instead establish fundamental costs that are applicable to all of the above schemes.

\paragraph{Quantum Program Testing and Analysis}
\vspace{-1ex}
Beyond runtime assertions, researchers have adapted a broad spectrum of software testing methodologies to quantum programs, ranging from test input generation~\cite{CombinationalTesting2021,MultiSubroutineTesting2024,SearchBasedTesting2022,Fuzz2021,propertyBasedTesting2020} to adequacy criteria~\cite{QUantumInpuTOutput2021,cretira2024,mutation2021,mutation2022a,mutation2022b} and statistical output analysis~\cite{fpaaTesting}.
Other work studies quantum programming from complementary perspectives, spanning programming language semantics~\cite{communicatingQuantumProcesses,Qunity,easyAsPi}, the design of logics and tools for formal verification~\cite{FloydHoareQuantum,relationalProofs,CoqQ,RapunSL,verifyCircuitWithTreeAutomata,pushBottonVerification}, and the estimation and optimization of circuit resources~\cite{typeBasedResourceEstimationInCircuitLanguage,localOptimization}.

\paragraph{Complexity}
\vspace{-1ex}
Resource trade-offs are an established theme in classical and quantum complexity theory, studied through branching programs~\citep{Wegener2000, BorodinCook1982,AblayevGKMP2005, SauerhoffSieling2005} and pebble games~\citep{Bennett1989}.
Similar to our unitary transition system (Def.~\ref{def:unitary-transition-system}) is the nondeterministic unitary ordered binary decision diagram (NUOBDD)~\cite{GainutdinovaYakaryilmaz2017}, a read-once branching program that evolves unitarily and measures at termination.
A single-round transition system solving \textsc{Exist} (Def.~\ref{def:distinguish}) induces an NUOBDD computing $\mathrm{AND}_n$,

\noindent
and Lemma~\ref{thm:unitary-dim-lb} agrees with a known NUOBDD width bound for $\mathrm{AND}_n$ \citep[Thm.~7]{GainutdinovaYakaryilmaz2017}.

Whereas the prior proof grows a linearly independent set by induction, our proof exhibits $n+1$ pairwise orthogonal final states, which is the key to our explicit multi-round strategies.
Moreover, whereas an NUOBDD accepts or rejects through a fixed accepting subspace, our transition systems apply a classical decoder to the measured outcomes, supporting exact, multi-valued answers.

More broadly, the complexity of program analysis tasks has been characterized in many settings, spanning concurrency reasoning~\cite{testingMessagePassingComplexity,complexityOfPredictingAtomicityViolations,optimalConsistencyChecking}, verification~\cite{verificationThreadPoolsComplexity}, optimization~\cite{complexityOfFunctionMerging}, as well as quantum program cost analysis~\cite{hardnessOfAnalyzingQuantitatively}; our results add the checking of quantum assertions to this family.

%% file: main/imply.tex
\section{Future Directions}
\label{sec:implication}

\paragraph{Mid-circuit Measurement}
Our main results explore the complexity of assertion checking in the terminal-measurement model.
We hope that as hardware platforms continue to mature, the costs of mid-circuit measurement will become more uniform and suitable for analysis.
In the near term, recent work~\cite{IBM-DynamicCircuits2025} to support mid-circuit measurements continues to distinguish their cost from terminal measurements, and their availability and reliability remains inconsistent today~\cite{midCircuitMeasuringErrorRates2025}.

As mid-circuit measurement becomes more prevalent and reliable, precisely modeling its cost would enable richer analysis.
One step is to count the number of measurements $M$ alongside $S$.
In this model, Sec.~\ref{sec:midcircuit-bridge} shows that our existing bounds directly apply to $M + S$.
That said, characterizing the tight trade-off between $S$ and $M$ individually remains an important open direction.

\paragraph{Adaptivity}
\vspace{-0.5ex}
The multi-round strategies in this work are not adaptive, in that the instrumentation of each round is fixed in advance. A natural extension is to permit later rounds to depend on the classical measurement outcomes of earlier ones. For example, an adaptive strategy for \textsc{FirstFail} could enable only the first half of the $n$ assertions, check them with the single-round \textsc{ExistFail} construction, and recurse on the relevant half, analogous to classical binary-search debugging~\cite{binarySearchDebugging}.
This strategy uses $S=O(\log n)$ ancillas and $T=O(\log n)$ rounds, so it does not improve on our non-adaptive strategy, which achieves $S=\Theta(\log n)$ with $T=1$. This initial answer suggests that adaptivity may have limited impact on the leading asymptotic trade-offs in the terminal model. However, it may still improve the constant factors or the operational costs of assertions.

The mid-circuit setting sharpens this question. As shown in Sec.~\ref{sec:midcircuit-bridge}, letting later rounds depend on earlier outcomes is equivalent to classical feedforward within a single execution, conditioning later segments on earlier measurements~\cite{ibmClassicalFeedforward2025}. In the limit, the multi-round terminal-measurement model with adaptivity converges to exactly the mid-circuit measurement model with feedforward.
Characterizing the precise power of adaptive strategies thus remains an interesting direction.

\section{Conclusion}
Runtime assertions are a key tool for testing and debugging quantum programs, yet one cannot check multiple assertions in a quantum program in the same way as for a classical program. As quantum programs grow in scale and complexity, this gap becomes increasingly consequential.

To narrow this gap, this paper initiates the study of a fundamental yet previously unformalized question: what is the time and space complexity of extracting outcome information from multiple runtime assertions in a quantum program?
Our results reveal that under a physically motivated machine model, the complexity of multiple-assertion checking varies with the information target and is shaped largely by the reversibility constraints of unitary quantum computation.

More broadly, our results indicate that the analysis of quantum programs can be viewed as an information-extraction problem in superposition under limited measurement.
Runtime assertions are one concrete instance, but similar problems arise in settings such as quantum error correction and circuit verification, where explicit resource bounds may lead to further connections between testing, verification, and physically constrained quantum architectures.
Studying these problems from the perspective of a programmer can reveal new insights in the theory of computation.

\clearpage

%% file: appendix/applying.tex
\clearpage
\section{Relating Our Framework to Prior Quantum Assertion Schemes}
\label{app:A}
This appendix shows how our framework relates to and may be applicable to the assertion schemes in prior work, in the following two complementary ways:
\begin{enumerate}[itemsep=1pt, topsep=1pt,leftmargin=1.5em]
\item[i)] \emph{Expressibility.}
We show that the assertion predicates supported by prior work can be expressed as projectors $P$, and hence checked by a checker unitary $C$ of the form in Eq.~\eqref{eq:checker}.
\item[ii)] \emph{Circuit realization.}
When prior work provides an explicit assertion-checking circuit, we show how it can be adapted to realize the corresponding checker unitary $C$ in Eq.~\eqref{eq:checker}.
\end{enumerate}
Together, (i) shows that our multi-assertion results apply to the same class of predicates considered by prior work, and (ii) shows that existing single-assertion circuit constructions can be reused to realize the checker unitaries our strategies operate on.

\subsection{Projection-Based Runtime Assertions}
\label{app:A:Proq}

\textsf{Proq}~\cite{ProjectionBased2020} introduces the assertion statement $\mathsf{assert}(\bar{q};P)$, where $P$ is a projector on the Hilbert space {\smash{$\mathcal{H}_{\bar q} \cong (\mathbb{C}^2)^{\otimes k}$}}, and defines its semantics via a mid-circuit projective measurement $M_P=\{P,\, I-P\}$ on the qubits $\bar{q}$: the program continues if the outcome corresponds to $P$, and aborts otherwise.

\paragraph{Equivalence to a checker unitary.}
We first show that this projective measurement is mathematically equivalent to applying a checker unitary (Eq.~\eqref{eq:checker}) and deferring the measurement of the ancilla to the end. The checker unitary for $\mathsf{assert}(\bar{q};P)$ is $C \triangleq P \otimes I + (I-P) \otimes X$, acting on $\mathcal{H}_{\bar q}$ together with a fresh ancilla $a$. For any pre-assertion state $\rho$ of the asserted qubits $\bar q$, with $a$ initialized to $\ket{0}$, measuring $a$ after applying $C$ yields
\[
\Pr[a=0]= \mathrm{Tr}(P\rho),\qquad \Pr[a=1]= \mathrm{Tr}((I-P)\rho) = 1-\mathrm{Tr}(P\rho),
\]
which matches exactly the outcome distribution of $M_P$. Moreover, conditioned on each outcome, the post-measurement state of $\bar q$ also coincides with that of $M_P$: outcome $a=0$ leaves $\bar q$ in $P\rho P / \mathrm{Tr}(P\rho)$, matching \textsf{Proq}'s continue branch, while $a=1$ leaves it in $(I-P)\rho(I-P)/\mathrm{Tr}((I-P)\rho)$, corresponding to the aborting branch.

\paragraph{Circuit realization.}
\textsf{Proq} also studies how to realize projection-based predicates on measurement-restricted devices, e.g., devices that only support computational-basis measurements or impose constraints on the rank of projectors~\cite[\S4]{ProjectionBased2020}. These techniques transform the checking of an assertion from a projective measurement into a computational-basis measurement on an integer number of qubits, with the transformation restoring the program state on the passing branch. After this transformation, the checking can be brought into checker-unitary form by adding a multi-controlled $X$ gate, controlled on the measured qubits and targeting an ancilla, and then deferring the measurement of that ancilla to the end.

\vspace{1ex}
In addition, see Sec.~\ref{sec:implication} for how our disjoint multi-round strategies can be recast into strategies under the mid-circuit measurement setting, where they apply to \textsf{Proq}~\cite{ProjectionBased2020} to trade additional ancillas for fewer measurements.

\subsection{Statistical Assertions}
\label{app:A:stat}

\citet{statisticalAssertions2019} propose \emph{statistical assertions}. Rather than checking a predicate on the quantum state directly, each assertion is defined over the \emph{distribution} of classical outcomes obtained by measuring the qubits under test, which is examined by a statistical hypothesis test.

\paragraph{Expressibility.}
Although statistical assertions are defined over measurement distributions rather than as predicates on the quantum state, projector predicates can express coherent counterparts of the intended predicates of their proposed checks:
\begin{itemize}[itemsep=1pt, topsep=1pt,leftmargin=1.5em]
\item \emph{Classical-value assertions.}
An assertion that a register deterministically measures to a classical value $t\in\{0,1\}^n$ corresponds to the rank-$1$ projector $P \triangleq \ketbra{t}{t}$.
\item \emph{Uniform Superposition assertions.}
Their superposition assertion checks that the measurement outcomes are consistent with the uniform distribution; this can be expressed as $P \triangleq \ketbra{+^{\otimes n}}{+^{\otimes n}}$.
\item \emph{Entanglement assertions.}
Their entanglement assertion checks whether two or more specified variables exhibit statistically significant association when measured. At the predicate level, one can model the intended correlations by a set $R \subseteq \{0,1\}^k$ of admissible joint outcomes, using the projector onto $\mathrm{span}\{\ket{x}\mid x\in R\}$, i.e., $P \triangleq \textstyle\sum_{x\in R} \ketbra{x}{x}$.
\end{itemize}

\paragraph{Realization.}
Statistical assertions are checked destructively: the program is instrumented with a \emph{breakpoint} at each assertion, run up to that point, and the qubits under test are measured directly, halting execution. Each assertion thus requires its own ensemble of runs, and different assertions are checked in separate ensembles. An operational benefit of testing through destructive measurement is that it requires no ancillas and no gate cost. In our framework, the same predicate-level checks can be given a non-destructive realization by replacing the breakpoint measurement with the corresponding checker unitary.

\subsection{Quantum Runtime Assertion Circuits}
\label{app:A:runtime}

\citet{RuntimeAssertionCircuts2020} propose runtime-assertion circuits for three specific assertion primitives: asserting classical values, asserting certain entangled states, and asserting arbitrary superposition states\@. \citet{ApproximateAssertion2021} generalize these to support (i) \emph{precise} assertions for a target \emph{pure} state, (ii) \emph{precise} assertions for a target \emph{mixed} state, and (iii) \emph{approximate} assertions formulated as membership checks over a set of candidate states. Since the constructions of \citet{RuntimeAssertionCircuts2020} are special cases of \citet{ApproximateAssertion2021}, we focus on the latter.

\paragraph{Expressibility.}
In terms of predicates, all assertions supported by~\citet{ApproximateAssertion2021} reduce to checking whether the program state lies in a pass subspace, and hence are captured by a projection-based predicate with an appropriate projector $P$:
\begin{itemize}[itemsep=1pt, topsep=1pt,leftmargin=1.5em]
\item \emph{Pure-state precise assertion.}
For a target pure state $\ket{\psi}$, the predicate is the projector $P \triangleq \ketbra{\psi}{\psi}$.

\item \emph{Mixed-state assertion and approximate assertion.}
\citet{ApproximateAssertion2021} treat both as \emph{support} membership checks. For a target mixed state $\rho$ with spectral decomposition $\rho = \sum_{j=0}^{t-1}\gamma_j \ketbra{\psi_j}{\psi_j}$ (where $\gamma_j > 0$), the assertion raises no error iff the state under test has support contained in $\mathrm{span}\{\ket{\psi_0},\ldots,\ket{\psi_{t-1}}\}$. For an approximate assertion specified by a set $\{\ket{\phi_1},\ldots,\ket{\phi_s}\}$, \citet{ApproximateAssertion2021} form the uniform mixture $\rho=\frac{1}{s}\sum_{k=1}^s \ketbra{\phi_k}{\phi_k}$ and apply the same technique. In both cases, the induced pass predicate is the projector onto $\mathrm{supp}(\rho)$, i.e., $P \triangleq \sum_{j=0}^{t-1}\ketbra{\psi_j}{\psi_j}$.
\end{itemize}

\paragraph{Circuit realization.}
\citet{ApproximateAssertion2021} present three circuit families for implementing assertions: logical OR-based circuits, NDD-based circuits, and SWAP-based circuits. Each family can be instantiated for pure/mixed-state predicates as well as membership checks:
\begin{enumerate}[itemsep=1pt, topsep=1pt,leftmargin=1.5em]
\item[(a)] \emph{Logical OR-based assertion circuits.}
The circuit design is shown in Figure~\ref{fig:OR-to-checker}(a). It already matches the semantics of a checker unitary: conditioned on the program state lying outside the asserted subspace, the circuit flips an ancilla. The only mismatch is the ancilla convention --- flipped on the \emph{pass} branch in~\citet{ApproximateAssertion2021}, whereas flipped on the \emph{fail} branch in our convention. Applying an additional $X$ gate on the ancilla aligns the two, as shown in Figure~\ref{fig:OR-to-checker}(b).
\begin{figure}[htbp]
  \centering
  \scalebox{0.85}{
  \begin{minipage}{\textwidth}
    \centering
    \begin{subfigure}[b]{0.45\textwidth}
      \centering
      \begin{quantikz}[row sep=0.25cm, column sep=0.5cm]
        \lstick[wires=2]{$\ket{\psi}$} & \gate[wires=2]{U^{-1}} & \octrl{2} & \gate[wires=2]{U} & \rstick[wires=2]{$\ket{\psi}$} \\
        & & \octrl{1} & & \\
        \lstick{$\ket{a}$} & \qw & \targ{} & \qw & \meter{}
      \end{quantikz}
      \caption{Logical-OR based assertion circuit}
      \label{fig:circuit1}
    \end{subfigure}
    \hfill
    \begin{subfigure}[b]{0.45\textwidth}
      \centering
      \begin{quantikz}[row sep=0.25cm, column sep=0.5cm]
        \lstick[wires=2]{$\ket{\psi}$} & \gate[wires=2]{U^{-1}} & \octrl{2} & \gate[wires=2]{U} & \rstick[wires=2]{$\ket{\psi}$} \\
        & & \octrl{1} & & \\
        \lstick{$\ket{a}$} & \gate{X} & \targ{} & \qw & \meter{}
      \end{quantikz}
      \caption{Modified circuit to fit the checker unitary.}
      \label{fig:circuit2}
    \end{subfigure}
  \end{minipage}}
  \caption{The logical-OR based assertion circuit construction (a) can be transformed by applying an additional $X$ gate on the ancilla (b), to fit the checker unitary.}
  \label{fig:OR-to-checker}
\end{figure}
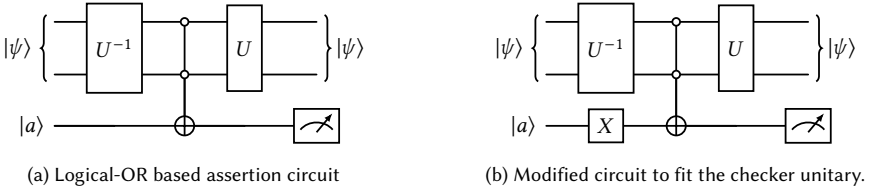

\item[(b)] \emph{NDD-based assertion circuits.}
NDD-based circuits (Fig.~\ref{fig:NDD}) implement a unitary $V$ defined via a spectral decomposition: $V$ assigns eigenvalue $+1$ to the pass subspace and $-1$ to its orthogonal complement. For a pure-state predicate with pass state $\ket{\psi_0}$ and an orthonormal basis $\{\ket{\psi_i}\}_{i\ge 1}$ of its orthogonal complement, \citet{ApproximateAssertion2021} define $V=\ketbra{\psi_0}{\psi_0}-\sum_{i\ge 1}\ketbra{\psi_i}{\psi_i}$. Letting $P\triangleq \ketbra{\psi_0}{\psi_0}$ and using $\sum_{i\ge 1}\ketbra{\psi_i}{\psi_i}=I-P$, we have $V = P-(I-P)=2P-I$. The NDD-based construction therefore realizes the checker unitary (Eq.~\eqref{eq:checker}) directly, by taking $P$ to be the projection-based predicate; the same identity extends to higher-rank projectors, e.g., mixed-state support predicates, by setting $P$ to the corresponding support projector.

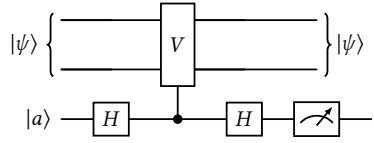
\begin{figure}[htbp]
  \centering
  \scalebox{0.85}{
\begin{quantikz}[row sep=0.25cm, column sep=0.5cm]
\lstick[wires=2]{$\ket{\psi}$} & \qw & \gate[wires=2]{V} & \qw  & \rstick[wires=2]{$\ket{\psi}$} \\
 & \qw & & \qw & \\
\lstick{$\ket{a}$} & \gate{H} & \ctrl{-2} & \gate{H} & \meter{} &
\end{quantikz}}
\vspace{-1ex}
\caption{The NDD-based assertion circuit, which already aligns with the checker unitary.}
\label{fig:NDD}
\end{figure}

\item[(c)] \emph{SWAP-based assertion circuits.}
The SWAP-based design is not intended to implement a check-only semantics. Operationally, it follows a \emph{repairing} semantics: it prepares ancillas in the target state and swaps them with the \textsf{prog} register, so that even upon an assertion failure the program state is overwritten (repaired) to the desired state. This differs from the checker unitary, which only routes the pass/fail outcome into an ancilla and does not modify the program state on failure. Compiling such a repairing circuit into our setting, by adding extra workspace, state preparation, and uncomputation, would be an interesting avenue to explore.
\end{enumerate}

%% file: appendix/stability.tex
\section{Semantic Stability Dichotomy}
\label{app:robust-proof}
In this appendix, we prove the strategy-independence claims for the semantics of $\textsc{ExistFail}_\eta$ and $\textsc{FirstFail}_\eta$ and the instability claim for $\textsc{ListAll}_\eta$, stated in Sec.~\ref{sec:prob-tasks}.

\begin{lemma}[Prefix Invariance under Arbitrary Instrumentation]
\label{thm:prefix-invariance-firstfail}
Assume the gap promise of Sec.~\ref{sec:prob-tasks}, and let $i^* \triangleq \emph{\textsc{FirstFail}}_\eta(\pazocal{Q}) \in [n]$.
Let $\pazocal{I}$ be any instrumentation with enabled set $\mathcal{E}$ (Def.~\ref{def:instr}), and for each $k \,{\in}\, [n]$ let $\sigma_k^{\pazocal{I}}$ denote the reduced state on the program register $\emph{\textsf{prog}}$ immediately before the $k$-th assertion position in the execution of $\pazocal{I}(\pazocal{Q})$.
Then for every $k \le i^*$, $\sigma_k^{\pazocal{I}} = |\varphi_k\rangle\langle\varphi_k|$, where $|\varphi_k\rangle$ is the bare-execution state checked by the $k$-th assertion (Def.~\ref{def:prog}).
\end{lemma}
\begin{proof}
We proceed by induction on $k$.
For base case, i.e, ($k=1$), the execution of $\pazocal{I}(\pazocal{Q})$ starts in the product state $|\varphi_0\rangle_{\textsf{prog}} \otimes |0^m\rangle_{\textsf{anc}}$, and the segment $U_0$ acts on \textsf{prog} alone, so $\sigma_1^{\pazocal{I}} = |\varphi_1\rangle\langle\varphi_1|$.
For the inductive step, assume $\sigma_k^{\pazocal{I}} = |\varphi_k\rangle\langle\varphi_k|$ for some $k < i^*$. The joint state is pure and its reduced state on \textsf{prog} is pure, so it factorizes as $|\varphi_k\rangle_{\textsf{prog}} \otimes |\xi\rangle_{\textsf{anc}}$ for some ancilla state $|\xi\rangle$.
Since $k < i^*$, by definition of $i^*$ and the gap promise we have $p_k = 0$, i.e., $\Vert(I - P_k)|\varphi_k\rangle\Vert = 0$, so $|\varphi_k\rangle$ lies entirely in the passing subspace $\mathrm{Im}(P_k)$. Now consider the effect of the $k$-th assertion position:
\begin{enumerate}[topsep=1pt,itemsep=1pt,leftmargin=1.5em]
\item[$a)$] If $k \,{\notin}\, \mathcal{E}$, the assertion is replaced by the identity, which does not change the program register.
\item[$b)$] If $k \,{\in}\, \mathcal{E}$, the assertion is replaced by the block $B_k$, a finite sequence of checker-unitary calls of $A_k$ and ancilla-only unitaries. Each ancilla-only unitary acts trivially on \textsf{prog} by definition. Each checker call meets the program state $|\varphi_k\rangle$ lying entirely in $\mathrm{Im}(P_k)$, so the fail branch of~\eqref{eq:checker-action} vanishes; applied to each computational-basis component of the designated ancilla and extended by linearity (as in the proof of Lemma~\ref{lem:ancilla-evolution}), the call leaves the program register in $|\varphi_k\rangle$, unentangled from the ancillas. Hence the entire block leaves the program register in $|\varphi_k\rangle$.
\end{enumerate}
The segment $U_k$ then acts on \textsf{prog} alone, mapping $|\varphi_k\rangle$ to $|\varphi_{k+1}\rangle$ exactly as in the bare execution, so $\sigma_{k+1}^{\pazocal{I}} = |\varphi_{k+1}\rangle\langle\varphi_{k+1}|$, completing the induction.
\end{proof}

\begin{theorem}[Strategy-Independent Semantics of \textsc{ExistFail} and \textsc{FirstFail}]
\label{thm:stable}
Assume the gap promise, and let $\pazocal{I}$ be any instrumentation with enabled set $\mathcal{E}$ (Def.~\ref{def:instr}). In the execution of $\pazocal{I}(\pazocal{Q})$, say that a checker call \emph{records a failure} with probability $p$ if its fail branch in~\eqref{eq:checker-action} carries weight $p$.
\begin{enumerate}[topsep=1pt,itemsep=1pt,leftmargin=1.5em]
\item If $~\emph{\textsc{FirstFail}}_\eta(\pazocal{Q}) = \bot$, i.e., every failure probability is $0$, then every checker call meets a program state satisfying its predicate with certainty: no execution can record a failure at any position.
\item If $~\emph{\textsc{FirstFail}}_\eta(\pazocal{Q}) = i^* \,{\in}\, [n]$, then no checker call at a position $j < i^*$ can record a failure; and if $i^* \,{\in}\, \mathcal{E}$, the first checker call at position $i^*$ meets the undisturbed bare-execution state $|\varphi_{i^*}\rangle$ and records a failure with probability exactly $p_{i^*} \ge \eta$. Any disturbance caused by this checker call can affect only subsequent
assertion positions $j>i^*$.
\end{enumerate}
The answers to $\emph{\textsc{ExistFail}}_\eta(\pazocal{Q})$ and $\emph{\textsc{FirstFail}}_\eta(\pazocal{Q})$ are defined on the bare execution and hence do not depend on the instrumentation; the two items state that no instrumented execution can observe anything inconsistent with them. In particular, for any strategy whose rounds jointly enable all assertions with $i^{*}$ checked in some round, then some round can record a failure iff $~\emph{\textsc{ExistFail}}_\eta(\pazocal{Q}) = 1$, and when the answer to $\emph{\textsc{FirstFail}}_\eta(\pazocal{Q})$ is $i^* \,{\in}\, [n]$, the least position at which a failure can be recorded is exactly $i^*$.
\end{theorem}

\begin{proof}
For item (1), every $p_k = 0$, so the induction in the proof of Lem.~\ref{thm:prefix-invariance-firstfail} applies to all indices $k \in [n]$: under any instrumentation, the program state before every assertion position equals the bare-execution state $|\varphi_k\rangle$, which lies entirely in $\mathrm{Im}(P_k)$ and remains there throughout the block. Every checker call therefore records a failure with probability $0$.

For item (2), by Lem.~\ref{thm:prefix-invariance-firstfail}, for every $k \le i^*$ the program state before the $k$-th assertion position equals $|\varphi_k\rangle$. For an enabled position $j < i^*$, we have $p_j = 0$, so $|\varphi_j\rangle$ lies entirely in $\mathrm{Im}(P_j)$ and remains there throughout the block; every checker call there records a failure with probability $0$. If $i^* \in \mathcal{E}$, the first checker call of the block at position $i^*$ acts on $|\varphi_{i^*}\rangle$ and records a failure with probability $\Vert(I - P_{i^*})|\varphi_{i^*}\rangle\Vert^2 = p_{i^*} \ge \eta$. Since the program register is undisturbed up to position $i^*$, any disturbance is confined to positions strictly after $i^*$.

The final claims follow: if $\textsc{ExistFail}_\eta(\pazocal{Q}) = 1$, any round enabling $i^*$ can record a failure there by item (2); if $\textsc{ExistFail}_\eta(\pazocal{Q}) = 0$, no round ever records one by item (1). Likewise, positions before $i^*$ can never record a failure in any round, while position $i^*$ can in any round enabling it, so the least recordable position is $i^*$.
\end{proof}

Now we show the instability of \textsc{ListAll}$_\eta$.

\begin{example}
\label{exp:instability}
Consider a correct program with 2 assertions, on a single qubit $x$ initialized to $|0\rangle$:
\[
\pazocal{Q}_0 \triangleq  (\textsf{prog}; \ket{0};  U_0 = I, A_1 = \mathrm{assert}(x;\, |0\rangle\langle 0|), U_1 = H, A_2 = \mathrm{assert}(x;\, |{+}\rangle\langle{+}|); \pazocal{M}).
\]
For this correct program, both two assertions pass with certainty because the bare-execution states $|\varphi_1\rangle = |0\rangle$ and $|\varphi_2\rangle = |{+}\rangle$, and $\textsc{ListAll}_\eta(\pazocal{Q}_0) = \varnothing$.
Now consider two buggy variants \emph{$\pazocal{Q}_1$} and \emph{$\pazocal{Q}_2$}:
\begin{itemize}[leftmargin=15pt,itemsep=2pt]
    \item $\pazocal{Q}_1$ (\emph{Misplaced Gates}): the two gates $I$ and $H$ are swapped, i.e., $U_0 = H$, $U_1 = I$. Formally,
    \[
    \pazocal{Q}_1 \triangleq  (\textsf{prog}; \ket{0};  U_0 = H, A_1 = \mathrm{assert}(x;\, |0\rangle\langle 0|), U_1 = I, A_2 = \mathrm{assert}(x;\, |{+}\rangle\langle{+}|); \pazocal{M}).
    \]
    Its bare program evolves as \smash{$\ket{0} \rightarrow_{H} \ket{+} \rightarrow_{I} \ket{+}$}, hence the bare-execution states  $|\varphi_1\rangle = |\varphi_2\rangle = |{+}\rangle$. So $p_1 = \tfrac12$ and $p_2 = 0$; with $\eta = \tfrac12$ the gap promise holds and $\textsc{ListAll}_\eta(\pazocal{Q}_1) = \{1\}$.

    \item $\pazocal{Q}_2$ (\emph{Additional Phase Flip}):
    the identity segment $U_1$ of $\pazocal{Q}_1$ is replaced by a phase flip. Formally,
    \[
    \pazocal{Q}_2 \triangleq  (\textsf{prog}; \ket{0};  U_0 = H, A_1 = \mathrm{assert}(x;\, |0\rangle\langle 0|), U_1 = Z, A_2 = \mathrm{assert}(x;\, |{+}\rangle\langle{+}|); \pazocal{M}).
    \]
    Its bare program evolves as \smash{$\ket{0} \rightarrow_{H} \ket{+} \rightarrow_{Z} \ket{-}$}, hence the bare-execution states  $|\varphi_1\rangle = |{+}\rangle $ and $|\varphi_2\rangle = |{-}\rangle$. So $p_1 = \tfrac12$ and $p_2 = 1$. Therefore $\textsc{ListAll}_\eta(\pazocal{Q}_2) = \{1, 2\}$.
\end{itemize}

We use the two buggy variants to prove the following instability theorem.
\end{example}

\begin{theorem}[Semantic instability of \textsc{ListAll}$_\eta$]
\label{thm:instable}
For some $\eta \,{>}\, 0 $, there exist two programs-with-assertions satisfying the gap promise whose $\emph{\textsc{ListAll}}_\eta$ answers differ, yet whose ancilla readouts under a single-round/$n$-ancillas instrumentation are identically distributed, leaving the answers irrecoverable; the single-ancilla/$n$-rounds strategy reports the reference answer of each under sampling.
\end{theorem}

\begin{proof}
    The witnesses are $\pazocal{Q}_1$ and $\pazocal{Q}_2$ of Example~\ref{exp:instability}, with $\eta \,{=}\, \tfrac12$.
    \begin{enumerate}[leftmargin=15pt, itemsep=3pt]
    \item[(i)] Applying a \emph{single-round/$n$-ancillas} instrumentation: 

    \vspace{0.5ex}
    On $\pazocal{Q}_1$, the gate $U_0 = H$ prepares $|{+}\rangle_x$, and the checker call $C_{1\to a_1}$ flips $a_1$ on the $|1\rangle_x$ component, yielding \smash{$\tfrac{1}{\sqrt2}\bigl(|0\rangle_x|0\rangle_{a_1} + |1\rangle_x|1\rangle_{a_1}\bigr) \otimes |0\rangle_{a_2}$}. The gate $U_1 \,{=}\, I$ acts trivially, and $C_{2\to a_2}$ flips $a_2$ on the $|{-}\rangle_x$ component; the final joint state is
     \[
        \tfrac12\bigl(\,
              |{+}\rangle_x|0\rangle_{a_1}|0\rangle_{a_2}
            + |{-}\rangle_x|0\rangle_{a_1}|1\rangle_{a_2}
            + |{+}\rangle_x|1\rangle_{a_1}|0\rangle_{a_2}
            - |{-}\rangle_x|1\rangle_{a_1}|1\rangle_{a_2}
            \,\bigr).
    \]
    Every ancilla transcript $(a_1, a_2) \,{\in}\, \{0,1\}^2$ carries squared amplitude $\tfrac14$.
    
    \vspace{1ex}
    On $\pazocal{Q}_2$, the execution proceeds identically through $C_{1\to a_1}$, and $U_1 = Z$ negates the $\ket{1}_x$ component, giving \smash{$\tfrac{1}{\sqrt2}\bigl(|0\rangle_x |0\rangle_{a_1} - |1\rangle_x |1\rangle_{a_1}\bigr)\otimes |0\rangle_{a_2}$}; after $C_{2\to a_2}$, the final joint state is
    \[
    \tfrac12\bigl(\,
      |{+}\rangle_x|0\rangle_{a_1}|0\rangle_{a_2}
    + |{-}\rangle_x|0\rangle_{a_1}|1\rangle_{a_2}
    - |{+}\rangle_x|1\rangle_{a_1}|0\rangle_{a_2}
    + |{-}\rangle_x|1\rangle_{a_1}|1\rangle_{a_2}
    \,\bigr),
    \]
    so every ancilla transcript again carries squared amplitude \smash{$\tfrac14$}. 

    \vspace{1ex}
    The ancilla readouts on $\pazocal{Q}_1$ and $\pazocal{Q}_2$ are thus identically distributed, as are the joint readouts of any number of independent repetitions; any decoder therefore has the same output distribution on both programs and errs on at least one of them, so no decoder recovers both answers.

    \item[(ii)] Applying the \emph{single-ancilla/$n$-rounds} strategy: 

    \vspace{0.5ex}
    in the round enabling $A_i$ alone, the other assertion is replaced by the identity, so the program reaches position $i$ in its bare state and the readout equals $1$ with probability exactly $p_i$. That is, the two rounds read out $1$ with probabilities $\tfrac12$ and $0$ on $\pazocal{Q}_1$, and with probabilities $\tfrac12$ and $1$ on $\pazocal{Q}_2$. Declaring $A_i$ failing iff its readout equals $1$ in at least one of $k$ repetitions of its round reports $\{1\}$ on $\pazocal{Q}_1$ and $\{1,2\}$ on $\pazocal{Q}_2$, each with probability at least $1 - 2^{-k}$.
    \end{enumerate}
    \vspace{-2ex}
\end{proof}

%% file: appendix/exist-lb.tex
\section{Proof of Lemma~\ref{thm:unitary-dim-lb}}
\label{app:unitary-dim-lb-proof}

Lemma~\ref{thm:unitary-dim-lb}.
\textit{
Let $\pazocal{U}_n=(\pazocal{H},\{\pazocal{T}^{\scriptscriptstyle (i)}_0\}_{i=1}^n,\{\pazocal{T}^{\scriptscriptstyle (i)}_1\}_{i=1}^n,|\psi_0\rangle)$ be a one-round finite-dimensional unitary transition system as in Def.~\ref{def:unitary-transition-system}.
Assuming that $\forall x\neq 0^n$, $\langle \psi(0^n) \,|\, \psi(x)\rangle = 0$, we have $\dim(\pazocal{H})\ge n+1$.
}

\begin{proof}
We first prove a pumping-style property (Lem.~\ref{lem:pumping}), then use it for a counting argument.

\paragraph{Abbreviations.}
Let $|z\rangle\triangleq |\psi(0^n)\rangle$.
For each $k\in\{0,\ldots,n\}$, define the intermediate state along the all-zero-string transitions $|s_k\rangle \triangleq |\psi(0^k)\rangle
~=~
(\pazocal{T}^{\scriptscriptstyle (k)}_{0}\cdots \pazocal{T}^{\scriptscriptstyle (2)}_{0}\pazocal{T}^{\scriptscriptstyle (1)}_{0})|\psi_0\rangle$,
with $|s_0\rangle=|\psi_0\rangle$ and $|s_n\rangle=|z\rangle$.
For $0\le i<n$ and a bit-string $u=u_1\cdots u_\ell\in\{0,1\}^\ell$ with
$1\le \ell\le n-i$, define the segment unitary
\[
\pazocal{T}^{[i]}_{u}
~\triangleq~
\pazocal{T}^{\scriptscriptstyle (i+\ell)}_{u_\ell}\cdots \pazocal{T}^{\scriptscriptstyle (i+2)}_{u_2}\pazocal{T}^{\scriptscriptstyle (i+1)}_{u_1}.
\]
By the definition of $|s_k\rangle$ we have $|s_{i+\ell}\rangle = \pazocal{T}^{[i]}_{0^\ell}|s_i\rangle$. For the empty string $\epsilon$, define $\pazocal{T}^{[i]}_{\epsilon} \triangleq I$.
Intuitively, \smash{$\pazocal{T}^{[i]}_{u}$} is the step-$i$ to step-$(i+\ell)$ transition induced by reading the length-$\ell$ segment $u$: for any intermediate state $|\phi\rangle$ that may arise after an $i$-bit prefix, applying the next $\ell$ input bits $u$ maps it to the state $\pazocal{T}^{[i]}_{u}|\phi\rangle$.

\begin{lemma}
\label{lem:pumping}
Assume that $\forall x\neq 0^n$, $\langle z \,|\, \psi(x)\rangle=0$. Then for all integers $i,\ell$ with $0\le i<n$ and $1\le \ell\le n-i$, and for all $u\in\{0,1\}^\ell$, we have
\[
\quad u\neq 0^\ell
\quad \text{ implies } \quad
\braketU{s_{i+\ell}}{\pazocal{T}^{[i]}_{u}}{s_i} = 0.
\]
\end{lemma}
\begin{proof}
We argue by contradiction. Suppose there exist integers $0\leq i <n $, $1 \leq \ell \leq n-i$, and a string
$u\in\{0,1\}^\ell$ with $u\neq 0^\ell$ but $\braketU{s_{i+\ell}}{\pazocal{T}^{\scriptscriptstyle [i]}_{u}}{s_i} \neq 0$.
Consider the length-$n$ input
\[
\widebar{x}  =  0^{\,i}\,u\,0^{\,n-i-\ell}.
\]
By expanding the computation at the cut after reading $0^{\,i}$, we get
\[
|\psi(\widebar{x})\rangle
~=~
\pazocal{T}^{[i+\ell]}_{0^{\,n-i-\ell}}  \pazocal{T}^{[i]}_{u}|s_i\rangle.
\]
Now use the equation $|z\rangle=|s_n\rangle=\pazocal{T}^{[i+\ell]}_{0^{\,n-i-\ell}}|s_{i+\ell}\rangle$.
Because $\pazocal{T}^{[i+\ell]}_{0^{\,n-i-\ell}}$ is unitary, we have
\[
\langle z|\,\pazocal{T}^{[i+\ell]}_{0^{\,n-i-\ell}}
~=~
\langle s_{i+\ell}|.
\]
Therefore, because $\widebar{x} \neq 0^n$, we have
\[
0 = \langle z\,|\,\psi(\widebar{x}) \rangle
~=~
\langle z|\,\pazocal{T}^{[i+\ell]}_{0^{\,n-i-\ell}}\,\pazocal{T}^{[i]}_{u}|s_i\rangle
~=~
\braketU{s_{i+\ell}}{\pazocal{T}^{[i]}_{u}}{s_i}
\neq 0,
\]
which is a contradiction.
\end{proof}

We now use Lemma~\ref{lem:pumping} to derive a counting argument that completes the proof.
For each $k\in\{0,\ldots,n\}$, consider the input string
\[
x_k  =  0^{\,n-k}1^k,
\]
and define the corresponding final state $|t_k\rangle\triangleq |\psi(x_k)\rangle$.
We claim that the states $|t_0\rangle,|t_1\rangle,\ldots,|t_n\rangle$ are pairwise orthogonal.
Because all $|t_k\rangle$ are unit vectors (unitaries preserve norm), pairwise orthogonality
implies linear independence, yielding $\dim(\pazocal{H})\ge n+1$.

\begin{enumerate}[label=(\roman*),leftmargin=20pt]
\item For $k=0$, $x_0=0^n$, hence $|t_0\rangle=|\psi(0^n)\rangle=|z\rangle$.
For every $k\ge 1$, $x_k\neq 0^n$, so by assumption
\[
\langle t_0 \,|\, t_k\rangle
~=~
\langle z \,|\, \psi(x_k)\rangle
~=~0.
\]
Thus $|t_0\rangle$ is orthogonal to all $|t_k\rangle$ with $k\ge 1$.

\item It remains to show that $\langle t_i \,|\, t_j\rangle=0$ for all $1\le i<j\le n$.
Fixing $i,j$, we write each $|t_k\rangle$ by cutting right before the final $k$ input bits of one:
\[
|t_i\rangle
~=~
\pazocal{T}^{[n-i]}_{1^i}|s_{n-i}\rangle,
\quad
|t_j\rangle
~=~
\pazocal{T}^{[n-j]}_{1^j} |s_{n-j}\rangle.
\]
Crucially, we can factor the common \emph{suffix}
unitary corresponding to the last $i$ ones:
\[
\pazocal{T}^{[n-j]}_{1^j}
~=~
\pazocal{T}^{[n-i]}_{1^i}\circ \pazocal{T}^{[n-j]}_{1^{j-i}}.
\]
Therefore,
\begin{align*}
\langle t_i \,|\, t_j\rangle
&=
\langle s_{n-i}|\bigl(\pazocal{T}^{[n-i]}_{1^i}\bigr)^\dagger
\Bigl(\pazocal{T}^{[n-j]}_{1^j}|s_{n-j}\rangle\Bigr) \\
&=
\langle s_{n-i}|\bigl(\pazocal{T}^{[n-i]}_{1^i}\bigr)^\dagger
\Bigl(\pazocal{T}^{[n-i]}_{1^i} \, \pazocal{T}^{[n-j]}_{1^{j-i}}|s_{n-j}\rangle\Bigr) \\
&=
\braketU{s_{n-i}}{\pazocal{T}^{[n-j]}_{1^{j-i}}}{s_{n-j}}
\quad
(\text{because } (\pazocal{T}^{[n-i]}_{1^i})^\dagger \pazocal{T}^{[n-i]}_{1^i}=I).
\end{align*}
Observe that the final line is in the form of Lemma~\ref{lem:pumping}
with $i' \,{=}\, n\,{-}\,j$, $\ell'\,{=}\,j\,{-}\,i$, and $u\,{=}\,1^{j{-}i}$:
\[
\langle t_i \,|\, t_j\rangle
=
\braketU{s_{i'+\ell'}}{\pazocal{T}^{[i']}_{u}}{s_{i'}}.
\]
Because $u=1^{j-i}\neq 0^{j-i}$, Lemma~\ref{lem:pumping} implies
$\langle t_i \,|\, t_j\rangle=0$.
\end{enumerate}

We have thus exhibited $n+1$ pairwise orthogonal unit vectors
$|t_0\rangle,|t_1\rangle,\ldots,|t_n\rangle$ in $\pazocal{H}$,
and therefore $\dim(\pazocal{H})\ge n+1$.
\end{proof}

%% file: appendix/alternative.tex
\section{An Alternative Upper Bound Strategy for \textsc{FirstFail}}
\label{app:firstfail-alternative}

Here we provide an alternative construction to witness the matching upper bound of \textsc{FirstFail} in a single round. Different from the \textsf{index transposition} strategy (Strategy~\ref{str:single-round-index-transposition}) that uses index-dependent ancilla updates for each index $i$, it gives an index-independent ancilla update, provided that the ancilla register is updated not only on seeing assertion failures but also on passes.

\begin{strategy}[\textsc{\emph{Single-Round Modulo Decrement}}]
\label{str:single-round-modulo-decrement}
\vspace{-0.5ex}
Let $\ell \,{\triangleq}\, \lceil \log(n \,{+}\, 1)\rceil$.
The instrumentation $\pazocal{I}$ uses $\ell \,{+}\, 1$ ancillas, partitioned into an $\ell$-qubit index register $\textsf{idx}$ and a one-qubit flag $\textsf{fail}$, both initialized to $\ket{0}$.
For each assertion $A_i$, the assertion-handling block used by $\pazocal{I}$ is:
{
\setlength{\abovedisplayskip}{3pt}
\setlength{\belowdisplayskip}{3pt}
\setlength{\abovedisplayshortskip}{2pt}
\setlength{\belowdisplayshortskip}{2pt}
\begin{align*}
B_i
=
C_{i \to \textsf{fail}};
\ V;
\ C_{i \to \textsf{fail}},
\end{align*}}

\noindent
where $V$ is fixed and acts on $\textsf{idx}$ as follows.
If $\textsf{fail} \,{=}\, \ket{1}$, it applies $V_{f=1}$, which rotates the basis states $|0^\ell\rangle, |\mathrm{bin}(n)\rangle, |\mathrm{bin}(n{-}1)\rangle, \ldots, |\mathrm{bin}(1)\rangle$ by one step cyclically, acting as an arbitrary fixed unitary on the remaining basis states.
If $\textsf{fail} \,{=}\, \ket{0}$, it applies $V_{f=0}$, which fixes $|0^\ell\rangle$ and rotates $|\mathrm{bin}(1)\rangle, |\mathrm{bin}(n)\rangle, |\mathrm{bin}(n{-}1)\rangle, \ldots, |\mathrm{bin}(2)\rangle$ by one step.
Equivalently, viewed as permutations:
{\setlength{\abovedisplayskip}{3pt}
\setlength{\belowdisplayskip}{3pt}
\setlength{\abovedisplayshortskip}{2pt}
\setlength{\belowdisplayshortskip}{2pt}
\begin{align*}
V_{f=0}: (0)\,(1,n,n{-}1,\ldots,2)
\quad \text{and} \quad
V_{f=1}: (0,n,n{-}1,\ldots,1).
\end{align*}
}

\noindent
Intuitively, the first failure writes the value $n$ into $\textsf{idx}$, and every subsequent assertion, whether pass or fail, decrements this nonzero value by one.
Hence, after processing all assertions, $\textsf{idx} \,{=}\, |\mathrm{bin}(i^*)\rangle$, where $i^*$ is the first failing index, or $|0^\ell\rangle$ if none fail.
Measuring $\textsf{idx}$ therefore solves \textsc{FirstFail}.
\end{strategy}

\begin{example}
\vspace{-0.5ex}
Let $n=4$ and suppose the failure pattern is $F=(0,1,0,1)$, so the first failure is at $2$.
Then $\ell=\lceil \log 5\rceil=3$, and the ancilla register evolves as:

\vspace{-2.5ex}
{
\small
\setlength{\fboxsep}{0.5pt}
\setlength{\fboxrule}{0pt}
\[
\begin{aligned}
\ket{000}_{\textsf{idx}}\ket{0}_{\textsf{fail}}
&\xrightarrow{\cdots}
&& \fcolorbox{gray!100}{gray!7.5}{$\displaystyle
\xrightarrow[\text{write \textsf{fail}}]{C_{1 \to \textsf{fail}}}\ket{000}\ket{0}
\xrightarrow[\text{apply }V_{f=0}:~0\mapsto 0]{V} \ket{000}\ket{0}
\xrightarrow[\text{uncompute \textsf{fail}}]{C_{1 \to \textsf{fail}}}\ket{000}\ket{0}
~$} && A_1~\text{pass} \\[0ex]
&\xrightarrow{\cdots}
&& \fcolorbox{gray!100}{gray!7.5}{$\displaystyle
\xrightarrow[\text{write \textsf{fail}}]{C_{2 \to \textsf{fail}}}\ket{000}\ket{1}
\xrightarrow[\text{apply }V_{f=1}:~0\mapsto 4]{V} \ket{100}\ket{1}
\xrightarrow[\text{uncompute \textsf{fail}}]{C_{2 \to \textsf{fail}}}\ket{100}\ket{0}
~$} && A_2~\text{fail} \\[0ex]
&\xrightarrow{\cdots}
&& \fcolorbox{gray!100}{gray!7.5}{$\displaystyle
\xrightarrow[\text{write \textsf{fail}}]{C_{3 \to \textsf{fail}}}\ket{100}\ket{0}
\xrightarrow[\text{apply }V_{f=0}:~4\mapsto 3]{V} \ket{011}\ket{0}
\xrightarrow[\text{uncompute \textsf{fail}}]{C_{3 \to \textsf{fail}}}\ket{011}\ket{0}
~$} && A_3~\text{pass} \\[0ex]
&\xrightarrow{\cdots}
&& \fcolorbox{gray!100}{gray!7.5}{$\displaystyle
\xrightarrow[\text{write \textsf{fail}}]{C_{4 \to \textsf{fail}}}\ket{011}\ket{1}
\xrightarrow[\text{apply }V_{f=1}:~3\mapsto 2]{V} \ket{010}\ket{1}
\xrightarrow[\text{uncompute \textsf{fail}}]{C_{4 \to \textsf{fail}}}\ket{010}\ket{0}
~$} && A_4~\text{fail} \\[0ex]
&\xrightarrow{\cdots}
&& \text{execution ends} \xrightarrow{\text{measure \textsf{idx}}} 010 \xrightarrow{\textsf{Dec}} 2
\end{aligned}
\]
}

\vspace{-0.5ex}
The final index register is $|\mathrm{bin}(2)\rangle \,{=}\, \ket{010}$, and measuring $\textsf{idx}$ yields the smallest failing index.
\end{example}

\paragraph{Cost of Strategy~\ref{str:single-round-modulo-decrement}}
\vspace{-0.5ex}
The strategy has $T \,{=}\, 1$ and $S \,{=}\, \lceil \log(n \,{+}\, 1)\rceil \,{+}\, 1$.
The checker unitary is invoked twice per assertion, so $C \,{=}\, 2$, and only $\textsf{idx}$ is measured at the end, hence $M \,{=}\, \lceil \log(n \,{+}\, 1)\rceil$.
For gate cost, each assertion applies one ancilla-processing unitary $V$.
Both branches can be realized with $O(\ell)$ elementary gates, via the staircase decrement and Gray-path transpositions of Strategy~\ref{str:single-round-index-transposition}.
Thus each assertion costs ${O}(\ell)$ gates, yielding total non-checker gate cost $G={O}(n\log n)$.

%% file: appendix/first-lb.tex
\section{Proof of Lemma~\ref{lem:first-packing}}
\label{app:first-packing-proof}
 
Throughout this appendix, fix a $T$-round unitary transition system $(\pazocal{U}^{1}_n, \ldots, \pazocal{U}^{T}_n)$ that solves \textsc{First} with a classical decoder $\textsf{Dec}$ (Def.~\ref{def:unitary-transition-system}, Def.~\ref{def:distinguish}), and let $\pazocal{H}_t$, $|\psi^{t}_0\rangle$, and $d_t \triangleq \dim(\pazocal{H}_t)$ denote the Hilbert space, initial state, and dimension of $\pazocal{U}^{t}_n$.
For a cut position $i$ and a string $u = u_1 \cdots u_\ell \in \{0,1\}^{\ell}$ with $\ell \le n - i$, let \smash{$\pazocal{T}^{ t,[i]}_{u}$} denote the round-$t$ unitary segment induced by reading $u$ from step $i+1$ to $i+\ell$. For a full input $x \,{\in}\, \{0,1\}^n$, we abbreviate its \emph{input operator} as \smash{$\pazocal{T}^{t}_{x} \triangleq \pazocal{T}^{t,[0]}_{x}$}.
Each segment is unitary, segments compose as {$\pazocal{T}^{t,[0]}_{wu} = \pazocal{T}^{t,[|w|]}_{u}\, \pazocal{T}^{t,[0]}_{w}$}, and the round-$t$ final state on $x$ is $|\psi^{t}(x)\rangle = \pazocal{T}^{t}_{x}|\psi^{t}_0\rangle$.
 
\begin{lemma}[Pairwise Separation]
\label{lem:first-separation}
For every pair of distinct inputs $x, y \in \{0,1\}^n$,
{\setlength{\abovedisplayskip}{3pt}
\setlength{\belowdisplayskip}{2pt}
\setlength{\abovedisplayshortskip}{0pt}
\setlength{\belowdisplayshortskip}{0pt}
\begin{align*}
\max_{t \in [T]} \big\Vert \pazocal{T}^{t}_{x} - \pazocal{T}^{t}_{y} \big\Vert_{\mathrm{op}} \ \ge\ \sqrt{2}, \quad \text{where $\Vert\cdot\Vert_{\mathrm{op}}$ denotes the operator norm.}
\end{align*}
}
\end{lemma}

\begin{proof}
Fix distinct $x, y \in \{0,1\}^n$, let $i \triangleq \min\{\,j \in [n] \mid x_j \neq y_j\,\}$ be their first differing position, and assume $x_i = 0$ and $y_i = 1$ (exchanging $x$ and $y$ if necessary).
Write $w \triangleq x_1 \cdots x_{i-1} = y_1 \cdots y_{i-1}$ for the common prefix, and define, for each round $t$, the witness vector
{\setlength{\abovedisplayskip}{3pt}
\setlength{\belowdisplayskip}{2pt}
\setlength{\abovedisplayshortskip}{0pt}
\setlength{\belowdisplayshortskip}{0pt}
\begin{align*}
|v_t\rangle \ \triangleq\ \big(\pazocal{T}^{t,[0]}_{w}\big)^{\dagger}\, \pazocal{T}^{t,[0]}_{0^{i-1}}\, |\psi^{t}_0\rangle,
\end{align*}
}
 
\noindent
a unit vector satisfying $\pazocal{T}^{t,[0]}_{w}|v_t\rangle = \pazocal{T}^{t,[0]}_{0^{i-1}}|\psi^{t}_0\rangle$; when $i = 1$, both segments are the empty product and $|v_t\rangle = |\psi^{t}_0\rangle$.
Using the composition of segments and $x_i = 0$,
{\setlength{\abovedisplayskip}{3pt}
\setlength{\belowdisplayskip}{2pt}
\setlength{\abovedisplayshortskip}{0pt}
\setlength{\belowdisplayshortskip}{0pt}
\begin{align*}
\pazocal{T}^{t}_{x}|v_t\rangle
~=~ \pazocal{T}^{t,[i-1]}_{x_i \cdots x_n}\, \pazocal{T}^{t,[0]}_{w}\, |v_t\rangle
~=~ \pazocal{T}^{t,[i-1]}_{0\, x_{i+1} \cdots x_n}\, \pazocal{T}^{t,[0]}_{0^{i-1}}\, |\psi^{t}_0\rangle
~=~ |\psi^{t}(\tilde{x})\rangle,
\end{align*}
}
 
\noindent
where $\tilde{x} \triangleq 0^{i} x_{i+1} \cdots x_n$; similarly, using $y_i = 1$, we get $\pazocal{T}^{t}_{y}|v_t\rangle = |\psi^{t}(\tilde{y})\rangle$ with $\tilde{y} \triangleq 0^{i-1} 1\, y_{i+1} \cdots y_n$.

\vspace{0.5ex}
The two modified inputs have different answers: $\textsc{First}(\tilde{y}) = i$, whereas $\textsc{First}(\tilde{x}) \in \{\bot\} \cup \{i+1, \ldots, n\}$.
By Lemma~\ref{lem:answer-separation}, there is a round $t$ in which
$|\psi^t(\tilde{x})\rangle$ and $|\psi^t(\tilde{y})\rangle$ have disjoint basis supports; in particular $\langle \psi^t(\tilde{x})\,|\,\psi^t(\tilde{y})\rangle = 0$, i.e., $\langle \pazocal{T}^t_x v_t \,|\, \pazocal{T}^t_y v_t\rangle = 0$.
Both vectors are unit vectors, since the segments are unitary and $|v_t\rangle$ is a unit vector; hence
{\setlength{\abovedisplayskip}{2pt}
\setlength{\belowdisplayskip}{3pt}
\setlength{\abovedisplayshortskip}{0pt}
\setlength{\belowdisplayshortskip}{0pt}
\begin{align*}
\big\Vert\big(\pazocal{T}^{t}_{x} - \pazocal{T}^{t}_{y}\big)|v_t\rangle\big\Vert^2
~=~ 1 + 1 - 2\,\mathrm{Re}\,\langle \pazocal{T}^{t}_{x} v_t \,|\, \pazocal{T}^{t}_{y} v_t\rangle
~=~ 2,
\end{align*}
}
 
\noindent
and therefore $\Vert\pazocal{T}^{t}_{x} - \pazocal{T}^{t}_{y}\Vert_{\mathrm{op}} \ge \big\Vert\big(\pazocal{T}^{t}_{x} - \pazocal{T}^{t}_{y}\big)|v_t\rangle\big\Vert = \sqrt{2}$.
\end{proof}
 
\begin{lemma}[Volumetric Packing]
\label{lem:volumetric-packing}
Let $(V, \Vert\cdot\Vert)$ be an $m$-dimensional real normed vector space, and let $z_1, \ldots, z_N \in V$ satisfy $\Vert z_i\Vert \le R$ for every $i$ and $\Vert z_i - z_j\Vert \ge \delta$ for all $i \neq j$. Then $N \le (1 + 2R/\delta)^m$.
\end{lemma}
 
\begin{proof}
Let $r \triangleq \delta/2$ and consider the open balls $B(z_i, r) \triangleq \{z \,{\in}\, V : \Vert z - z_i\Vert < r\}$.
They are pairwise disjoint: a common point $z$ of $B(z_i, r)$ and $B(z_j, r)$ would give $\Vert z_i - z_j\Vert \le \Vert z_i - z\Vert + \Vert z - z_j\Vert < 2r = \delta$, contradicting the separation.
Moreover, every $B(z_i, r)$ is contained in $B(0, R + r)$, since any of its points satisfies $\Vert z\Vert \le \Vert z - z_i\Vert + \Vert z_i\Vert < r + R$.
Fix any linear identification $V \cong \mathbb{R}^m$ and use the standard volume in these coordinates; the choice of identification scales all volumes by a common positive constant, which cancels below.
Balls scale as $\mathrm{Vol}(B(0, a)) = a^m \, \mathrm{Vol}(B(0,1))$, so disjointness and containment give $N \cdot r^m \le (R + r)^m$, i.e., $N \le (1 + 2R/\delta)^m$.
\end{proof}
 
\begin{proof}[Proof of Lemma~\ref{lem:first-packing}]
Consider the real vector space $V \triangleq \bigoplus_{t=1}^{T} M_{d_t}(\mathbb{C})$, where $M_{d_t}(\mathbb{C})$ is the space of $d_t \times d_t$ complex matrices regarded as a real vector space, so that $m \triangleq \dim_{\mathbb{R}} V = 2\sum_{t=1}^{T} d_t^{\,2}$.
Equip $V$ with the norm $\Vert(X_1, \ldots, X_T)\Vert \triangleq \max_{t \in [T]} \Vert X_t\Vert_{\mathrm{op}}$, and associate with each input $x \in \{0,1\}^n$ the tuple of its input operators
{\setlength{\abovedisplayskip}{0pt}
\setlength{\belowdisplayskip}{2pt}
\setlength{\abovedisplayshortskip}{0pt}
\setlength{\belowdisplayshortskip}{0pt}
\begin{align*}
Z(x) \ \triangleq\ \big(\pazocal{T}^{1}_{x}, \ldots, \pazocal{T}^{T}_{x}\big) \in V.
\end{align*}
}
 
\noindent
Every component is unitary, so $\Vert Z(x)\Vert = 1$ and all $2^n$ tuples lie in the unit ball of $V$; by Lemma~\ref{lem:first-separation}, they are pairwise $\sqrt{2}$-separated.
Applying Lemma~\ref{lem:volumetric-packing} with $N = 2^n$, $R = 1$, and $\delta = \sqrt{2}$ gives
{\setlength{\abovedisplayskip}{3pt}
\setlength{\belowdisplayskip}{2pt}
\setlength{\abovedisplayshortskip}{0pt}
\setlength{\belowdisplayshortskip}{0pt}
\begin{align*}
2^n \ \le\ \big(1 + \tfrac{2}{\sqrt{2}}\big)^{m} \ =\ (1+\sqrt{2})^{\,2\sum_{t=1}^{T} d_t^{2}}.
\end{align*}
}
 
\noindent
Taking logarithms yields $n \le 2\log(1+\sqrt{2}) \cdot \sum_{t=1}^{T} d_t^{\,2}$, which rearranges to the claimed bound.
\end{proof}